\documentclass[]{amsart}

\usepackage[letterpaper, hmarginratio=1:1]{geometry}

\usepackage{amsmath, amsthm, amsfonts, thmtools, amssymb,tikz,hyperref,cleveref,bm,stmaryrd,bbm}
\usetikzlibrary{arrows}

\usepackage[mathscr]{euscript}

\numberwithin{equation}{section}

\DeclareMathOperator{\eig}{eig}
\DeclareMathOperator{\Tr}{Tr}
\DeclareMathOperator{\Imaginary}{Im}
\DeclareMathOperator{\Real}{Re}

\DeclareMathOperator{\eff}{eff}
\DeclareMathOperator{\Id}{Id}

\DeclareMathOperator{\Mat}{Mat}
\DeclareMathOperator{\SymMat}{SymMat}
\DeclareMathOperator{\dist}{dist}

\DeclareMathOperator{\sgn}{sgn}

\newtheorem{thm}{Theorem}[section]
\newtheorem{prop}[thm]{Proposition}
\newtheorem{lem}[thm]{Lemma}
\newtheorem{cor}[thm]{Corollary}

\theoremstyle{remark}
\newtheorem{rem}[thm]{Remark}
\theoremstyle{definition}
\newtheorem{definition}[thm]{Definition}

\newtheorem{assumption}[thm]{Assumption}

\makeatletter
\@addtoreset{equation}{section}

\makeatother



\title{Asymptotic Scattering Relation for the Toda Lattice}
\author{Amol Aggarwal}

\begin{document}

	\begin{abstract}
		
		In this paper we consider the Toda lattice $(\bm{p}(t); \bm{q}(t))$ at thermal equilibrium, meaning that its variables $(p_i)$ and $(e^{q_i-q_{i+1}})$ are independent Gaussian and Gamma random variables, respectively. We justify the notion from the physics literature that this model can be thought of as a dense collection of ``quasiparticles'' that act as solitons by, (i) precisely defining the locations of these quasiparticles; (ii) showing that local charges and currents for the Toda lattice are well-approximated by simple functions of the quasiparticle data; and (iii) proving an asymptotic scattering relation that governs the dynamics of the quasiparticle locations. Our arguments are based on analyzing properties about eigenvector entries of the Toda lattice's (random) Lax matrix, particularly, their rates of exponential decay and their evolutions under inverse scattering. 
		
	\end{abstract}
	
	\maketitle 
	
	\tableofcontents
	
	\section{Introduction} 
	
	\label{Introduction} 
	
	The \emph{Toda lattice} is a Hamiltonian dynamical system $(\bm{p}(t); \bm{q}(t))$, where $\bm{p}(t) = (p_i(t))$ and $\bm{q}(t) = (q_i(t))$ are indexed by a one-dimensional integer lattice $i \in \mathscr{I}$ that could either be an interval $\mathscr{I} = [N_1, N_2]$, a torus $\mathscr{I} = \mathbb{Z} / N \mathbb{Z}$, or the full line $\mathscr{I} = \mathbb{Z}$. Its Hamiltonian is given by 
	\begin{flalign*}
		\mathfrak{H} (\bm{p}; \bm{q}) = \displaystyle\sum_{i \in \mathscr{I}} \bigg( \displaystyle\frac{p_i^2}{2} + e^{q_i - q_{i+1}} \bigg),
	\end{flalign*}
	
	\noindent so the dynamics $\partial_t q_i = \partial_{p_i} \mathfrak{H}(\bm{p}; \bm{q})$ and $\partial_t p_i = -\partial_{q_i} \mathfrak{H} (\bm{p}; \bm{q})$ are  
	\begin{flalign*}
		\partial_t q_i (t) = p_i (t), \quad \text{and} \quad \partial_t p_i (t) = e^{q_{i-1} (t) - q_i (t)} - e^{q_i (t) - q_{i+1} (t)}.
	\end{flalign*}
	
	\noindent This model may be thought of as a system of particles moving on the real line, with locations $(q_i)$ and momenta $(p_i)$. It was originally introduced by Toda \cite{VCNI} as a Hamiltonian dynamic that admits soliton solutions, which (loosely speaking) are localized, wave-like functions that retain their shape as they propagate in time. Since the works of Flaschka \cite{LEI} and Manakov \cite{CIS} exhibiting its full set of conserved quantities, and that of Moser \cite{FMPL} determining its scattering shift, the Toda lattice has become recognized as an archetypal example of a completely integrable system. 
	
	A basic question about the Toda lattice is to understand its behavior as the domain $\mathscr{I}$ and time $t$ become large. In situations when its initial data either decays or approximates a profile that is smooth almost everywhere, detailed answers have been attained (largely by analyzing associated  Riemann--Hilbert problems) in works of Venakides--Deift--Oba \cite{TS}, Deift--Kamviss--Kriecherbauer--Zhou \cite{TR}, Deift--McLaughlin \cite{CL}, Bloch--Golse--Paul--Uribe \cite{DO}, and Kr\"{u}ger--Teschl \cite{LADI}.
	
	However, much less is known when the initial data is rough and not decaying, or random. These situations were emphasized Zakharov, broadly within the context of classical integrable systems, under the terms ``soliton gases'' \cite{KES} and ``integrable turbluence'' \cite{TIS}. They are especially prominent, since invariant measures for integrable dynamics are of this type. For example, perhaps the most natural invariant measure for the Toda lattice, sometimes called \emph{thermal equilibrium}, is when the $(p_j)$ and $(e^{q_j - q_{j+1}})$ are independent Gaussian and Gamma random variables, respectively.
	
	Over the past decade, an extensive framework has emerged in the physics literature for predicting how an integrable system behaves in these situations. It is based on the notion that the system can be thought of as a dense collection of many objects called ``quasiparticles,'' which behave as (and thus are occasionally informally identified with) solitons. Under this notion, each quasiparticle possesses an spectral parameter $\lambda_j$ and a location $Q_j(t)$, which have two properties that together pinpoint the system's asymptotics. First, interesting quantities describing the integrable system, such as local charges and currents, are approximable by simple  expressions of the quasiparticles (for example, in the Toda lattice, the number of particles in an interval should nearly equal the number of quasiparticles in it). Second, the evolution of these quasiparticles is well-approximated by an explicit equation, which for the Toda lattice reads \cite[Section IV]{GHCS}
	\begin{flalign}
		\label{qktqk0} 
		\begin{aligned} 
			Q_k (t) \approx Q_k (0) + \lambda_k & t -  2 \displaystyle\sum_{j : Q_j(t) < Q_k (t)} \log |\lambda_k - \lambda_j|  + 2 \displaystyle\sum_{j : Q_j(0) < Q_k (0)} \log |\lambda_k - \lambda_j|. 
		\end{aligned} 
	\end{flalign}
	
	\noindent The relation \eqref{qktqk0} can be interpreted as follows. The $k$-th quasiparticle, initially at $Q_k (0)$, moves with velocity $\lambda_k$ until it meets another quasiparticle, say the $j$-th one. At that moment, the $k$-th quasiparticle instantaneously moves forward or backward by $2 \log |\lambda_k - \lambda_j|$, depending on whether it met the $j$-th quasiparticle from the right or left, respectively (this could make it pass another quasiparticle, producing a ``cascade'' of such interactions, but no two quasiparticles can interact with each other more than once in this way). Then, the $k$-th quasiparticle proceeds at velocity $\lambda_k$ until meeting another quasiparticle, when the procedure repeats. The reason for the choice $2 \log |\lambda_k - \lambda_j|$ is that it is the Toda lattice's scattering shift, describing the phase displacement of just two quasiparticles passing through each other.
	
	This is directly analogous to the behavior of solitons in integrable systems. Indeed, in for example the Korteweg--de Vries (KdV) equation, a soliton solution takes the form of a traveling wave. Its amplitude (shape) is conserved and determined by a spectral parameter; its position is the wave's peak, which moves linearly (``freely'') in time. When two distant solitons evolve under the KdV equation, they initially proceed independently, until they get close to each other. They then undergo a fairly intricate interaction where they interfere, during which their precise locations are ``blurred.'' After this, they emerge as independent solitons; their shapes are the same as before the interaction, but their positions are translated (relative to their free evolutions) by an overall scattering shift. In this way, there is a close parallel between the soliton and quasiparticle dynamics (so, the terms ``soliton'' and ``quasiparticle'' are sometimes used interchangeably). As such, if one runs the KdV equation under an initial data constituting a sparse collection of solitons, one may expect the solitons to mainly interact successively in pairs, thus giving rise to the dynamics described below \eqref{qktqk0}. This reasoning led to the prediction in \cite{KES} that a form of \eqref{qktqk0} should hold for the KdV equation, in this sparse setting.
	
	That a relation such as \eqref{qktqk0} should persist if the solitons are dense was first postulated by Bertini--Collura--De Nardis--Fagotti \cite{TEP} and Castro-Alvaredo--Doyon--Yoshimura \cite{EHSOE} for quantum integrable systems. Although it is less intuitive in this dense regime, some of its consequences for the Toda lattice at thermal equilibrium were investigated through simulations by Cao--Bulchadani--Spohn in \cite{ACCC}, and verified to striking numerical accuracy. 
	
	Below, we refer to \eqref{qktqk0} as the \emph{asymptotic scattering relation}. It is also sometimes called the ``collision rate ansatz'' or ``flea-gas algorithm.'' The latter was originally termed in the context of certain quantum integrable systems by Doyon--Yoshimura--Caux \cite{SGGH}, who provided the informal interpretation of \eqref{qktqk0} described above. It was later studied in greater generality by Doyon--H\"{u}bner--Yoshimura \cite{NCIS}, who also highlighted its resemblance to the dynamics of soliton positions for general classical integrable systems, and to those of wave packets \cite{DGHG} for quantum ones. \\
	
	The question of mathematically making sense of, and justifying, the above framework (for at least some Hamiltonian integrable system) has received considerable interest in recent years. To that end, the spectral parameters $\lambda_j$ of the quasiparticles are understood. They are defined to be the conserved quantities for the system, given by the eigenvalues of its \emph{Lax matrix}. For the Toda lattice, this is the tridiagonal, symmetric matrix $\bm{L}(t)$ whose diagonal and off-diagonal entries are the $(p_i)$ and $(e^{(q_i - q_{i+1})/2})$, respectively \cite{LEI,CIS}. When $\mathscr{I}$ is large and $(\bm{p} (t); \bm{q} (t))$ is random, $\bm{L}(t)$ becomes a high-dimensional random matrix. Its eigenvalue density then prescribes the distribution of quasiparticle spectral parameters in the Toda lattice, whose computation was addressed by Spohn \cite{GECC} (after initial work of Opper \cite{ASCEC}) under thermal equilibrium and many other invariant measures. He predicted formulas for its limiting density (and derived expectations for local currents), which were later verified in works of Mazzuca, Guionnet, and Memin \cite{DSME,LDC,EIST}.
	
	The above results do not address the evolution of the quasiparticle locations $(Q_j)$ in time. Even a coherent definition of these locations $(Q_j)$ from a given state $(\bm{p}; \bm{q})$ of the Toda lattice does not seem to exist in the mathematics literature. In fact, we only know of two integrable systems for which such a definition has been given; they are the hard rods model and box-ball system. For the former, the quasiparticle locations are simply the positions of the rods; for the latter, they can be recovered through a combinatorial algorithm of Takahashi--Satsuma \cite{SCA}. In each case, the analog of the approximate evolution \eqref{qktqk0} becomes exact, and it is an eventual consequence of the inverse scattering that linearizes the system. See works by Boldrighini--Dubroshin--Suhov \cite{OHR} for the hard rods model, and by Ferrari--Nguyen--Rolla--Wang \cite{SDBS} and Croydon--Sasada \cite{GHLBS} for the box-ball one.
	
	Most other integrable systems, including essentially all Hamiltonian ones of interest, lack a definition for their quasiparticle locations from a given state. So, researchers have turned to studying specific families of solutions for them, called finite-gap solutions (see the survey of Dubrovin--Matveev--Novikov \cite{NEV}). These are associated with an algebraic curve and allow logical candidates for solition locations to be incorporated as tunable parameters called phases. The physics work of El \cite{TLE} proposed a specific scaling limit, involving making the curve's genus large and the phases random, under which finite-gap solutions to the KdV equation should satisfy the analog of \eqref{qktqk0}. While certain infinite genus finite-gap solutions have been studied by mathematicians since the work of McKean--Trubowitz \cite{OHFT}, the scaling limit from \cite{TLE} does not seem to have received a thorough mathematical treatment. Still, in a different direction, the  papers by Girotti, Grava, Jenkins, McLaughlin, Minakov, and Najnudel analyzed finite-gap solutions of low genus with many solitons (allowing for random spectral parameters) \cite{RASG,LCRS}. By modifiying these solutions to include one large soliton passing through the many small ones, they used Riemann--Hilbert methods to prove \cite{SGAR} that the position of this ``tracer'' soliton, which should be thought of as a single ``large'' quasiparticle, satisfies a version of \eqref{qktqk0}. 
	
	While finite-gap solutions can sometimes be amenable to analysis, they become very complicated when expressed in terms of the original variables on which the integrable system is defined (which, for the Toda lattice, would be the $(\bm{p}; \bm{q})$ ones). Partly for this reason, a precise interpretation and proof of \eqref{qktqk0} for Hamiltonian integrable systems under most natural families of random initial data (particularly product measures, such as thermal equilibria) had until now been unavailable. \\
	
	The task of making sense of, and justifying, the framework behind \eqref{qktqk0} may be viewed as threefold. 
	\begin{enumerate} 
		\item Define the quasiparticle locations $(Q_j (t))$. 
		\item Show that relevant quantities of the integrable system (such as local charges and currents) are approximable by simple functions of the quasiparticle data $(\lambda_j, Q_j(t))$. 
		\item Establish that the approximate asymptotic scattering relation \eqref{qktqk0} holds.
	\end{enumerate}
	
	In this paper we implement these three tasks for the Toda lattice at thermal equilibrium. See \Cref{ucenter} for the first; \Cref{currentestimate} for the second; and \Cref{ztlambda} for the third. 
	
	We should mention that there is also the fourth task of using \eqref{qktqk0} to prove asymptotic results about the Toda lattice. This can be done as well, though we defer its implementation to the sequel paper \cite{P}, as the ideas used there are almost entirely orthogonal from the ones introduced here. 
	
	Before describing the ideas in this paper, we briefly explain on the setup. Recall that the domain $\mathscr{I}$ of the Toda lattice is either an interval, a torus, or the line. As \Cref{ajbjequation2} below, we will show (if the lengths of the interval and torus are sufficiently large) that these three choices lead to approximately equal systems, which differ by an error decaying exponentially in the time parameter $t$. So, we mainly assume below that $\mathscr{I} = [N_1, N_2]$ is an interval of length $N = N_2 - N_1 + 1 \gg t$. Then, the rows and columns of the Lax matrix $\bm{L}(t) = [L_{ij}(t)]$ are indexed by $i, j \in \mathscr{I} = [N_1, N_2]$. 
	
	We consider the Toda lattice at thermal equilibrium, with parameters $\beta, \theta > 0$. This means (see \Cref{mubeta2}) that we sample the diagonal and superdiagonal entries of the tridiagonal, symmetric Lax matrix $\bm{L}(0)$ independently, with probability densities $C_{\beta} e^{-\beta x^2 / 2}$ and $C_{\beta,\theta} x^{2\theta-1} e^{-\beta x^2}$, respectively (where $C_{\beta}, C_{\beta,\theta} > 0$ are normalization constants). The near-invariance of thermal equilibrium implies that this description continues to approximately\footnote{Thermal equilibrium is exactly invariant if $\mathscr{I}$ is the torus or line. While this does not quite hold when $\mathscr{I}$ is an interval, the above-mentioned comparison between these domains implies that it is approximately true in this case as well. To ease this introductory exposition, we ignore this subtlety for now and act as if invariance held exactly.} hold for $\bm{L}(t)$, when $t > 0$.
	
	1. \emph{Quasiparticle locations}: First, we must provide a definition for the location $Q_j (t)$ of the $j$-th quasiparticle at time $t$. Recall that its spectral parameter is an eigenvalue $\lambda_j$ of the Lax matrix $\bm{L}(t)$. So, we examine the corresponding unit eigenvector $\bm{u}_j (t) = (u_j (N_1; t), \ldots , u_j (N_2; t))$. Central to our analysis is the fact that, if $\bm{L}(t)$ is under thermal equilibrium, then $\bm{u}_j (t)$ is \emph{exponentially localized}. This means that it admits some ``center'' $\varphi \in [N_1, N_2]$ such that $|u_j (i; t)| \le C e^{-c|i-\varphi|}$ likely holds for any $i \in [N_1, N_2]$. Such exponential decay of eigenvectors holds generally for random tridiagonal (or bounded band) matrices with independent (or weakly correlated) entries. It was first proven in the context of one-dimensional Anderson localization by Molchanov \cite{SES} and Kunz--Souillard \cite{SSO}, though the precise estimates we will use are due to Schenker \cite{LRBM} and Aizenman--Schenker--Friedrich--Hundertmark \cite{FFCL}.
	
	Now fix a parameter $\zeta > 0$ that is not too small; we will take $\zeta = e^{-C(\log N)^{3/2}}$ (which decays in $N$ superpolynomially, but not stretched exponentially). Define a \emph{$\zeta$-localization center} of $\lambda_j$ with respect to $\bm{L}(t)$ to be any index $\varphi_t (j) \in [N_1, N_2]$ for which $|u_j (\varphi_t (j); t)| \ge \zeta$; see \Cref{ucenter}. We view this localization center $\varphi_t (j)$ as the index of the particle associated with the $j$-th quasiparticle. So, we define the location of this quasiparticle on $\mathbb{R}$ to be this particle's position $Q_j (t) = q_{\varphi_t (j)} (t)$. 
	
	Observe that there may be multiple choices\footnote{This is analogous to the fact that, in integrable systems (such as the KdV equation), the locations of solitons become ``blurred'' when they interact (and, if the collection of solitons is dense, then they are always interacting with each other).} for this index $\varphi_t (j)$, and thus for the quasiparticle location $Q_j (t)$. However, exponential localization of $\bm{u}_j (t)$ quickly implies that all such choices will be very close to (within $\mathcal{O}(|\log \zeta|) = N^{o(1)}$ of) each other, and hence asymptotically equivalent. 
	
	Before proceeding, let us mention that, guided by quantum mechanics, one might be inclined to instead define the index $\varphi_t (j)$ to be random variable, whose probability of equalling any $k \in [N_1, N_2]$ is $u_j (k; t)^2$. Indeed, this notion was hypothesized in the earlier physics work of Bulchandini--Cao--Moore \cite[Section 3.2]{KTQL}, who observed this probability distribution concentrates on few indices (also attributing it to Anderson localization). If we had adopted this as the definition of $\varphi_t (j)$ in our context, we would have obtained essentially the same results as below, with similar proofs.
	
	2. \emph{Approximate locality}: Next, we must verify that interesting quantities of the Toda lattice (local charges and currents) are simply expressible through the quasiparticle locations $Q_j (t)$ defined above. This will follow from a more fundamental property, which is the ``approximate locality'' of eigenvalues. To explain this, a \emph{local} quantity for the Toda lattice is one that is expressible as a function of the Lax matrix entries $L_{ik} (t)$, for $i$ and $k$ in a uniformly bounded subinterval of $[N_1, N_2]$. The momentum $p_i(t)$ of the $i$-th particle is an example of one, as it is the diagonal $(i,i)$-entry $L_{ii}(t)$ of $\bm{L}(t)$. However, the eigenvalues $\lambda_j$ of $\bm{L}(t)$ are not local, as they depend on all entries of $\bm{L}(t)$. 
	
	Still, we will see under thermal equilibrium that $\lambda_j$ is ``approximately local'' around its localization center $\varphi_t (j)$; informally, this means that the ``dependence'' of $\lambda_j$ on the $i$-th row and column of $\bm{L}(t)$ decays exponentially in $|i-\varphi_t(j)|$. A more specific (but slightly simplified) statement is, if $\ell \gg 1$ is a parameter and $\tilde{\bm{L}}(t)$ is a tridiagonal matrix whose $(i,k)$-entry coincides with that of $\bm{L}(t)$ whenever $i, k \in [\varphi_t (j) - \ell, \varphi_t (j) + \ell]$, then $\tilde{\bm{L}}(t)$ likely has an eigenvalue $\tilde{\lambda}_j$ satisfying $|\lambda_j - \tilde{\lambda}_j| \le C e^{-c\ell}$; see \Cref{lleigenvalues2} and \Cref{lleigenvalues} for precise formulations. This approximate locality can be deduced from the exponential eigenvector localization of $\bm{L}(t)$. Similar deductions were used by Molchanov \cite{LSSO} and Minami \cite{LFSM} in their studies of spectral statistics for random Schr\"{o}dinger operators. 
	
	Given this approximate locality, let us outline how local charges and currents can be recovered from quasiparticle data. Instead of explaining this in complete generality (which would require us to recall the full family of Toda charges and currents; see \Cref{current} for that), to simplify the exposition, we consider a specific example, which computes the total momentum of Toda particles in a large interval $\mathcal{J} \subseteq \mathbb{R}$. In this case, the predicted relation reads \cite[Section III.B]{GHCS} 
	\begin{flalign*}
		\displaystyle\sum_{i: q_i (t) \in \mathcal{J}} p_i (t) \approx \displaystyle\sum_{j: Q_j (t) \in \mathcal{J}} \lambda_j.
	\end{flalign*}
	
	\noindent Its left side is the total momentum (which is the first local charge) in $\mathcal{J}$, and its right side is the sum of the spectral parameters of all quasiparticles in $\mathcal{J}$; see \Cref{currentestimate} for the general statement. It can be shown that the Toda particles $(q_i(t))$ are nearly ordered, so since $Q_j (t) = q_{\varphi_t (j)} (t)$, the above reduces to 
	\begin{flalign}
		\label{plambda} 
		\displaystyle\sum_{i \in \mathcal{I}} p_i (t) \approx \displaystyle\sum_{j : \varphi_t (j) \in \mathcal{I}} \lambda_j,
	\end{flalign} 
	
	\noindent for any large interval of indices $\mathcal{I} \subseteq [N_1, N_2]$ (whose endpoints are the smallest and largest $i$ for which $q_i (t) \in \mathcal{J}$). The left side of \eqref{plambda} is the trace of the matrix $\tilde{\bm{L}}(t)$ obtained from $\bm{L}(t)$ by setting any entry with indices not in $\mathcal{I}$ to $0$. Now, by approximate locality, most eigenvalues $\lambda_j$ on the right side of \eqref{plambda} (namely, those whose localization centers $\varphi_t (j)$ are not too close to an endpoint of $\mathcal{I}$) are approximately equal to a corresponding eigenvalue $\tilde{\lambda}_j$ of $\tilde{\bm{L}}(t)$. This confirms \eqref{plambda}, as it implies that
	\begin{flalign*}
		\displaystyle\sum_{i \in \mathcal{I}} p_i (t) = \Tr \tilde{\bm{L}}(t) = \displaystyle\sum_j \tilde{\lambda}_j \approx \displaystyle\sum_{j: \varphi_t (j) \in \mathcal{I}} \lambda_j.
	\end{flalign*}

	3. \emph{Asymptotic scattering relation}: It remains to access the quasiparticle dynamics, by proving \eqref{qktqk0}. We first use the inverse scattering relation for the Toda lattice, shown in \cite{FMPL}, providing the explicit (linear) evolution for the first entry $u_k (N_1; t)$ of the eigenvector $\bm{u}_k (t)$ of $\bm{L}(t)$. It states that
	\begin{flalign}
		\label{uklambdak}  
		2 \log |u_k (N_1; t)| = 2 \log |u_k (N_1; 0)| - \lambda_k t + \mathfrak{C}(t),
	\end{flalign}

	\noindent for some constant $\mathfrak{C}(t)$ independent of $k$. Now, recall since $\bm{u}_k (s)$ is exponentially localized for any $s \in [0, t]$ that $\log |u_k (N_1; s)|$ decays linearly in $\varphi_s (k) - N_1$. In fact, there is a \emph{Lyapunov exponent} $\gamma_k \le 0$ satisfying $\log |u_k (N_1; s)| \approx \gamma_k \cdot (\varphi_s (k) - N_1)$, which admits an exact formula called the \emph{Thouless relation}. Predicted by Thouless \cite{DSRL} and established by Avron--Simon \cite{APO}, it states 
	\begin{flalign}
		\label{gammakintegral} 
		2 \gamma_k = 2 \cdot \mathbb{E} [\log L_{i,i+1}(s)] - 2 \displaystyle\int_{-\infty}^{\infty} \log |\lambda_k - \lambda| \varrho (\lambda) d \lambda,
	\end{flalign}
	
	\noindent where $\varrho$ is the limiting eigenvalue density for $\bm{L}(s)$ (as its dimension becomes large), and $i \in [N_1, N_2-1]$ is any index (since the $L_{i,i+1}(s)$ are identically distributed under thermal equilibrium, $\mathbb{E}[L_{i,i+1}(s)]$ does not depend $i$). While the relation \eqref{gammakintegral} will not be of direct use to us in proving \eqref{qktqk0}, a discrete variant of it will be. More precisely, by combining a linear algebraic identity in \cite{DSRL}, the transfer matrix arguments of \cite{APO}, and the approximate locality of eigenvalues, we will show (see \Cref{uk12}) that
	\begin{flalign}
		\label{uksn1} 
		2 \log |u_k (N_1; s)| \approx q_{N_1} (s) - Q_k (s) - 2 \displaystyle\sum_{j: \varphi_s (j) < \varphi_s (k)} \log |\lambda_k - \lambda_i|.
	\end{flalign}
	
	Let us briefly indicate the relation between \eqref{uksn1} and \eqref{gammakintegral}. Since (by definition) $2 \log L_{i,i+1} (s) = q_i (s) - q_{i+1} (s)$ and $Q_k = q_{\varphi_s (k)} (s)$, one may expect the first two terms on the right side of \eqref{uksn1} to approximate the first term on the right side of \eqref{gammakintegral}, multiplied by $\varphi_s(k)-N_1$. Also, since the sum on the right side of \eqref{uksn1} constitutes $\varphi_s (k) - N_1$ terms (and since $\varrho$ is the limiting density of the $(\lambda_i)$), one may expect it to approximate the integral on the right side of \eqref{gammakintegral}, multiplied by $\varphi_s (k) - N_1$. In this way, \eqref{uksn1} may be thought of as a discrete form of \eqref{gammakintegral}.
	
	Inserting \eqref{uksn1} at $s \in \{ 0, t \}$ into \eqref{uklambdak} yields 
	\begin{flalign}
		\label{qktsum} 
		\begin{aligned} 
			Q_k (t) & \approx Q_k (0) + \lambda_k t - 2 \displaystyle\sum_{j: \varphi_t (j) < \varphi_t (k)} \log |\lambda_k - \lambda_i| \\ 
			& \qquad + 2 \displaystyle\sum_{j: \varphi_0 (j) < \varphi_0 (k)} \log |\lambda_k - \lambda_i|  + q_{N_1} (t) - q_{N_1} (0) - \mathfrak{C}(t).
		\end{aligned} 
	\end{flalign} 
	
	\noindent It can be shown that $\mathfrak{C}(t) \approx q_{N_1} (t) - q_{N_1} (0)$. Using this in \eqref{qktsum}, with the fact that the $(q_j)$ are ordered (to transform the restrictions on $\varphi_s (j)$ in \eqref{qktsum} to ones on $q_{\varphi_s (j)} (s) = Q_j (s)$), gives \eqref{qktqk0}.
	
	In fact, upon implementing this outline, one finds (see \Cref{ztlambda}) that the error in \eqref{qktqk0} is quite small, of order $(\log N)^C$. This indicates that the asymptotic scattering relation is close to exact, which might explain why the simulations in \cite{ACCC} had such a high degree of numerical accuracy. \\
	
	Before proceeding, we briefly comment on the potential applicability of the above framework to other initial data or different integrable systems. First, there are many (an infinite-dimensional family of) invariant measures for the Toda lattice other than thermal equilibrium; they are the \emph{generalized Gibbs ensembles} with polynomial potential (see \cite[Section 2]{GECC}). Under these measures, the entries $(p_i, e^{(q_i - q_{i+1})/2})$ of the Lax matrix $\bm{L}(t)$ are not independent, but their coupling is of finite range. It was indicated in \cite[Section 2]{LRBM} that the arguments of \cite{LRBM} should apply to show that such tridiagonal matrices also have exponentially localized eigenvectors. Thus, it could be possible to establish above results under generalized Gibbs ensembles, as well. 
	
	Next, one might inquire about initial data that is not invariant.\footnote{In this case, the belief is that the Toda lattice in the large time limit should be asymptotically governed by a certain \emph{generalized hydrodynamic equation} \cite{GHCS,HSIMS}. Properties of these limiting equations have been analyzed for other integrable systems, such as the Lieb--Liniger model by Doyon--H\"{u}bner--Yoshimura \cite{GTDC}, an KdV and nonlinear Schr\"{o}dinger equations by Kuijlaars--Tovbis \cite{MESC}.} If this initial data is sampled under a product measure (or one with finite-range couplings), then again the framework of \cite{LRBM} should yield exponential eigenvector localization of the initial Lax matrix $\bm{L}(0)$. However, since the initial data is not invariant, this does not immediately extend to later times $t \gg 1$. If one could show that the eigenvectors of $\bm{L}(t)$ were still localized for $t \gg 1$, then it may again be possible to show the asymptotic scattering relation \eqref{qktqk0} remains valid away from invariant initial data. 
	
	There are numerous other integrable systems admitting Lax matrices that are (skew-)symmetric tridiagonal or of Cantero--Moral--Vel\'{a}zquez (CMV) \cite{FMZP} type. These include the Volterra lattice; generalized Toda lattices studied by Deift--Nanda--Tomei \cite{DSE}; and Ablowitz--Ladik (and Schur flow) hierarchy considered by Killip--Nenciu \cite{UAM}. As explained by Grava--Gisonni--Gubbiotti--Mazzuca \cite{DSRM} and Spohn \cite{HSIMS}, invariant measures analogous to thermal equilibrium exist for these systems. Moreover, relevant eigenvector properties of the associated random matrices are known in many of these cases (though sometimes in a weaker form than what we use here). For example, exponential eigenvector localization and an analog of the Thouless relation for random CMV matrices were established in \cite{OPUC} by Simon. These properties also hold in certain settings for various continuum integrable equations, such as the KdV equation\footnote{As observed by Saitoh \cite{TCLE} (and by Matetski--Quastel--Remenik \cite{PGL}), the KdV equation can also be formally recovered from a limit degeneration of the Toda lattice. It remains to be seen whether this can be justified to allow asymptotic frameworks about the latter to carry over to the former.} \cite{SES,APO} (for which the well-posedness of relevant invariant measures was established by Forlano, Kilip, Murphy, and Visan \cite{INL,IMIV}). It would be of interest to investigate if the framework developed here can be applied to these contexts, as well.

	\subsection*{Acknowledgements} 
	
	The author heartily thanks Alexei Borodin, Patrick Lopatto, Jeremy Quastel, and Herbert Spohn for valuable conversations, as well as Benjamin Doyon, Matthew Nicoletti, and Philippe Sosoe for helpful comments on this manuscript. This work was partially supported by a Clay Research Fellowship and a Packard Fellowship for Science and Engineering.

	\section{Results} 
	
	\label{Chain} 
	
	We introduce the Toda lattice and its Lax matrix in \Cref{Lattice}, and we define its thermal equilibrium initial data in \Cref{0Stationary}. In \Cref{Center0}, we introduce localization centers, and we explain how they can be used to access locally conserved quantities (also called local charges) and currents. We then state the asymptotic scattering relation for the Toda lattice under thermal equilibrium (\Cref{ztlambda}), in \Cref{Asymptotic}.
	
	Throughout, for any $a, b \in \mathbb{R}$, set $\llbracket a, b \rrbracket = [a, b] \cap \mathbb{Z}$. A vector $\bm{v} = (v_1, v_2, \ldots , v_n) \in \mathbb{C}^n$ is a unit vector if $\sum_{i=1}^n v_i^2 = 1$; we call it nonnegatively normalized if $v_j > 0$, where $j \in \llbracket 1, n \rrbracket$ is the minimal index such that $v_j \ne 0$. For any real symmetric $n \times n$ matrix $\bm{M}$, let $\eig \bm{M} = (\lambda_1, \lambda_2, \ldots , \lambda_n)$ denote the eigenvalues of $\bm{M}$, counted with multiplicity and ordered so that $\lambda_1 \ge \lambda_2 \ge \cdots \ge \lambda_n$.

	\subsection{Toda Lattice} 
	
	\label{Lattice}

	\subsubsection{Open Toda Lattice} 
	
	\label{Open} 
	
	In this section we recall the Toda lattice on an interval. Throughout, we fix integers $N_1 \le N_2$ and set $N = N_2 - N_1 + 1$ (which will prescribe the interval's endpoints and length, respectively). 
	
	The state space of the Toda lattice on the interval $\llbracket N_1, N_2 \rrbracket$, also called the \emph{open Toda lattice}, is given by a pair of $N$-tuples $( \bm{p} (t); \bm{q}(t) ) \in \mathbb{R}^N \times \mathbb{R}^N$, where $\bm{p}(t) = ( p_{N_1} (t), \ldots , p_{N_2} (t) )$ and $\bm{q}(t) = ( q_{N_1} (t), \ldots , q_{N_2} (t) )$; both are indexed by a real number $t \ge 0$ called the time. Given any \emph{initial data} $( \bm{p}(0); \bm{q}(0) ) \in \mathbb{R}^N \times \mathbb{R}^N$, the joint evolution of $( \bm{p}(t); \bm{q}(t) )$ for $t \ge 0$ is prescribed by the system of ordinary differential equations 
	\begin{flalign}
		\label{qtpt} 
		\partial_t q_j (t) = p_j (t), \quad \text{and} \quad \partial_t p_j (t) = e^{q_{j-1} (t) - q_j (t)} - e^{q_j (t) - q_{j+1} (t)},
	\end{flalign}
	
	\noindent for all $(j, t) \in \llbracket N_1, N_2 \rrbracket \times \mathbb{R}_{\ge 0}$; here, we set $q_{N_1-1}(t) = -\infty$ and $q_{N_2+1}(t) = \infty$ for all $t \ge 0$. One might interpret this as the dynamics for $N$ points (indexed by $\llbracket N_1, N_2 \rrbracket$) moving on the real line, whose locations and momenta at time $t \ge 0$ are given by the $(q_i(t))$ and $(p_i(t))$, respectively. 	
	
	The system of differential equations \eqref{qtpt} is equivalent to the Hamiltonian dynamics generated by the Hamiltonian $\mathfrak{H} : \mathbb{R}^N \times \mathbb{R}^N \rightarrow \mathbb{R}$ that is defined, for any $\bm{p} = (p_0, p_1, \ldots , p_{N-1}) \in \mathbb{R}^N$ and $\bm{q} = (q_0, q_1, \ldots , q_{N-1}) \in \mathbb{R}^N$, by setting
	\begin{flalign}
		\label{hpq}
		\mathfrak{H} (\bm{p}; \bm{q}) = \displaystyle\sum_{j=0}^{N-1} \bigg( \displaystyle\frac{p_j^2}{2} + e^{q_j - q_{j+1}} \bigg),
	\end{flalign} 
	
	\noindent where $q_N = \infty$. The existence and uniqueness of solutions to \eqref{qtpt} for all time $t \ge 0$, under arbitrary initial data $(\bm{p}; \bm{q}) \in \mathbb{R}^N \times \mathbb{R}^N$, is thus a consequence of the Picard--Lindel\"{o}f theorem (see, for example, the proof of \cite[Theorem 12.6]{OINL}). 
	
	It will often be useful to reparameterize the variables of the Toda lattice, following \cite{LEI}. To that end, for any $(i, t) \in \llbracket N_1, N_2 \rrbracket \times \mathbb{R}_{\ge 0}$, define 
	\begin{flalign}
		\label{abr} 
		r_i (t) = q_{i+1} (t) - q_i (t); \qquad a_i (t) = e^{-r_i(t)/2}; \qquad b_i (t) = p_i (t).
	\end{flalign}
	
	\noindent Denoting $\bm{a}(t) = ( a_{N_1} (t), \ldots , a_{N_2} (t) ) \in \mathbb{R}_{\ge 0}^N$ and $\bm{b}(t) = ( b_{N_1} (t), \ldots , b_{N_2} (t) ) \in \mathbb{R}^N$, the $( \bm{a} (t); \bm{b} (t) )$ are sometimes called \emph{Flaschka variables}; they satisfy $r_{N_2}(t) = q_{N_2+1}(t) - q_{N_2} (t) = \infty$ and $a_{N_2} (t) = e^{-r_{N_2}(t) / 2} = 0$. Then, \eqref{qtpt} is equivalent to the system  
	\begin{flalign}
		\label{derivativepa} 
		\partial_t a_j (t) = \displaystyle\frac{a_j (t)}{2} \cdot \big(b_j (t) - b_{j+1} (t) \big), \qquad \text{and} \qquad \partial_t b_j (t) = a_{j-1} (t)^2 - a_j (t)^2,
	\end{flalign} 
	
	\noindent for each $(j, t) \in \llbracket N_1, N_2 \rrbracket \times \mathbb{R}_{\ge 0}$. 
	
	It will at times be necessary to define the original Toda state space variables $( \bm{p}(t); \bm{q}(t) )$ from the Flaschka variables $( \bm{a}(t); \bm{b}(t) )$; it suffices to do this at $t = 0$, as $( \bm{p} (t); \bm{q}(t) )$ is determined from $( \bm{p}(0); \bm{q}(0) )$, by \eqref{qtpt}. We explain how to do this when $0 \in \llbracket N_1, N_2 \rrbracket$ (as otherwise we may translate $(N_1, N_2)$ to guarantee that this inclusion holds).\footnote{In this work, we will usually have $N_1$ and $N_2$ be large negative and large positive integers, respectively, and we will be interested in the $(p_i(t); q_i (t))$ for $i$ in some interval containing $0$.} By \eqref{abr}, the Flaschka variables $\bm{a}(0)$ only specify the differences between consecutive entries in $\bm{q}(0)$, so the former only determines the latter up to an overall shift. We will fix this shift by setting $q_0 (0) = 0$, that is, we define $(\bm{p}(0); \bm{q}(0) )$ by imposing  
	\begin{flalign}
		\label{q00} 
		q_0 (0) = 0; \qquad q_{i+1} (0) - q_i (0) = - 2 \log a_i (0); \qquad p_i (0) = b_i (0),
	\end{flalign} 
	
	\noindent  for each $i \in \llbracket N_1, N_2 \rrbracket$. Then, $( \bm{p} (0); \bm{q} (0) )$ is called the Toda state space intial data associated with $(\bm{a} (0); \bm{b}(0) )$. The evolution $( \bm{p}(t), \bm{q}(t) )$ of this initial data under \eqref{qtpt} is called the Toda state space dynamics associated with $( \bm{a}(t), \bm{b}(t) )$; observe that we may have $q_0 (t) \ne 0$ if $t \ne 0$.

	\subsubsection{Lax Matrices and Locally Conserved Quantities}
	
	\label{MatrixL} 
	
	In this section we recall the Lax matrix, and local charges, associated with the Toda lattice. Throughout, we fix integers $N_1 \le N_2$ and set $N = N_2 - N_1 + 1$. Let $( \bm{a}(t); \bm{b}(t) ) \in \mathbb{R}_{\ge 0}^N \times \mathbb{R}^N$ be a pair of $N$-tuples indexed by a real number $t \in \mathbb{R}_{\ge 0}$, where $\bm{a}(t) = ( a_j (t) )$ and $\bm{b}(t) = ( b_j (t) )$ satisfies the system \eqref{derivativepa} for each $(j, t) \in \llbracket N_1, N_2 \rrbracket \times \mathbb{R}$; we assume (as explained above \eqref{derivativepa}) that $a_{N_2}(t) = 0$ for each $t \in \mathbb{R}_{\ge 0}$. Under this notation, the Lax matrix (introduced in \cite{LEI,CIS}) is defined as follows.
	
	\begin{definition}[Lax matrix]
		
		\label{matrixl} 
		
		For any real number $t \ge 0$, the \emph{Lax matrix} $\bm{L}(t) = [ L_{ij}] = [L_{ij}(t)]$ is an $N \times N$ real symmetric matrix, with entries indexed by $i, j \in \llbracket N_1, N_2 \rrbracket$, defined as follows. Set
		\begin{flalign*} 
			& L_{ii} = b_i (t), \qquad \qquad \qquad \quad \text{for each $i \in \llbracket N_1, N_2 \rrbracket$}; \\
			& L_{i,i+1} = L_{i+1,i} = a_i (t), \qquad \text{for each $i \in \llbracket N_1, N_2-1 \rrbracket$}.
		\end{flalign*} 
		
		\noindent Also set $L_{ij} = 0$ for any $i,j \in \llbracket N_1, N_2 \rrbracket$ with $|i-j| \ge 2$. 
		
	\end{definition} 
	
	A fundamental feature of the Lax matrix is that its eigenvalues are preserved under the Toda dynamics \eqref{derivativepa}. This was originally due to \cite{LEI}; see also \cite[Section 2]{FMPL}.
	
	\begin{lem}[\cite{LEI,FMPL}]
		
		\label{ltt}
		
		For any real numbers $t, t' \in \mathbb{R}_{\ge 0}$, we have $\eig \bm{L} (t) = \eig \bm{L}(t')$. 
		
	\end{lem} 
	
	Lemma \ref{ltt} provides a large family of conserved quantities for the Toda lattice, given by the eigenvalues of the Lax matrix. However, these are ``non-local,'' in the sense that they depend on all of the Flaschka variables $(\bm{a}(t);\bm{b}(t))$, as opposed to only the $(a_i(t),b_i(t))$ for $i$ in some (uniformly) bounded interval. To remedy this, observe that \Cref{ltt} implies that $\Tr \bm{L}(t)^m$ is conserved for any integer $m \ge 0$. Since $\bm{L}(t)$ is tridiagonal, the diagonal entries of $\bm{L}(t)^m$ are local quantities (see \cite[Section 2.1]{HSIMS} for further information), whose total is preserved; they are called local charges. In what follows, for any matrix $\bm{M}$, we let $[M]_{ij}$ denote the $(i,j)$ entry of $\bm{M}$. 
	
	\begin{definition}[Local charges]
		
		\label{current} 
		
		Fix an integer $m \ge 0$, an index $i \in \llbracket N_1, N_2 \rrbracket$, and a real number $t \ge 0$. Define the \emph{$m$-th local charge} $\mathfrak{k}_i^{[m]} (t)$ of $\bm{L}(t)$ at $i$ by setting $\mathfrak{k}_i^{[m]} (t) = [\bm{L}(t)^m]_{ii}$. 
		
	\end{definition} 
	
	For example, the first local charge $\mathfrak{k}_i^{[1]} (t) = b_i (t) = p_i (t)$ denotes momentum.

	\subsection{Thermal Equilibrium} 
	
	\label{0Stationary} 
	
	In this section we describe a class of random initial data that we will study for the Toda lattice; it is sometimes referred to as thermal equilibrium, and is given by independent Gamma random variables for the $\bm{a}$ Flaschka variables, and independent Gaussian random variables for the $\bm{b}$ ones. In what follows, we fix an integer $N \ge 1$ and real numbers $\beta, \theta > 0$.
	
	\begin{definition}[Open thermal equilibrium]
		
		\label{mubeta2} 
		
		The \emph{(open) thermal equilibrium} with parameters $(\beta, \theta; N)$ is the product measure $\mu = \mu_{\beta, \theta} = \mu_{\beta, \theta; N-1,N}$ on $\mathbb{R}^{N-1} \times \mathbb{R}^N$ defined by 
		\begin{flalign*}
			\mu (d \bm{a}; d \bm{b}) = \bigg( \displaystyle\frac{2 \beta^{\theta}}{\Gamma(\theta)} \bigg)^{N-1}  (2 \pi \beta^{-1})^{-N/2} \cdot \displaystyle\prod_{j=0}^{N-2} a_j^{2\theta-1} e^{-\beta a_j^2} da_j  \displaystyle\prod_{j=0}^{N-1} e^{-\beta b_j^2/2} db_j,
		\end{flalign*}
		
		\noindent where $\bm{a} = (a_0, \ldots , a_{N-2}) \in \mathbb{R}_{\ge 0}^{N-1}$ and $\bm{b} = (b_0, b_1, \ldots , b_{N-1}) \in \mathbb{R}^N$. It will be convenient to view $\mu_{\beta,\theta;N-1,N}$ as a measure on $\mathbb{R}^N \times \mathbb{R}^N$ by, if we denote $\hat{\bm{a}}(0) = (a_0, a_1, \ldots , a_{N-2}, 0) \in \mathbb{R}_{\ge 0}^N$, then also saying $(\hat{\bm{a}}; \bm{b}) \in \mathbb{R}^N \times \mathbb{R}^N$ is sampled under $\mu_{\beta,\theta;N-1,N}$. 
		
	\end{definition} 
	
	Thermal equilibrium, as described above, is related to invariant measures for the Toda lattice; the latter are measures on the Flaschka variable initial data $(\bm{a}(0);\bm{b}(0))$ such that, for any $t \ge 0$, the law of $(\bm{a}(t); \bm{b}(t))$ under the Toda lattice is the same as that of $(\bm{a}(0); \bm{b}(0))$. The open Toda lattice on a finite interval $\llbracket N_1, N_2 \rrbracket$ admits no nontrivial invariant measures. However, the Toda lattice on the full line $\mathbb{Z}$ does, though this must be made sense of, as it involves infinitely many variables. In this context, the thermal equilibrium product measure of \Cref{mubeta2} (extrapolated to when $N=\infty$) is perhaps the most natural invariant measure. 
	
	The following proposition, to be shown in \Cref{RLimit} below (as a consequence of \Cref{infiniteab}, which addresses more general initial data), states this to be the case. Specifically, it shows that at thermal equilibrium the Toda lattice on $\mathbb{Z}$ can be defined by taking a limit of open Toda lattices on growing intervals; confirms that thermal equilibrium is invariant for this infinite Toda lattice; and provides a quantitative approximation for it by open Toda lattices at thermal equilibrium on finite intervals. 
	
	\begin{prop} 
		
		\label{ajbjequation2} 
		
		Let $(a_j)_{j \in \mathbb{Z}}$ and $(b_j)_{j \in \mathbb{Z}}$ be mutually independent random variables, whose laws for each $j \in \mathbb{Z}$ are given by
		\begin{flalign*}
			& \mu(da_j) = 2\beta^{\theta} \Gamma(\theta)^{-1} \cdot a_j^{2\theta-1} e^{-\beta a_j^2} da_j; \\
			& \mu(db_j) = (2 \pi \beta^{-1})^{-1/2} \cdot e^{-\beta b_j^2/2} db_j.
		\end{flalign*}
		
		\noindent For any integers $N_1 \le N_2$, let $( \bm{a}^{[N_1,N_2]} (t), \bm{b}^{[N_1,N_2]} (t) )$ denote the open Toda lattice \eqref{derivativepa} on $\llbracket N_1, N_2 \rrbracket$, with initial data $( \bm{a}^{[N_1,N_2]} (0); \bm{b}^{[N_1,N_2]} (0) )$ given by $\bm{a}^{[N_1,N_2]} (0) = (a_{N_1}, a_{N_1+1}, \ldots , a_{N_2-1})$ and $\bm{b}^{[N_1,N_2]} (0) = (b_{N_1}, b_{N_1+1}, \ldots , b_{N_2})$. 
		
		\begin{enumerate} 
			\item For each $(j, t) \in \mathbb{Z} \times \mathbb{R}_{\ge 0}$, the limits $\lim_{N \rightarrow \infty} a_j^{[-N,N]} (t) = a_j (t)$ and $\lim_{N \rightarrow \infty} b_j^{(N)} (t) = b_j (t)$ exist almost surely; are finite; and solve \eqref{derivativepa}. 	
			\item For each $t \ge 0$, denoting $\bm{a}(t) = (a_j(t))_{j \in \mathbb{Z}}$ and $\bm{b}(t) = (b_j(t))_{j \in \mathbb{Z}}$, the laws of $(\bm{a}(0);\bm{b}(0))$ and $(\bm{a}(t);\bm{b}(t))$ coincide.
			\item For any $\mathfrak{p} \in (0, 1)$, there exist constants $c = c(\beta,\theta,\mathfrak{p})>0$ and $C = C (\beta, \theta, \mathfrak{p}) >1$ such that the following holds. Let $R \ge 1$ and $T \ge 0$ be real numbers, and $K \ge 1$ and $N_1 \le - K \le K \le N_2$ be integers, such that 
			\begin{flalign*} 
				R \ge \log N, \quad \text{and} \quad RT \le K \le N^{\mathfrak{p}}, \qquad \text{where} \quad N = N_2-N_1+1.
			\end{flalign*} 
			
			\noindent Then, with probability at least $1 - C e^{-cR^2}$, we have
			\begin{flalign*}
				\displaystyle\sup_{t \in [0, T]} \displaystyle\max_{j \in \llbracket N_1+K, N_2 - K \rrbracket} \big( \big| a_j(t) - a_j^{[N_1,N_2]} (t) \big| + \big| b_j (t) - b_j^{[N_1,N_2]} (t) \big| \big) \le e^{-K/5}. 
			\end{flalign*}
		\end{enumerate} 
		
	\end{prop} 
	
	The system $( \bm{a}(t); \bm{b}(t) )$ from \Cref{ajbjequation2} is the Toda lattice on $\mathbb{Z}$, under thermal equilibrium with parameters $(\beta, \theta)$. Let us briefly interpret the third part of \Cref{ajbjequation2}, when the interval length $N$ satisfies $T \ll N \ll e^T$. It states that, with high probability (at least $1 - Ce^{-(\log T)^2}$), the Flaschka variables at site $j$ of the Toda lattice on $\mathbb{Z}$ approximately coincide with those of the open Toda lattice on $\llbracket N_1, N_2 \rrbracket$, up to an error that is exponentially small in $T$, as long as $j$ is in the ``bulk'' of $\llbracket N_1, N_2 \rrbracket$ (that is, not too close to its endpoints $N_1$ and $N_2$; otherwise, boundary effects on the interval should become more visible and make this comparison invalid). Thus, asymptotic questions about the Toda lattice run for some large time $T$, on $\mathbb{Z}$ under thermal equilibrium, can be recovered from those about the open Toda lattice on a finite interval at thermal equilibrium. Hence, we will throughout focus on the open Toda lattice on the finite interval $\llbracket N_1, N_2 \rrbracket$ for times $t \in [0, T]$ with $T \in [N^{\delta}, N^{1-\delta}]$ for some small constant $\delta>0$ (though the statements below will permit more flexibility in $T$), and be interested in its variables at sites $j$ in the bulk of $\llbracket N_1, N_2 \rrbracket$.
	
	Before proceeding, let us mention that one might also be interested in analyzing Toda lattice on the torus $\mathbb{T}_N = \mathbb{Z} / N \mathbb{Z}$, for which thermal equilibrium is also invariant (see \Cref{ModelPeriodic} below). It can be deduced from \Cref{aabbk2} below that a counterpart of \Cref{ajbjequation2} holds comparing this Toda lattice on $\mathbb{T}_N$ to the open one on $\llbracket N_1, N_2 \rrbracket$, indicating they are likely very close to each other in the bulk of $\llbracket N_1, N_2 \rrbracket$, for times $T \ll N = N_2-N_1+1$ (for much larger times, we expect the periodicity of the torus to become more visible, again making such a comparison invalid). As such, asymptotic questions about the Toda lattice on $\mathbb{T}_N$ at thermal equilibrium, for times $T \in [ N^{\delta}, N^{1-\delta}]$, also reduce to those about the model on the interval.

	\subsection{Localization Centers} 
	
	\label{Center0} 
	
	As explained in \Cref{MatrixL}, while the eigenvalues of the Lax matrix $\bm{L}(t)$ are conserved under the Toda dynamics, they are not local, as they depend on all entries of $\bm{L}(t)$. Still, we will see under thermal equilibrium that they are ``approximately local,'' in that they only depend on a few entries of $\bm{L}(t)$, up to a small error. These entries will correspond to those on which the associated eigenvectors of $\bm{L}(t)$ are mainly supported; we call them localization centers, given (in a more general context) by the below definition. Here, we again let $N_1 \le N_2$ be integers and set $N = N_2 - N_1 + 1$. 
	
	\begin{definition}[Localization centers]
		
		\label{ucenter}
		
		Let $\bm{u} = ( u(N_1), \ldots , u(N_2) ) \in \mathbb{R}^N$ be a unit vector. For any $\zeta \in \mathbb{R}_{\ge 0}$, we call an index $\varphi \in \llbracket N_1, N_2 \rrbracket$ a \emph{$\zeta$-localization center} for $\bm{u}$ if $| u(\varphi) | \ge \zeta$. 
		
		Next, let $\bm{M} = [M_{ij}]$ be a symmetric $N \times N$ matrix, with entries indexed by $i, j \in \llbracket N_1, N_2 \rrbracket$. If $\lambda \in \eig \bm{M}$, then we call $\varphi \in \llbracket N_1, N_2 \rrbracket$ a $\zeta$-localization center for $\lambda$ with respect to $\bm{M}$ if $\varphi$ is a $\zeta$-localization center for some unit eigenvector $\bm{u} \in \mathbb{R}^N$ of $\bm{M}$ with eigenvalue $\lambda$. Further let $(\bm{u}_1, \bm{u}_2, \ldots , \bm{u}_N)$ denote an orthonormal eigenbasis of $\bm{M}$. We call a bijection $\varphi : \llbracket 1, N \rrbracket \rightarrow \llbracket N_1, N_2 \rrbracket$ a \emph{$\zeta$-localization center bijection} for $\bm{M}$ if $\varphi(j)$ is a $\zeta$-localization center for $\bm{u}_j$ for each $j \in \llbracket 1, N \rrbracket$.

	\end{definition}

	One might first ask whether $\zeta$-localization center bijections exist. The following lemma, whose quick proof appears in \Cref{Center} below, verifies this if $\zeta \le (2N)^{-1}$. 
	
	\begin{lem} 
		
		\label{bijectionm} 
		
		Any symmetric $N \times N$ matrix $\bm{M}$ admits a $(2N)^{-1}$-localization center bijection.
		
	\end{lem} 
	
	One might next ask if such localization centers bijections are unique; this is not always the case. An extreme example is if $\bm{M}$ is the adjacency matrix for the discrete Laplacian on a torus of length $N$, in which case one can verify that any bijection from $\llbracket 1, N \rrbracket$ to $\llbracket N_1, N_2 \rrbracket$ is an $N^{-1/2}$-localization center bijection.
	
	However, we will show that this uniqueness is true ``up to some error,'' if $\bm{M}$ is the Lax matrix for the Toda lattice run under thermal equilibrium. To make this precise, we require some notation, which will frequently be adopted throughout the remainder of this paper. 
	
	\begin{assumption}
		
		\label{lbetaeta} 
		
		Fix real numbers\footnote{Throughout this paper, constants may depend on $\beta$ and $\theta$, even when not explicitly stated.} $\beta, \theta > 0$, and assume that the \emph{stretch parameter}  
		\begin{flalign}
			\label{alpha} 
			\alpha =  \log \beta - \displaystyle\frac{\Gamma'(\theta)}{\Gamma(\theta)} \ne 0.
		\end{flalign} 
		
		\noindent For each $t \ge0 $, let $\bm{L}(t) = [L_{ij}(t)]$ denote the Lax matrix for the open Toda lattice $(\bm{a}(t);\bm{b}(t))$ on $\llbracket N_1, N_2 \rrbracket$ (as in \Cref{matrixl}). Set $\eig \bm{L}(t) = (\lambda_1, \lambda_2, \ldots , \lambda_N)$, which does not depend on $t$ by \Cref{ltt}, and for each $k \in \llbracket 1, N \rrbracket$ let $\bm{u}_k (t) = (u_k (N_1;t), u_k (N_2;t), \ldots,  u_k (N_2; t))$ denote the nonnegatively normalized, unit eigenvector of $\bm{L}(t)$ with eigenvalue $\lambda_k$. At $t = 0$, abbreviate $\bm{L} = \bm{L}(0)$; $(\bm{a}; \bm{b}) = ( \bm{a}(0); \bm{b}(0))$; and $\bm{u}_k = \bm{u}_k (0) = ( u_k (N_1), u_k (N_1 + 1), \ldots , u_k (N_2) )$. Assume that the initial data $(\bm{a}; \bm{b})$ is sampled under the thermal equilibrium $\mu_{\beta, \theta; N-1,N}$ from \Cref{mubeta2}. Let $( \bm{p}(s); \bm{q}(s) )$, over $s \in \mathbb{R}_{\ge 0}$, denote the Toda state space dynamics associated with (Flaschka variable) initial data $( \bm{a}(0); \bm{b}(0) )$, as in \Cref{Open}.
		
		Let $T \ge 1$ and $\zeta \ge 0$ be real numbers satisfying
		\begin{flalign}
			\label{n1n2zetat} 
			\begin{aligned} 
				& \quad N_1 \le -N(\log N)^{-1} \le N (\log N)^{-1} \le N_2; \\
				& 1 \le T \le N (\log N)^{-7}; \qquad \zeta \ge e^{-100(\log N)^{3/2}}.
			\end{aligned} 
		\end{flalign}
		
		\noindent For each $s \in \mathbb{R}$ let $\varphi_s: \llbracket 1, N \rrbracket \rightarrow \llbracket N_1, N_2 \rrbracket$ be a $\zeta$-localization center bijection for $\bm{L}(s)$, and denote 
		\begin{flalign}
			\label{qjs2}
			Q_j (s) = q_{\varphi_s (j)} (s), \qquad \text{for each $(j, s) \in \llbracket 1, N \rrbracket \times \mathbb{R}_{\ge 0}$}.
		\end{flalign}

	\end{assumption}

	The reason for the term, ``stretch parameter,'' is given by \Cref{qij} below, which indicates that the average spacing $\mathbb{E}[q_i(0) - q_{i+1}(0)] = \alpha$ between points in the Toda lattice at thermal equilibrium is $\alpha$. The assumption \eqref{alpha} states that the stretch parameter $\alpha$ is nonzero, which implies that the Toda particles $(q_j(t))$ do not ``accumulate'' (with infinite density) around $0$. 
	
	Let us briefly explain \eqref{n1n2zetat}. Its first bound indicates that $0$ is in the bulk of $\llbracket N_1, N_2 \rrbracket$. Its second ensures that the time scale $T$ is sublinear in $N$, in accordance with the discussion at the end of \Cref{0Stationary} (as this will guarantee that the boundary of $\llbracket N_1, N_2 \rrbracket$ does not asymptotically affect its bulk under the Toda lattice). Its third ensures that $\zeta$ is not too small. In \eqref{n1n2zetat}, the constants $7$, $100$, and $3/2$ are of little significance (we will need $e^{-(\log N)^2} \ll \zeta \ll N^{-C}$ for $C$ sufficiently large).
	
	Under \Cref{lbetaeta}, we view $\varphi_j (s)$ as the ``location'' of $\lambda_j$ on the lattice $\llbracket N_1, N_2 \rrbracket$, and $Q_j (s)$ as the location of the corresponding Toda particle\footnote{One might ask why we impose that $\varphi_t : \llbracket 1, N \rrbracket \rightarrow \llbracket N_1, N_2 \rrbracket$ is a bijection, as opposed to any map such that $\varphi_t (j)$ is a $\zeta$-localization center of $\bm{u}(j;t)$ for each $j \in \llbracket 1, N \rrbracket$. As suggested by \Cref{centerdistance} below (which addresses arbitrary localization centers, and not only localization center bijections), this would not be necessary for our results below to hold. Still, we adopt it since it has the aesthetic property of exhibiting a one-to-one correspondence betwen eigenvalues of the Lax matrix and particles in the Toda lattice.} $q_{\varphi_s (j)} (s)$ on the line $\mathbb{R}$ (the latter may be thought of as the location for the ``quasiparticle with amplitude $\lambda_j$,'' as described in \Cref{Introduction}). 
	
	The following proposition, to be shown in \Cref{ProofCenterDistance} below, indicates that the location $\varphi_s (j)$ is likely unique, up to a small error (as long as it is in the bulk of $\llbracket N_1, N_2 \rrbracket$).
	
	\begin{prop} 
		
		\label{centerdistance} 
		
		Adopt \Cref{lbetaeta}. There exists a constant $c > 0$ such that the following holds with probability at least $1-c^{-1} e^{-c(\log N)^2}$. For any real number $t \in [0, T]$; eigenvalue $\lambda \in \eig \bm{L}(t)$; and two $\zeta$-localization centers $\varphi, \tilde{\varphi} \in \llbracket N_1, N_2 \rrbracket$ for $\lambda$ with respect to $\bm{L}(t)$, satisfying
		\begin{flalign}
			\label{tzeta}
			N_1 + T (\log N)^3 \le \varphi \le N_2 - T (\log N)^3,
		\end{flalign} 
		
		\noindent  we have $|\varphi - \tilde{\varphi}| \le (\log N)^3$. 
		
	\end{prop}
	
	As mentioned above, one may think of the ``location'' of $\lambda_j$ on the lattice $\llbracket N_1, N_2 \rrbracket$ at time $s$ as being $\varphi_j (s)$. One manifestation of this is through \Cref{lleigenvalues2} (and \Cref{lleigenvalues}) below, which essentially states the following. Let $\lambda \in \eig \bm{L}$ have $\zeta$-localization center $\varphi$. If one perturbs the entries of $\bm{L}$, whose indices are sufficiently far from $\varphi$, to form a tridiagonal matrix $\tilde{\bm{L}}$, then there exists an eigenvalue $\tilde{\lambda} \in \eig \tilde{\bm{L}}$ that is close to $\lambda$ (and whose $N^{-1} \zeta$-localization center is close to $\varphi$). 
	
	We will not state that result more precisely in this section, but instead mention the following relevant consequence of it, to be proven in \Cref{ProofCurrent} below. It indicates that the (sums of) local charges from \Cref{current}, and their currents (that is, their net transfer across a given location in space), are well-approximated by local sums of $\lambda_j^m$, over $j$ satisfying $Q_j (t) \approx q_i (t)$. Thus, to analyze asymptotics of the former, it suffices to understand the limiting dynamics of the particle locations $Q_j (t)$.

	\begin{prop}
		
		\label{currentestimate} 
		
		Adopt \Cref{lbetaeta}. There exists a constant $c>0$ such that the following two statements hold with probability at least $1-c^{-1} e^{-c (\log N)^2}$, for any integer $m \in \llbracket 0, \log N \rrbracket$. Below, we let $t \in [0, T]$ be any real number and $\mathcal{J} \subset \mathbb{R}$ be any interval satisfying
		\begin{flalign}
			\label{j} 
			\begin{aligned} 
				\mathcal{J} \subseteq [ \alpha N_1 + (T+|N_1|^{1/2})(\log N)^5, \alpha N_2 - (T+N_2^{1/2}) (\log N)^5], \quad \text{if $\alpha > 0$}; \\
				\mathcal{J} \subseteq [ \alpha N_2 + (T+N_2^{1/2})(\log N)^5, \alpha N_1 - (T+|N_1|^{1/2}) (\log N)^5], \quad \text{if $\alpha < 0$}.
			\end{aligned}
		\end{flalign} 
		
		\begin{enumerate} 
			\item We have
			\begin{flalign}
				\label{currentsum} 
				\Bigg| \displaystyle\sum_{i: q_i (t) \in \mathcal{J}} \mathfrak{k}_i^{[m]} (t) - \displaystyle\sum_{j: Q_j (t) \in \mathcal{J}} \lambda_j^m \Bigg| \le (\log N)^{m+6}.
			\end{flalign} 
			\item For any $k \in \mathcal{J}$, we have 
			\begin{flalign} 
				\label{currentsum2} 
				\Bigg| \displaystyle\sum_{i: q_i (t) < k} \mathfrak{k}_i^{[m]} (t) -  \displaystyle\sum_{i: q_i (0) < k} \mathfrak{k}_i^{[m]} (0)  - \displaystyle\sum_{j : Q_j (t) < k} \lambda_j^m + \displaystyle\sum_{j: Q_j (0) < k} \lambda_j^m \Bigg| \le (\log N)^{m+6}. 
			\end{flalign} 
		\end{enumerate} 
		
	\end{prop} 
	
	Let us briefly comment on \Cref{currentestimate}. First, the assumption \eqref{j} on $\mathcal{J}$ implying that $\mathcal{J}$ is bounded away from $\alpha N_1$ and $\alpha N_2$ indicates that\footnote{Indeed, since $q_0 (0) = 0$ and the stretch parameter $\alpha$ denotes the expected value of $q_{i+1} - q_i$, the extreme particles in the Toda lattice should reside around $q_{N_1} (0) \approx \alpha N_1$ and $q_{N_2} (0) \approx \alpha N_2$.}  this inveral only contains particles $q_i$ of index $i$ in the bulk of $\llbracket N_1, N_2 \rrbracket$. Second, it is typical in the context of hydrodynamical limits to consider local charges on an interval $\mathcal{J}$ whose length is of order $T$. Then, the sums on the left side of \eqref{currentsum} (or their differences in \eqref{currentsum2}) will also be of order $T$, which can be taken to be much larger than $(\log N)^{m+6}$, so the error on the right side of \eqref{currentsum} is negligible in comparison (for $m \ll \log N$).

	\subsection{Asymptotic Scattering Relation} 
	
	\label{Asymptotic}

	The following theorem, proven in \Cref{YEstimate} below, provides the asymptotic scattering relation that governs the dynamics of the $Q_k (t)$. In what follows, the condition \eqref{n1n2k022} ensures that the ``initial location'' $\varphi_0 (k)$ of $\lambda_k$ is in the bulk of $\llbracket N_1, N_2 \rrbracket$. Observe that the error\footnote{No effort was made to optimize the exponent $15$, which could likely be improved.} $(\log N)^{15}$ on the right side of \eqref{lambdak22} is quite small, negligible in comparison to any power of $N$. So, \Cref{ztlambda} indicates that the asymptotic scattering relation is quite close to an equality.

	\begin{thm}[Asymptotic scattering relation]
		
		\label{ztlambda} 
		
		Adopt \Cref{lbetaeta}. There exists a constant $c > 0$ such that the following holds with probability at least $1 - c^{-1} e^{-c(\log N)^2}$. Let $k \in \llbracket 1, N \rrbracket$ satisfy
		\begin{flalign}
			\label{n1n2k02}
			N_1 + T (\log N)^6 \le \varphi_0 (k) \le N_2 - T (\log N)^6.
		\end{flalign}
		
		\noindent Then, for each $t \in [0, T]$, we have
		\begin{flalign}
			\label{lambdak22}
			\begin{aligned}
				\Bigg| \lambda_k & t - Q_k(t) + Q_k (0) - 2 \sgn (\alpha) \displaystyle\sum_{i:  Q_i (t) <  Q_k (t)}  \log |\lambda_{k} - \lambda_{i}| \\
				& + 2 \sgn (\alpha) \displaystyle\sum_{i: Q_i (0) < Q_k (0)} \log |\lambda_k - \lambda_i| \Bigg| \le (\log N)^{15}.
			\end{aligned} 
		\end{flalign} 
		
	\end{thm}

	In \Cref{0Speed} below, we will provide a heuristic (following the physics literature) for how \Cref{ztlambda} can be used to evaluate the limiting velocities of the $\lambda_k$, as predicted in \cite{GHCS,HSIMS}. While we do not know how to mathematically justify this heuristic directly, its output will be established for sufficiently small $\theta$ in a sequel work \cite{P}, through a different method.

	\subsection{Notation} 
	
	\label{Notation} 
	
	For any integer $n \ge 1$, denote $\mathbb{T}_n = \mathbb{Z} / n \mathbb{Z}$. For any point $z \in \mathbb{C}$ and set $\mathcal{A} \subseteq \mathbb{C}$, denote $\dist (z, \mathcal{A}) = \inf_{s \in \mathcal{A}} |z-s|$. Denote the complement of any event $\mathsf{E}$ by $\mathsf{E}^{\complement}$. Denote the set of $n \times n$ real matrices by $\Mat_{n \times n}$. For any $\bm{M} \in \Mat_{n \times n}$, denote its transpose by $\bm{M}^{\mathsf{T}}$. Denote the set of $n \times n$ real symmetric matrices by $\SymMat_{n \times n} = \{ \bm{M} \in \Mat_{n \times n} : \bm{M} = \bm{M}^{\mathsf{T}} \}$. 	
	
	As in \Cref{matrixl}, it will often be convenient to index the rows and columns of $n \times n$ matrices by index sets different from $\llbracket 1, n \rrbracket$. Given a nonempty index set $\mathscr{I} \subset \mathbb{Z}$ of size $n = |\mathscr{I}|$, let $\Mat_{\mathscr{I}}$ denote the set of $n \times n$ real matrices $\bm{M} = [M_{ij}]_{i,j \in \mathscr{I}} \in \Mat_{n \times n}$, whose rows and columns are indexed by $\mathscr{I}$; also let $\SymMat_{\mathscr{I}} = \Mat_{\mathscr{I}} \cap \SymMat_{n \times n}$ denote the set of real symmetric matrices whose rows and columns are indexed by $\mathscr{I}$. Given any matrix $\bm{M} \in \Mat_{\mathscr{I}}$ and subset $\mathcal{K} \subseteq \mathscr{I}$, let $\bm{M}^{(\mathcal{K})} = [ M_{ij}^{(\mathcal{K})} ]_{i,j \in \mathscr{I}} \in \Mat_{\mathscr{I}}$ denote the matrix obtained from $\bm{M}$ by setting all its entries in a row or column indexed by an element of $\mathcal{K}$ to $0$, namely, $M_{ij}^{(\mathcal{K})} = \mathbbm{1}_{i \notin \mathcal{K}} \cdot \mathbbm{1}_{i \notin \mathcal{K}}  \cdot M_{ij}$. 
	
	Throughout, given some integer parameter $N \ge 1$ and event $\mathsf{E}_N$ depending on $N$, we say that $\mathsf{E}_N$ \emph{holds with overwhelming probability} if there exists a constant $c > 0$ such that $\mathbb{P} [\mathsf{E}_N^{\complement}] \le c^{-1} e^{-c(\log N)^2}$. In this case, we call $\mathsf{E}_N$ \emph{overwhelmingly probable}. Observe that, whenever proving that $\mathsf{E}_N$ is overwhelmingly probable, we may assume $N \ge N_0$ is sufficiently large; we will often do this implicitly (and without comment) throughout this work.

	\subsection{Outline} 
	
	\label{Outline}
	
	The remainder of this paper is organized as follows. We begin in \Cref{MatrixLattice} with some miscellaneous facts about the Toda lattice and random matrices; the results there are mainly known (and those for which we lack a reference are proven in \Cref{Proof2}). In \Cref{LCompare} we provide comparison estimates for the Toda lattice on different domains, showing \Cref{ajbjequation2} in the process. In \Cref{EigenvalueLocal} we show various results (\Cref{lleigenvalues2} and \Cref{lleigenvalues} below) indicating that eigenvalues of the Lax matrix under thermal equilibrium are approximately local, which will be used in \Cref{ProofsCenter} to prove the results about localization centers stated in  \Cref{Center0}. In \Cref{Gamma} we provide an effective approximation for the rate of exponential decay for entries of random Lax matrix eigenvectors, which is used to establish the asymptotic scattering relation (\Cref{ztlambda}) in \Cref{InverseAsymptotic}.

	\section{Miscellaneous Preliminaries} 
	
	\label{MatrixLattice} 
	
	\subsection{Periodic Toda Lattice} 
	
	\label{ModelPeriodic}

	In this section we recall the periodic Toda lattice on an torus (which will be parallel to \Cref{Open}). Throughout, we fix an integer $N \ge 1$ and a real number $\Upsilon \in \mathbb{R}$ (which will prescribe the torus size and a periodicity constraint, respectively). The state space of the Toda lattice on $\mathbb{T}_N$, also called the \emph{periodic Toda lattice}, with parameter $\Upsilon$ is given by a pair of $N$-tuples of real numbers $( \bm{p}(t); \bm{q}(t) ) \in \mathbb{R}^N \times \mathbb{R}^N$, where $\bm{p}(t) = ( p_0 (t), p_1 (t), \ldots , p_{N-1} (t) )$ and $\bm{q}(t) = ( q_0 (t), q_1 (t), \ldots , q_{N-1} (t) )$; both are by indexed a real number $t \ge 0$ called the time. Given any initial data $( \bm{p} (0); \bm{q}(0) ) \in \mathbb{R}^N \times \mathbb{R}^N$, the joint evolution of $( \bm{p} (t); \bm{q}(t) )$ for $t \ge 0$ is prescribed by the system of differential equations \eqref{qtpt} for all $(j, t) \in \llbracket 0, N-1 \rrbracket \times \mathbb{R}_{\ge 0}$. Here, we have set
	\begin{flalign*} 
		p_{i+N} (t) = p_i (t), \quad \text{and} \quad q_{i+N} (t) = q_i (t) + \Upsilon, \quad \text{for each $(i, t) \in \mathbb{Z} \times \mathbb{R}_{\ge 0}$},
	\end{flalign*}
	
	\noindent which corresponds to the periodicity constraints of considering the Toda lattice on the torus $\mathbb{T}_N$. The system of differential equations \eqref{qtpt} is equivalent to the Hamiltonian dynamics generated by the Hamiltonian $\mathfrak{H}$ from \eqref{hpq} (where again we set $q_N = q_0 + \Upsilon$). The existence and uniqueness of solutions to \eqref{qtpt} for all time $t \ge 0$, under arbitrary initial data $(\bm{p}; \bm{q}) \in \mathbb{R}^N \times \mathbb{R}^N$, is thus a consequence of the Picard--Lindel\"{o}f theorem; see \cite[Theorem 12.6]{OINL}.
	
	We moreover again define $\bm{r} (t) = ( r_0 (t), r_1 (t), \ldots , r_{N-1} (t) )$, as well as the Flaschka variables $\bm{a}(t) = ( a_0 (t), a_1 (t), \ldots , a_{N-1} (t) ) \in \mathbb{R}_{\ge 0}^N$ and $\bm{b}(t) = ( b_0 (t), b_1 (t), \ldots , b_{N-1} (t) ) \in \mathbb{R}^N$, according to \eqref{abr}. Then, \eqref{qtpt} is equivalent to the system \eqref{derivativepa}, for each $(j, t) \in \mathbb{Z} \times \mathbb{R}_{\ge 0}$, where the periodicity constraints become 
	\begin{flalign*} 
		a_i (t) = a_{i+N} (t); \quad b_i (t) = b_{i+N} (t); \quad r_i (t) = r_{i+N} (t); \quad \Upsilon = \displaystyle\sum_{j=1}^N r_i (t),
	\end{flalign*} 
	
	\noindent for each $(i, t) \in \mathbb{Z} \times \mathbb{R}_{\ge 0}$ (where the fact that $\Upsilon$ is independent of $t$ follows from the fact that $\partial_t r_j (t) = 2 a_j(t)^{-1} \cdot \partial_t a_j(t)$, by dividing both sides of the first equation in \eqref{derivativepa} by $a_j (t)$ and then summing over $j \in \llbracket 0, N-1 \rrbracket$). If $0 \in \llbracket N_1, N_2 \rrbracket$, then we can prescribe Toda state space variables $( \bm{p}(t); \bm{q}(t) )$ associated with these Flaschka variables $( \bm{a}(t); \bm{b}(t))$ as in \Cref{Open}. 
	
	Next we recall the Lax matrix associated with the periodic Toda lattice. In what follows, we adopt the notation above, letting $(\bm{a}(t);\bm{b}(t)) \in \mathbb{R}^N \times \mathbb{R}^N$ denote the periodic Toda lattice.
	
	\begin{definition}[Periodic Lax matrix]
		
		\label{matrixl2} 
		
		For any $t \in \mathbb{R}_{\ge 0}$, define the \emph{Lax matrix} $\bm{L}(t) = [ L_{ij}]_{i,j \in \mathscr{I}} = [L_{ij} (t)] \in \SymMat_{\llbracket 0, N-1 \rrbracket}$ as follows, where $L_{ij} = L_{ij} (t)$. Set
		\begin{flalign*} 
			L_{ii} = b_i (t), \qquad \text{for each $i \in \llbracket 0, N-1 \rrbracket$}.
		\end{flalign*} 
		
		\noindent Set
		\begin{flalign*}
			& L_{i,i+1} = L_{i+1,i} = a_i (t), \quad \text{for each $i \in \llbracket 0, N-2 \rrbracket$}; \\
			& L_{0,N-1} = L_{N-1,0} = a_{N-1} (t).
		\end{flalign*} 
		
		\noindent Also set $L_{ij} = 0$ if $(i, j)$ with $i \ne j$ is not of the above form. 
		
	\end{definition} 
	
	As in \Cref{ltt} (for the open model), the eigenvalues of the Lax matrix are preserved under the periodic Toda dynamics \eqref{derivativepa}. This was originally due to \cite{LEI}; see also \cite[Theorem 12.5]{OINL}.
	
	\begin{lem}[\cite{LEI,OINL}]
		
		\label{ltt2}
		
		For any real numbers $t, t' \in \mathbb{R}$, we have $\eig \bm{L} (t) = \eig \bm{L}(t')$. 
		
	\end{lem} 
	
	Next we define thermal equilibrium for the periodic Toda lattice (similarly to \Cref{mubeta2}).

	\begin{definition}[Periodic thermal equilibrium]
		
		\label{mubeta} 
		
		Fix $\beta, \theta \in \mathbb{R}_{> 0}$. The \emph{(periodic) thermal equilibrium} with parameters $(\beta, \theta; N)$ is the product measure $\mu = \mu_{\beta, \theta} = \mu_{\beta, \theta; N}$ on $\mathbb{R}^N \times \mathbb{R}^N$ defined by 
		\begin{flalign*}
			\mu (d \bm{a}; d \bm{b}) = \bigg( \displaystyle\frac{2^{1/2} \beta^{\theta+1/2}}{ \pi^{1/2}  \Gamma(\theta)} \bigg)^N \cdot \displaystyle\prod_{j=0}^{N-1} a_j^{2\theta-1} e^{-\beta a_j^2} da_j  \displaystyle\prod_{j=0}^{N-1} e^{-\beta b_j^2/2} db_j,
		\end{flalign*}
		
		\noindent where $\bm{a} = (a_0, a_1, \ldots , a_{N-1}) \in \mathbb{R}_{\ge 0}^N$ and $\bm{b} = (b_0, b_1, \ldots , b_{N-1}) \in \mathbb{R}^N$. 
	\end{definition}

	We now recall the notion of invariance for the periodic Toda lattice. To that end, fix an integer $N \ge 1$; let $\mu$ denote a probability measure on $ \mathbb{R}_{\ge 0}^N \times \mathbb{R}^N$; and sample $(\bm{a}; \bm{b}) \in \mathbb{R}_{\ge 0}^N \times \mathbb{R}^N$ under $\mu$. Consider the periodic Toda lattice in the Flaschka variables \eqref{derivativepa} on $\mathbb{T}_N$, denoted by $( \bm{a}(t); \bm{b}(t) )$, with initial data $( \bm{a}(0); \bm{b}(0) ) = (\bm{a}; \bm{b})$. If the law of $( \bm{a}(t); \bm{b}(t) )$ coincides with the law $\mu$ of $(\bm{a}; \bm{b})$ for each $t \ge 0$, then we say that the measure $\mu$ is \emph{invariant} for the periodic Toda lattice.

	The following (known) lemma states the invariance of the above measure for the periodic Toda lattice. Its proof is a quick consequence of the Liouville theorem, with the fact that the quantities 
	\begin{flalign}
		\label{12q} 
		\sum_{j=0}^{N-1} \log a_j (t), \qquad \text{and} \qquad \sum_{j=0}^{N-1} \big( 2 a_j (t)^2 + b_j(t)^2 \big),
	\end{flalign} 
	
	\noindent are conserved under the Toda dynamics \eqref{derivativepa}; see, for example, \cite{DSRM} for further details.

	\begin{lem}[{\cite[Sections 2 and 3]{DSRM}}]
		
		\label{betathetainvariant}
		
		For any integer $N \ge 1$ and real numbers $\beta, \theta \in \mathbb{R}_{>0}$, the thermal equilibrium $\mu_{\beta,\theta;N}$ is invariant for the periodic Toda lattice on $\mathbb{T}_N$.   
	\end{lem}

	\subsection{Lax Matrix Eigenvector Evolution} 
	
	\label{LinearOpen}

	Before proceeding, we state the following consequence of \Cref{ltt}, to be proven in \Cref{ProofR} below. It indicates that, over time $t \in \mathbb{R}_{\ge 0}$, the maximum magnitude of an entry of the Lax matrix $\bm{L}(t)$ can change by at most a constant factor. 
	
	\begin{lem} 
		
		\label{abltt} 
		
		Adopt the notation of either \Cref{matrixl} or \Cref{matrixl2}, and let $\mathscr{I} = \llbracket N_1, N_2 \rrbracket$ in the first (open) case and $\mathscr{I} = \llbracket 0, N-1 \rrbracket$ in the second (periodic) one. Denote
		\begin{flalign*}
			A(t) = \displaystyle\max_{i \in \mathscr{I}} | a_i (t) |; \qquad B(t) = \displaystyle\max_{i \in \mathscr{I}} | b_i (t) |.
		\end{flalign*} 
		
		\noindent Then, for any $t, t' \in \mathbb{R}$, we have $A(t) + B(t) \le 6 ( A(t') + B(t') )$. 	
	\end{lem}

	We next discuss how the eigenvectors of Lax matrices evolve under the open Toda lattice. This discussion will only apply to the open case, so for the remainder of this section we restrict to that case, meaning we consider the Toda lattice on $\llbracket N_1, N_2 \rrbracket$. Then $\bm{L}(t)$ is for any $t \in \mathbb{R}_{\ge 0}$ a real, symmetric, tridiagonal matrix with nonzero off-diagonal entries; as such, its eigenvalues are mutually distinct (see \cite[Proposition 2.40(a)]{PRM}). 
	
	Denote $\eig \bm{L}(t) = (\lambda_1,  \ldots , \lambda_N)$, which is independent of $t$ by \Cref{ltt}. For any $t \in \mathbb{R}_{\ge 0}$ and $j \in \llbracket 1, N \rrbracket$, let $\bm{u}_j (t) = ( u_j (N_1; t),  \ldots , u_j (N_2; t) ) \in \mathbb{R}^N$ denote the nonnegatively normalized, unit eigenvector of $\bm{L}(t)$ with eigenvalue $\lambda_j$ (so that $\bm{L}(t) \cdot \bm{u}_j (t) = \lambda_j \cdot \bm{u}_j (t)$). The following result, due to \cite{FMPL}, identifies the evolution of the first entry of these eigenvectors under the Toda dynamics (indicating that their logarithms evolve linearly, upon suitable normalization).

	\begin{lem}[{\cite[Section 3]{FMPL}}]
		
		\label{uknt} 
		
		For any $t \in \mathbb{R}_{\ge 0}$ and $k \in \llbracket 1, N \rrbracket$, we have 
		\begin{flalign*}
			u_k (N_1; t)^2  = e^{-\lambda_k t}  u_k (N_1; 0)^2 \cdot \Bigg( \displaystyle\sum_{j=1}^N e^{-\lambda_j t}  u_j (N_1; 0)^2 \Bigg)^{-1}.
		\end{flalign*}
		
	\end{lem}

	\subsection{Resolvents} 
	
	\label{ResolventG}
	
	In this section we recall properties of resolvents. Throughout, we let $\mathscr{I} \subset \mathbb{Z}$ denote a nonempty, finite set of indices, and set $n = |\mathcal{I}|$. Let $\bm{M} = [M_{ij}] \in \Mat_{\mathscr{I}}$. Given $z \in \mathbb{C}$, the associated \emph{resolvent} of $\bm{M}$ is $\bm{G}(z) = (\bm{M}-z)^{-1} = [ G_{ij} (z) ] = [ G_{i,j} (z) ]$, assuming that this inverse exists. Denote $\eig \bm{M} = (\lambda_1, \lambda_2, \ldots , \lambda_n)$, and let $(\bm{u}_1, \bm{u}_2, \ldots , \bm{u}_n)$ denote an orthonormal family of eigenvectors of $\bm{M}$, so that $\bm{u}_k = ( u_k (i) )_{i \in \mathscr{I}}$ is an eigenvector of $\bm{M}$ with eigenvalue $\lambda_k$. Then, for any $i, j \in \mathscr{I}$, observe that  
	\begin{flalign}
		\label{gsum} 
		G_{ij} (z) = \displaystyle\sum_{k = 1}^n \displaystyle\frac{u_k (i) u_k (j)}{\lambda_k - z},
	\end{flalign}
	
	\noindent which by the orthonormality of the $(\bm{u}_k)$ implies for $\Imaginary z \ne 0$ that 
	\begin{flalign}
		\label{gijeta} 
		| G_{ij} (z) | \le |\Imaginary z|^{-1}. 
	\end{flalign}
	
	\noindent Further observe, for any invertible matrices $\bm{A}, \bm{B} \in \Mat_{n \times n}$, that we have   
	\begin{flalign}
		\label{ab} 
		\bm{A}^{-1} - \bm{B}^{-1} = \bm{A}^{-1} (\bm{B}-\bm{A}) \bm{B}^{-1} = \bm{B}^{-1} (\bm{B} - \bm{A}) \bm{A}^{-1}.
	\end{flalign}
	
	The following lemma indicates that, if the resolvents of two matrices are close, then their eigenvalues are as well. Its proof will appear in \Cref{ProofR} below.
	
	\begin{lem}
		
		\label{abgh} 
		
		Fix an index $\varphi \in \mathscr{I}$; real numbers $\eta, \zeta, \delta > 0$; and symmetric $n \times n$ matrices $\bm{A}, \bm{B} \in \SymMat_{\mathscr{I}}$. Fix an eigenvalue $\lambda \in \eig \bm{A}$, with an associated unit eigenvector $\bm{u} = ( u(i) )_{i \in \mathscr{I}} \in \mathbb{R}^n$; denote $z = \lambda + \mathrm{i} \eta$; and set $\bm{G} = [G_{ij}] = (\bm{A}-z)^{-1}$ and $\bm{H} = [H_{ij}] = (\bm{B}-z)^{-1}$. Assume that 
		\begin{flalign}
			\label{deltaetachi} 
			\delta \le (2\eta)^{-1} \zeta^2; \qquad | u(\varphi) | \ge \zeta; \qquad \big| G_{\varphi \varphi} (z) - H_{\varphi \varphi} (z) \big| \le \delta.
		\end{flalign}
		
		\noindent Then there exists an eigenvalue $\mu \in \eig \bm{B}$, with an associated unit eigenvector $\bm{v} = ( v(i) )_{i \in \mathscr{I}} \in \mathbb{R}^n$, such that $|\lambda-\mu| \le 3 n \zeta^{-2} \eta$ and $| v(\varphi) | \ge (6n)^{-1/2} \zeta$.
		
	\end{lem}

	\subsection{Eigenvectors of Tridiagonal Matrices}
	
	\label{EigenvectorM} 
	
	In this section we provide (mostly known) facts about the eigenvectors of tridiagonal matrices. Throughout, we let $\bm{M} = [M_{ij}] \in \Mat_{\llbracket N_1, N_2 \rrbracket}$ denote a tridiagonal, real symmetric $N \times N$ matrix, where $N=N_2-N_1+1$; we assume that $M_{i,i+1} \ne 0$ for each $i \in \llbracket N_1, N_2 - 1 \rrbracket$. The below lemma expresses certain eigenvector entries of $\bm{M}$ in terms of the matrix entries of $\bm{M}$; it was originally due to \cite{DSRL}, but we provide its proof in \Cref{ProofMatrixS} below.
	
	\begin{lem}[{\cite[Equation (4)]{DSRL}}]
		
		\label{n1n2u}
		
		Let $\mu \in \eig \bm{M}$, and let $\bm{u} = (u_{N_1}, \ldots , u_{N_2}) \in \mathbb{R}^N$ denote a unit eigenvector of $\bm{M}$ with eigenvalue $\mu$. Then, 
		\begin{flalign*}
			\log |u_{N_1}| + \log |u_{N_2}| = \displaystyle\sum_{i=N_1}^{N_2-1} \log |M_{i,i+1}| - \displaystyle\sum_{\mu' \in \eig \bm{M} \setminus \{ \mu \}} \log |\mu - \mu'|.
		\end{flalign*}
		
	\end{lem}

	Next we discuss more precise behavior for the eigenvectors of $\bm{M}$ through the method of transfer matrices, following \cite[Section 4]{APO}. For any integer $k \in \llbracket N_1, N_2-1 \rrbracket$ and real number $E \in \mathbb{R}$, define the \emph{transfer matrix} $\bm{S}_k (E) = \bm{S}_k (E; \bm{M}) \in \Mat_{2\times2}$ by setting 
	\begin{flalign}
		\label{ks} 
		\bm{S}_k (E)= \left[ \begin{array}{cc} 0 & 1 \\ -M_{k,k+1}^{-1} M_{k-1,k} & M_{k,k+1}^{-1} (E - M_{k,k}) \end{array} \right].
	\end{flalign}
	
	\noindent Moreover, for any subset $\mathcal{K} = (k_1, k_2, \ldots , k_m) \subseteq \llbracket N_1, N_2 - 1 \rrbracket$ with $k_1 < k_2 < \cdots < k_m$, define $\bm{S}_{\mathcal{K}} (E) = \bm{S}_{\mathcal{K}} (E; \bm{M}) \in \Mat_{2 \times 2}$ by setting
	\begin{flalign}
		\label{k2s} 
		\bm{S}_{\mathcal{K}} (E) = \bm{S}_{k_m} (E) \cdot \bm{S}_{k_{m-1}} (E) \cdots \bm{S}_{k_1} (E). 
	\end{flalign} 
	
	The following (standard) lemma indicates how the above transfer matrices can be used to evaluate eigenvector entries of $\bm{M}$. Its proof is given in \Cref{ProofMatrixS} below.
	
	\begin{lem}
		
		\label{smatrixu} 
		
		Let $\mu \in \eig \bm{M}$, and let $\bm{u} = (u_{N_1}, u_{N_1+1}, \ldots , u_{N_2}) \in \mathbb{R}^N$ denote an eigenvector of $\bm{M}$ with eigenvalue $\mu$. For each $k \in \llbracket N_1, N_2 \rrbracket$, set $\bm{w}_k = (u_{k-1}, u_k) \in \mathbb{R}^2$, where $u_{N_1-1} = 0$. Then $\bm{S}_{\llbracket i,j \rrbracket} (\mu) \cdot \bm{w}_i = \bm{w}_{j+1}$ for any $i, j \in \llbracket N_1, N_2 - 1 \rrbracket$ with $i \le j$. 
		
	\end{lem}

	The next lemma expresses entries of $\bm{S}_{\llbracket i, j \rrbracket} (E)$ through eigenvalues of truncations of $\bm{M}$; its proof appears in \cite{OINL}, but we also provide it in \Cref{ProofMatrixS} below. In what follows, for any integers $k, \ell \in \llbracket N_1, N_2 \rrbracket$ with $k \le \ell$, we let $\bm{M}^{[k, \ell]}$ denote the $(\ell-k+1) \times (\ell-k+1)$ matrix obtained by restricting $\bm{M}$ to rows and columns indexed by $i, j \in \llbracket k, \ell \rrbracket$; stated equivalently, for each $i, j \in \llbracket k, \ell \rrbracket$, the $(i, j)$-entry of $\bm{M}^{[k,\ell]}$ is equal to the $(i,j)$-entry of $\bm{M}$.
	
	\begin{lem}[{\cite[Equation (1.65)]{OINL}}]
		
		\label{sij} 
		
		Set $\eig \bm{M}^{[i,j]} = ( \mu_1^{[i,j]}, \mu_2^{[i, j]}, \ldots , \mu_{j-i+1}^{[i, j]} )$ for any $i, j \in \llbracket N_1, N_2 \rrbracket$ with $i \le j$. Then, for any $i, j \in \llbracket N_1, N_2 - 1 \rrbracket$, denoting $\ell = j - i + 1$ we have
		\begin{flalign*}
			\bm{S}_{\llbracket i, j \rrbracket} (E) = \left[ \begin{array}{cc} - M_{i-1,i}  \displaystyle\prod_{k=i}^{j-1} M_{k,k+1}^{-1}  \displaystyle\prod_{h=1}^{\ell-2} ( E - \mu_h^{[i+1,j-1]} ) & \displaystyle\prod_{k=i}^{j-1} M_{k,k+1}^{-1}  \displaystyle\prod_{h=1}^{\ell-1} ( E - \mu_h^{[i,j-1]} ) \\
				- M_{i-1,i}  \displaystyle\prod_{k=i}^{j} M_{k,k+1}^{-1}   \displaystyle\prod_{h=1}^{\ell-1} ( E - \mu_h^{[i+1,j]} ) & \displaystyle\prod_{k=i}^{j} M_{k,k+1}^{-1} \displaystyle\prod_{h=1}^{\ell} ( E - \mu_h^{[i, j]} )  \end{array} \right].
		\end{flalign*}
		
	\end{lem}

	\subsection{Spectral Behavior of Random Lax Matrices} 
	
	\label{Localization}

	In this section we describe spectral properties of Lax matrices whose Flaschka variables are sampled from thermal equilibrium. In what follows, we fix real numbers $\beta, \theta > 0$; the constants below might depend on them, even if not stated explicitly. 
	
	The below two lemmas indicate that the stretch parameter $\alpha$ from \Cref{lbetaeta} prescribe the average distance between particles in the Toda lattice under thermal equilibrium. We establish them in \Cref{ProofIntegral} below.
	
	\begin{lem}
		
		\label{aintegral} 
		
		Let $\mathfrak{a} > 0$ be a random variable with law $\mathbb{P} [\mathfrak{a} \in (a,a+da)] = 2 \beta^{\theta} \cdot \Gamma(\theta)^{-1} \cdot a^{2\theta-1} e^{-\beta a^2} da$. Denoting $\mathfrak{a} = e^{-\mathfrak{r}/2}$, we have that $\mathbb{E} [\mathfrak{r}] = \alpha$.
	\end{lem}

	\begin{lem}
		
		\label{qij} 
		
		Adopt \Cref{lbetaeta}. There exists a constant $c > 0$ such that the following holds. For any distinct indices $i, j \in \llbracket N_1, N_2 \rrbracket$ and real number $R \ge 1$, we have 
		\begin{flalign*}
			\mathbb{P} \big[ | q_j (0) - q_i (0) - \alpha(j-i) | \ge R \big] \le 2 (e^{-cR^2/|i-j|} + e^{-cR}).
		\end{flalign*}
	\end{lem} 
	
	We next define events on which a symmetric matrix has bounded entries and eigenvalues, and on which its eigenvalues are separated.

	\begin{definition}[Bounded and separated events] 
		
		\label{adelta}
		
		Fix real numbers $A, \delta > 0$; let $\mathscr{I}$ denote an index set; and let $\bm{M} = [M_{ij}] \in \SymMat_{\mathscr{I}}$. Define the events
		\begin{flalign*}
			& \mathsf{BND}_{\bm{M}} (A) = \Bigg\{ \displaystyle\max_{i,j \in \mathscr{I}} |M_{ij}| \le A \Bigg\} \cap \Bigg\{ \displaystyle\max_{\lambda \in \eig \bm{M}} |\lambda| \le A \Bigg\}; \\
			& \mathsf{SEP}_{\bm{M}} (\delta) = \Bigg\{ \displaystyle\min_{\substack{\nu, \nu' \in \eig \bm{M} \\ \nu \ne \nu'}} |\nu - \nu'| \ge \delta \Bigg\}. 
		\end{flalign*}
	\end{definition} 
	
	\begin{rem} 
		
		\label{mmj} 
		
		Adopt the notation of \Cref{adelta}; let $\mathscr{J} \subseteq \mathscr{I}$ be some index set; and denote $\bm{M}' = \bm{M}^{(\mathscr{J})}$ (recall \Cref{Notation}). By the Weyl interlacing inequality, we have $\mathsf{BND}_{\bm{M}} (A) \subseteq \mathsf{BND}_{\bm{M}'} (A)$.
		
	\end{rem} 
	
	If $\bm{L}(t)$ is the Lax matrix for a Toda lattice initialized under thermal equilibrium (recall \Cref{mubeta2} and \Cref{mubeta}), then the following lemma bounds its entries and eigenvalues with high probability. Its proof will be given in \Cref{ProofL} below.

	\begin{lem}
		
		\label{l0eigenvalues}
		
		There exists a constant $c > 0$ such that the following holds. Fix an integer $N \ge 1$; let $(\bm{a}(t);\bm{b}(t))$ denote the Flaschka variables for a Toda lattice \eqref{derivativepa}, either in the open case on $\llbracket N_1, N_2 \rrbracket$ (as in \Cref{Open}) or in the periodic one on $\mathbb{T}_N$ (as in \Cref{ModelPeriodic}). Denote the associated Lax matrix by $\bm{L}(t) = [L_{ij}(t)]$, as in \Cref{matrixl} for the open case (on $\llbracket N_1, N_2 \rrbracket$) and \Cref{matrixl2} for the periodic one (on $\mathbb{T}_N$). Assume that $( \bm{a}(0); \bm{b}(0) )$ is sampled under $\mu_{\beta, \theta; N-1,N}$ in the open case (on $\llbracket N_1, N_2 \rrbracket$), and under $\mu_{\beta,\theta;N}$ in the open case (on $\mathbb{T}_N$). Then, for any real number $A \ge 1$, 
		\begin{flalign*} 
			\mathbb{P} \Bigg[ \bigcap_{t \in \mathbb{R}_{\ge 0}} \mathsf{BND}_{\bm{L}(t)} (A) \Bigg] \ge 1 - c^{-1} N e^{-cA^2}.
		\end{flalign*} 
		
	\end{lem}

	The following result indicates that the off-diagonal entries in the resolvent of a Lax matrix, of the open Toda lattice under the thermal equilibrium, decay exponentially. The bound \eqref{gijs} is a special case of \cite[Theorem 4]{LRBM} (and the remark following it) when $z \in \mathbb{R}$ is real; together with \cite[Equation (B.8)]{FFCL}, this implies that it continues to hold for any complex $z \in \mathbb{C}$.
	
	\begin{lem}[{\cite{LRBM,FFCL}}]
		
		\label{gijexponential}
		
		Adopt \Cref{lbetaeta}. For any real number $s \in (0, 1)$, there exists a constant $c = c(s) > 0$ such that the following holds. For any $z \in \mathbb{C}$, denote $\bm{G} (z) = [G_{ij} (z) ] = (\bm{L} - z)^{-1}$. We have  
		\begin{flalign}
			\label{gijs}
			\displaystyle\sup_{z \in \mathbb{R}} \mathbb{E} \big[ |G_{ij} (z) |^s \big] \le c^{-1} e^{-c|i-j|}.
		\end{flalign}  
	\end{lem}
	
	The below corollary, which is due to \cite[Theorem A.1]{FFCL} (using \Cref{gijexponential} to verify its hypotheses), will quickly imply that eigenvectors of $\bm{L}$ are localized (see \Cref{bijectionl} below).			
	
	\begin{cor}[{\cite[Theorem A.1]{FFCL}}]
		
		\label{uijexponential} 
		
		Adopt \Cref{lbetaeta}. There exists a constant  $c > 0$ so that 
		\begin{flalign*}
			\displaystyle\max_{k \in \llbracket 1, N \rrbracket} \displaystyle\max_{i, j \in \llbracket N_1, N_2 \rrbracket} \mathbb{E} \big[ | u_k (i) u_k (j) | \big] \le c^{-1} e^{-c|i-j|}.
		\end{flalign*}
	\end{cor}

	The following lemma states that, with high probability, no distinct eigenvalues of $\bm{L}$ can be too close (recall \Cref{adelta}). It is a quick consequence of the Minami estimates \cite{LFSM}; we provide its proof in \Cref{ProofL} below. 
	
	\begin{lem}
		
		\label{eigenvalues0}
		
		Adopt \Cref{lbetaeta}. There exists a constant $c > 0$ such that $\mathbb{P} [ \mathsf{SEP}_{\bm{L}} (\delta) ] \ge 1 - c^{-1} (\delta N^3 + e^{-cN^2})$ holds for any real number $\delta > 0$.
		
	\end{lem}

	\section{Comparison Estimates}
	
	\label{LCompare} 	
	
	\subsection{Comparisons for the Toda Lattice on Different Domains}
	
	\label{DomainCompare}

	In this section we compare Toda lattices (through their Flaschka variables \eqref{abr}) on different domains, that initially coincide on a subdomain. Such estimates are variants of the Lieb--Robinson estimates for quantum systems \cite{FGVQSS}; analogous bounds (in a slightly different form from what we use, but with similar proofs) have also appeared for classical ones, including the Toda lattice \cite{TES,BL}.
	
	We begin with the below proposition that compares two open Toda lattices on different intervals, assuming that they initially coincide on a subinterval of both.

	\begin{prop} 
		
		\label{aabbk}
		
		Let $\tilde{N}_1 \le N_1 \le N_2 \le \tilde{N}_2$ be integers; set $\tilde{N} = \tilde{N}_2 - \tilde{N}_1 + 1$ and $N = N_2 - N_1 + 1$. For each $t \in \mathbb{R}_{\ge 0}$, fix $\tilde{N}$-tuples $\tilde{\bm{a}}(t), \tilde{\bm{b}} (t) \in \mathbb{R}^{\tilde{N}}$ and $N$-tuples $\bm{a}(t), \bm{b}(t) \in \mathbb{R}^N$, indexed as 
		\begin{flalign*} 
			& \tilde{\bm{a}} (t) = ( \tilde{a}_{\tilde{N}_1} (t), \tilde{a}_{\tilde{N}_1+1} (t), \ldots , \tilde{a}_{\tilde{N}_2} (t) ); \quad \tilde{\bm{b}} (t) = ( \tilde{b}_{\tilde{N}_1} (t), \tilde{b}_{\tilde{N}_1+1} (t), \ldots , \tilde{b}_{\tilde{N}_2} (t) ); \\
			&\bm{a} (t) = ( a_{N_1} (t), a_{N_1+1} (t), \ldots , a_{N_2} (t) ); \quad \bm{b}(t) = ( b_{N_1} (t), b_{N_1+1}(t), \ldots , b_{N_2} (t) ).
		\end{flalign*} 
		
		\noindent For each $s \in \mathbb{R}_{\ge 0}$, also set $\tilde{a}_i (s) = 0 = \tilde{b}_i (s)$ if $i \in \mathbb{Z} \setminus \llbracket \tilde{N}_1, \tilde{N}_2 - 1\rrbracket$, and set $a_i (s) = 0 = b_i (s)$ if $i \in \mathbb{Z} \setminus \llbracket N_1, N_2 - 1 \rrbracket$. Assume $( \tilde{\bm{a}} (t); \tilde{\bm{b}} (t) )$ satisfies \eqref{derivativepa} for each $(j, t) \in \llbracket \tilde{N}_1, \tilde{N}_2 \rrbracket \times \mathbb{R}_{\ge 0}$, and $( \bm{a}(t), \bm{b}(t) )$ satisfies \eqref{derivativepa} for each $(j, t) \in \llbracket N_1, N_2 \rrbracket \times \mathbb{R}$. For any integers $I \le J$ and real number $t \ge 0$, let 
		\begin{flalign}
			\label{ghijt}
			\begin{aligned}
				& G_{\llbracket I, J \rrbracket} (t) = \displaystyle\sup_{s \in [0, t]} \bigg( \displaystyle\max_{i \in \llbracket I, J \rrbracket} \big| a_i (s) - \tilde{a}_i (s) \big| + \displaystyle\max_{i \in \llbracket I, J \rrbracket} \big| b_i (s) - \tilde{b}_i (s) \big| \bigg); \\
				& H_{\llbracket I, J \rrbracket} (t) = 6 \cdot \displaystyle\sup_{s \in [0,t]} \bigg(  \displaystyle\max_{i \in \llbracket I, J \rrbracket}  | a_i (s) | +  \displaystyle\max_{i \in \llbracket I, J \rrbracket} | \tilde{a}_i (s) | +  \displaystyle\max_{i \in \llbracket I, J \rrbracket} | b_i (s) | +  \displaystyle\max_{i \in \llbracket I, J \rrbracket} | \tilde{b}_i (s) | \bigg).
			\end{aligned}
		\end{flalign}
		
		\noindent Now let $K \ge 1$ and $N_1' \le N_2'$ be integers such that $N_1 \le N_1'  \le N_2' \le N_2$ and $N_1' + K \le N_2' - K$. If $a_j (0) = \tilde{a}_j (0)$ and $b_j (0) = \tilde{b}_j (0)$ for each $j \in \llbracket N_1', N_2'  \rrbracket$, then for any $T \in \mathbb{R}_{\ge 0}$ we have
		\begin{flalign}
			\label{gn12hn12t}
			G_{\llbracket N_1' + K, N_2' - K \rrbracket} (T) \le \displaystyle\frac{T^K}{K!} \cdot  G_{\llbracket N_1', N_2' \rrbracket} (T) \cdot \displaystyle\prod_{i=0}^{K-1} H_{\llbracket N_1' + i, N_2' - i \rrbracket} (T).
		\end{flalign}
		
	\end{prop}

	To establish this proposition, we first require the following variant of the Gr\"{o}nwall inequality. 
	
	\begin{lem} 
		
		\label{hkgk} 
		
		Let $T \ge 0$ be a real number; $K \ge 1$ be an integer; and let $g_k, h_k : \mathbb{R}_{\ge 0} \rightarrow \mathbb{R}_{\ge 0}$ be nondecreasing functions, for each $k \in \llbracket 0, K \rrbracket$. Assume for each $(k, t) \in \llbracket 0, K-1 \rrbracket \times [0, T]$ that
		\begin{flalign}
			\label{0gkt}
			g_k (t) \le h_{k+1} (t) \displaystyle\int_0^t g_{k+1}(s) ds.
		\end{flalign}
		
		\noindent Then, for any $(j, t) \in \llbracket 0, K \rrbracket \times [0, T]$, we have 
		\begin{flalign}
			\label{gkj}
			g_j (t) \le \displaystyle\frac{t^{K-j}}{(K-j)!} \cdot g_K (t) \cdot \displaystyle\prod_{i=j+1}^{K} h_i (t).
		\end{flalign} 
		
	\end{lem} 
	
	\begin{proof}
		
		We verify \eqref{gkj} by induction on $K-j$. Since \eqref{gkj} holds at $K-j=0$, let us fix $j_0 \in \llbracket 0, K-1 \rrbracket$ and prove that \eqref{gkj} holds for $K-j = K-j_0$, assuming it holds whenever $K-j \le K-j_0-1$. This follows from the estimates
		\begin{flalign*}
			g_{j_0} (t) & \le h_{j_0+1} (t) \displaystyle\int_0^t g_{j_0+1} (s) ds \\
			& \le h_{j_0+1} (t)  \displaystyle\int_0^t \displaystyle\frac{s^{K-j_0-1}}{(K-j_0-1)!} \cdot g_K (s) \cdot \displaystyle\prod_{i=j_0+2}^{K} h_i (s) ds \\
			& \le g_K (t) \cdot \displaystyle\prod_{i=j_0+1}^K h_i (t) \cdot \displaystyle\int_0^t \displaystyle\frac{s^{K-j_0-1}}{(K-j_0-1)!} ds \\
			& = \displaystyle\frac{t^{K-j_0}}{(K-j_0)!} \cdot g_K (t) \cdot \displaystyle\prod_{i=j_0+1}^K h_i (t),
		\end{flalign*} 
		
		\noindent where the first bound holds by the $k=j_0$ case of \eqref{0gkt}; the second by the inductive hypothesis; the third by the fact that the $g_i$ and $h_i$ are nondecreasing; and the fourth by performing the integration.  
	\end{proof} 	 
	
	\begin{proof}[Proof of \Cref{aabbk}] 
		
		Observe from \eqref{derivativepa} that, for any $(i, s) \in \llbracket \tilde{N}_1, \tilde{N}_2 \rrbracket \times \mathbb{R}_{\ge 0}$, we have 
		\begin{flalign*}
			& \partial_t (a_i - \tilde{a}_i) = \displaystyle\frac{1}{2} \cdot (b_i - b_{i+1}) (a_i - \tilde{a}_i) + \displaystyle\frac{\tilde{a}_i}{2} \cdot \big( (b_i - \tilde{b}_i) - (b_{i+1} - \tilde{b}_{i+1}) \big) ; \\
			& \partial_t (b_i - \tilde{b}_i) = (a_{i-1} - \tilde{a}_{i-1} ) (a_{i-1} + \tilde{a}_{i-1}) - (a_i - \tilde{a}_i)(a_i + \tilde{a}_i),
		\end{flalign*}
		
		\noindent where we have abbreviated $(\tilde{a}_i, a_i; \tilde{b}_i, b_i) = ( \tilde{a}_i (s), a_i (s); \tilde{b}_i (s), b_i (s) )$. Using \eqref{ghijt}, it follows that 
		\begin{flalign*}
			|\partial_t a_i (s)  - \partial_t \tilde{a}_i (s)| & = \displaystyle\frac{1}{2} \cdot |b_i - b_{i+1}| \cdot G_i (s) + |\tilde{a}_i | \cdot G_{\llbracket i,i+1 \rrbracket} (s) \\
			& \le \displaystyle\frac{1}{3} \cdot H_{\llbracket i,i+1 \rrbracket} (s) \cdot G_{\llbracket i,i+1 \rrbracket} (s),
		\end{flalign*} 
		
		\noindent and 
		\begin{flalign*} 
			|\partial_t b_i (s) - \partial_t \tilde{b}_i (s)| & = |a_{i-1}  + \tilde{a}_{i-1} | \cdot G_{i-1} (s) + |a_i  + \tilde{a}_i | \cdot G_i (s) \\
			& \le \displaystyle\frac{2}{3} \cdot H_{\llbracket i-1,i \rrbracket} (s) \cdot G_{\llbracket i-1, i \rrbracket} (s).
		\end{flalign*}  
		
		\noindent Summing these two bounds, fixing $t_0 \in \mathbb{R}_{\ge 0}$, integrating over $s \in [0, t_0]$, and using \eqref{ghijt} yields 
		\begin{flalign*}
			& \big| a_i (t_0) - \tilde{a}_i (t_0) \big| + \big| b_i (t_0) - \tilde{b}_i (t_0) \big| \le H_{\llbracket i-1, i+1 \rrbracket} (t_0) \displaystyle\int_0^{t_0} G_{\llbracket i-1, i+1 \rrbracket} (s) ds,
		\end{flalign*}
		
		\noindent for any $i \in \llbracket N_1', N_2' \rrbracket$, where we have also used that $a_i (0) = \tilde{a}_i (0)$ and $b_i (0) = \tilde{b}_i (0)$ for such $i$. Taking the supremum over $t_0 \in [0, t]$, we deduce that \eqref{0gkt} holds for any $(k, t) \in \llbracket 0, K-1 \rrbracket \times \mathbb{R}_{\ge 0}$, if we set 
		\begin{flalign*}
			g_k (t) = G_{\llbracket N_1' + K - k, N_2' - K + k \rrbracket} (t); \qquad h_k (t) = H_{\llbracket N_1' + K - k, N_2' - K + k \rrbracket} (t). 
		\end{flalign*} 
		
		\noindent The lemma then follows from the $(j,t) = (0, T)$ case of \Cref{hkgk}.
	\end{proof}

	The next proposition compares a Toda lattice $(\bm{a}(t); \bm{b}(t) )$ on a torus to one $( \tilde{\bm{a}}(t), \tilde{\bm{b}}(t) )$ on an interval (of the same size), if the two coincide on a subinterval of their domains. Its proof is entirely analogous to that of \Cref{aabbk} and is therefore omitted. 
	
	\begin{prop} 
		
		\label{aabbk2}
		
		Let $N_1 \le N_2$ be integers; set $N = N_2 - N_1 + 1$. For each real number $t \in \mathbb{R}_{\ge 0}$, fix $N$-tuples $\tilde{\bm{a}}(t), \bm{a}(t), \tilde{\bm{b}} (t), \bm{b}(t) \in \mathbb{R}^N$, indexed as
		\begin{flalign*}
			& \tilde{\bm{a}}(t) = ( \tilde{a}_{N_1}(t), \tilde{a}_{N_1+1}(t), \ldots , \tilde{a}_{N_2}(t) ); \qquad \tilde{\bm{b}}(t) = ( \tilde{b}_{N_1}(t), \tilde{b}_{N_1+1}(t), \ldots , \tilde{b}_{N_2}(t) ); \\
			& \bm{a}(t) = ( a_0 (t), a_1 (t), \ldots , a_{N-1}(t) ); \qquad \quad \bm{b} (t) = ( b_0 (t), b_1 (t), \ldots , b_{N-1}(t) ).
		\end{flalign*}
		
		\noindent For each $s \in \mathbb{R}_{\ge 0}$, also set $\tilde{a}_i (s) = 0 = \tilde{b}_i (s)$ if $i \notin \llbracket \tilde{N}_1, \tilde{N}_2 - 1 \rrbracket$; additionally set $a_i (s) = a_{i+N} (s)$ and $b_i (s) = b_{i+N} (s)$ for each $(i, s) \in \mathbb{Z} \times \mathbb{R}_{\ge 0}$. Assume that $( \tilde{\bm{a}}(t); \tilde{\bm{b}}(t) )$ satisfies \eqref{derivativepa} for each $(j, t) \in \llbracket N_1, N_2 \rrbracket \times \mathbb{R}_{\ge 0}$ and that $( \bm{a}(t); \bm{b}(t) )$ satisfies \eqref{derivativepa} for each $(j, t) \in \llbracket 0, N-1 \rrbracket \times \mathbb{R}_{\ge 0}$. Under this notation, for any $I \le J$ and $t \in \mathbb{R}_{\ge 0}$, define $G_{\llbracket I, J \rrbracket} (t)$ and $H_{\llbracket I, J \rrbracket} (t)$ as in \eqref{ghijt}. 
		
		Now let $K \ge 1$ and $N_1' \le N_2'$ be integers such that $N_1 \le N_1' \le N_2' \le N_2$ and $N_1' + K \le N_2' - K$. If $a_j (0) = \tilde{a}_j (0)$ and $b_j (0) = \tilde{b}_j (0)$ for each $j \in \llbracket N_1', N_2' \rrbracket$, then \eqref{gn12hn12t} holds for any $T \in \mathbb{R}_{\ge 0}$.
		
	\end{prop} 
	
	\subsection{Approximate Thermal Equilibrium for the Open Toda Lattice}
	
	\label{ApproximateMu}
	
	Unlike for the periodic Toda lattice, thermal equilibrium $\mu_{\beta,\theta;N-1,N}$ (from \Cref{mubeta2}) is not an invariant measure for the open Toda lattice. In this section we establish the below proposition indicating that it, in a certain sense, still ``approximately is,'' if time scale is much shorter than the domain size. Its proof uses \Cref{aabbk2} to couple the periodic Toda lattice to the closed one.

	\begin{prop} 
		
		\label{ltl0} 
		
		Adopt \Cref{lbetaeta}, and fix $t \in [0, T]$. There exists a random matrix $\bm{M} = [M_{ij}] \in \SymMat_{\llbracket N_1, N_2 \rrbracket}$, whose law coincides with that of $\bm{L}(0)$, such that the following holds with overwhelming probability. For any real number $K \ge T \log N $, we have that 
		\begin{flalign}
			\label{estimatelm} 
			\displaystyle\max_{i,j \in \llbracket N_1+K, N_2-K \rrbracket} | L_{ij} (t) - M_{ij} | \le e^{-K/5}.
		\end{flalign}
		
	\end{prop}
	
	To establish \Cref{ltl0}, we first require the following quick consequence of \Cref{aabbk} and \Cref{aabbk2}.

	\begin{lem} 
		
		\label{a2p2} 
		
		Adopt the notation and assumptions of either \Cref{aabbk} or \Cref{aabbk2}, and suppose that $T \ge 1$. Fix a real number $A \ge 1$, and assume that $K \ge 200AT$ and that $| a_i (0) | + | \tilde{a}_i (0) | + | b_i (0) | + | \tilde{b}_i (0) | \le A$ for each $i \in \mathbb{Z}$. Then, $G_{\llbracket N_1' + K, N_2'-K \rrbracket} (T) \le e^{-K/4}$.
		
	\end{lem}
	
	\begin{proof}
		
		By \Cref{abltt}, we have 
		\begin{flalign}
			\label{gha} 
			G_{\llbracket N_1', N_2' \rrbracket} (T) \le H_{\llbracket N_1', N_2' \rrbracket} (T) & \le 36 \cdot \displaystyle\max_{i \in \mathbb{Z}} \big( | a_i (0) | + | b_i (0) | + | \tilde{a}_i (0) | + | \tilde{b}_i (0) | \big) \\ 
			& \le 36 A.
		\end{flalign} 
		
		\noindent Thus,  
		\begin{flalign*}
			G_{\llbracket N_1' + K, N_2' - K \rrbracket} (T) \le \displaystyle\frac{T^K}{K!} (36A)^{K+1} & \le 36A \bigg( \displaystyle\frac{36e AT}{K} \bigg)^K \\ 
			& \le 36A \cdot 2^{-K} \le 2^{-K/2} \le e^{-K/4}.
		\end{flalign*}
		
		\noindent where in the first inequality we used \Cref{aabbk} with \eqref{gha}; in the second we used the fact that $K! \ge (e^{-1} K)^K$; in the third we used the fact that $36eATK^{-1} \le 100AT K^{-1} \le 1 / 2$; in the fourth we used the fact that $2^{-K/2} \le 2^{-100AT} \le (36A)^{-1}$ (as $A, T \ge 1$); and in the fifth we used the fact that $e \le 4$. This establishes the lemma.
	\end{proof}

	\begin{proof}[Proof of \Cref{ltl0}]

		We will compare the open Lax matrix $\bm{L}(t)$ to a periodic one $\bm{R}(t)$. To prescribe initial data for the latter (through its Flaschka variables), define $( \bm{\mathfrak{a}} (0), \bm{\mathfrak{b}} (0) )$, where $\bm{\mathfrak{a}} (0) =  ( \mathfrak{a}_j (0) )$ and $\bm{\mathfrak{b}} (0) = ( \mathfrak{b}_j (0) )$ for $j \in \mathbb{Z}$, as follows. First, for each $i \in \llbracket N_1, N_2-1 \rrbracket$, set 
		\begin{flalign*} 
			\mathfrak{a}_i (0) = L_{i,i+1}(0) = a_i (0), \qquad \text{and} \qquad \mathfrak{b}_i (0) = L_{i,i} (0) = b_i (0).
		\end{flalign*} 
		
		\noindent Letting $\mathfrak{r} \in \mathbb{R}_{>0}$ denote a Gamma random variable, independent from $\bm{L}(0)$, with density $\mathbb{P} [\mathfrak{r} \in (r,r+dr)] = 2 \beta^{\theta} \cdot \Gamma(\theta)^{-1} \cdot r^{2\theta-1} e^{-\beta r^2} dr$, further set $\mathfrak{a}_{N_2} = \mathfrak{r}$ and $\mathfrak{b}_{N_2} = L_{N_2,N_2} (0) = b_{N_2} (0)$. We then extend these Flaschka variables periodically, by imposing $( \mathfrak{a}_{j+N} (0); \mathfrak{b}_{j+N} (0) ) = ( \mathfrak{a}_j (0); \mathfrak{b}_j (0) )$ for each $j \in \mathbb{Z}$. In this way, $( \mathfrak{a}_j (0); \mathfrak{b}_j (0) )$ over $j \in \llbracket N_1, N_2 \rrbracket$ is sampled under $\mu_{\beta,\theta;N}$ from \Cref{mubeta}. 
		
		Next, let $(\bm{\mathfrak{a}} (t); \bm{\mathfrak{b}}(t) )$, where $\bm{\mathfrak{a}}(t) = ( \mathfrak{a}_j (t) )$ and $\bm{\mathfrak{b}}(t) = ( \mathfrak{b}_j (t) )$, denote the solution to the periodic Toda lattice \eqref{derivativepa} on the torus $\mathbb{T}_N$, with initial data $( \bm{\mathfrak{a}} (0); \bm{\mathfrak{b}}(0) )$. Here, $(j, t)$ ranges over $\mathbb{Z} \times \mathbb{R}$, imposing the periodicity constraint $( \mathfrak{a}_{j+N} (t); \mathfrak{b}_{j+N} (t) ) = ( \mathfrak{a}_j (t); \mathfrak{b}_j (t) )$. It follows from \Cref{betathetainvariant} that, for each $t \in \mathbb{R}_{\ge 0}$, 
		\begin{flalign}
			\label{atbta0b02}
			\text{the law of} \quad ( \bm{\mathfrak{a}} (t); \bm{\mathfrak{b}} (t) ) \quad \text{coincides with that of} \quad ( \bm{\mathfrak{a}} (0); \bm{\mathfrak{b}} (0) ).
		\end{flalign}
		
		\noindent Define the associated Lax matrix\footnote{This convention is slightly different from the one in \Cref{matrixl2} (where the rows and columns are indexed by $\llbracket 0, N-1 \rrbracket$ instead of $\llbracket N_1, N_2 \rrbracket$); we use it for notational convenience.} $\bm{R}(t) = [ R_{ij} (t) ] \in \SymMat_{\llbracket N_1, N_2 \rrbracket}$, by setting $R_{jj} (t) = \mathfrak{b}_j (t)$ and $R_{j,j+1} (t) = R_{j+1,j} (t) = \mathfrak{a}_j (t)$ for each $j \in \llbracket N_1, N_2 \rrbracket$, where we have denoted $R_{N_2+1,N_2} (t) = R_{N_1, N_2} (t)$ and $R_{N_2, N_2+1} (t) = R_{N_2, N_1} (t)$. If $(i, j)$ is not of the above form, then we set $R_{ij} (t)= 0$.	
		
		Now define $\bm{M} = [M_{ij}] \in \SymMat_{\llbracket N_1, N_2 \rrbracket}$ by setting the $(N_1, N_2)$-entry and $(N_2, N_1)$-entry of $\bm{R}(T)$ to $0$, namely $M_{ij} = \mathbbm{1}_{(i,j) \ne (N_1, N_2)} \cdot \mathbbm{1}_{(i,j) \ne (N_2,N_1)} \cdot R_{ij} (T)$. By \eqref{atbta0b02}, the fact that $( \bm{\mathfrak{a}}(0); \bm{\mathfrak{b}}(0) )$ is sampled under $\mu_{\beta,\theta;N}$; and \Cref{lbetaeta}, $\bm{M}$ has the same law as $\bm{L}(0)$. 
		
		It therefore remains to confirm that \eqref{estimatelm} holds with high probability. To that end, recalling \Cref{adelta}, define the event $\mathsf{E}_1 = \bigcap_{s \ge 0} \mathsf{BND}_{\bm{R}(s)} (\log N / 800 )$. By \Cref{l0eigenvalues}, $\mathsf{E}_1$ is overwhelmingly probable, so we will restrict to $\mathsf{E}_1$ in what follows. 
		
		Now we apply \Cref{aabbk2} and \Cref{a2p2}. More specifically, since $L_{ij}(0) = R_{ij} (0)$ for $i, j \in \llbracket N_1 + 1, N_2 - 1 \rrbracket$, applying \Cref{a2p2} (with the $(A; I, J; K)$ there equal to $(\log N / 200; N_1+1, N_2-1; K-1)$ here, using the fact that $(K-1) / 4 \le K / 5$ for sufficiently large $N$, and our restriction to $\mathsf{E}_1$), yields \eqref{estimatelm}, since $M_{ij} = R_{ij} (t)$ for $i, j \in \llbracket N_1 + K, N_2 - K \rrbracket$.		
	\end{proof}

	\subsection{Infinite Volume Limit} 
	
	\label{RLimit} 
	
	Recall that our previous descriptions of the Toda lattice involved a finite number of variables, and were thus defined on a finite domain. In this section we explain conditions under which solutions of the Toda lattice on infinite domains can be realized as limits of those on finite domains. These conditions are described through the below assumption.
	
	\begin{assumption}
		
		\label{initialrp}
		
		Fix real numbers $R > 1 > \mathfrak{p} \ge 0$. For each real number $t \in \mathbb{R}_{\ge 0}$ and integers $n \ge 1 \ge m$, let 
		\begin{flalign*}  
			& \bm{a}^{[m,n]} (t) = ( a_j^{[m,n]} (t) )_{j \in \llbracket m, n \rrbracket} \in \mathbb{R}^{m+n+1}; \\
			& \bm{a}^{(n)} (t) = ( a_j^{(n)} (t) )_{j \in \llbracket -n, n \rrbracket} \in \mathbb{R}^{2n+1}; \quad  \bm{a} = ( a_j )_{j \in \mathbb{Z}}; \\
			& \bm{b}^{[m,n]} (t) = ( b_j^{[m,n]} (t) )_{j \in \llbracket m, n \rrbracket} \in \mathbb{R}^{m+n+1}; \\ 
			& \bm{b}^{(n)} (t) = ( b_j^{(n)} (t) )_{j \in \llbracket -n, n \rrbracket} \in \mathbb{R}^{2n+1}; \quad  \bm{b} = ( b_j )_{j \in \mathbb{Z}},
		\end{flalign*}
		
		\noindent be $(m+n+1)$-tuples, $(2n+1)$-tuples, and infinite sequences of real numbers. Assume that $( \bm{a}^{[m,n]} (t); \bm{b}^{[m,n]} (t) )$ is a solution to the open Toda lattice \eqref{derivativepa} on the interval $\llbracket m, n \rrbracket$, and that  $( \bm{a}^{(n)} (t); \bm{b}^{(n)} (t) )$ is a solution to the periodic Toda lattice \eqref{derivativepa} on the torus $\mathbb{T}_{2n+1}$, which we identify with $\llbracket -n, n \rrbracket$. Further assume for any integers $n \ge 1 \ge m$ and $j \in \mathbb{Z}$ that 
		\begin{flalign}
			\label{estimaterjp}
			\begin{aligned} 
				& | a_j^{[m,n]} (0) | + | a_j^{[m,n]} (0) | \le R ( |j|+1)^{\mathfrak{p}}; \\
				& | a_j^{(n)} (0) | + | b_j^{(n)} (0) | \le R ( |j|+1 )^{\mathfrak{p}}.
			\end{aligned} 
		\end{flalign}
		
		\noindent Also assume $(a_j^{(n)} (0), b_j^{(n)} (0)) = (a_j, b_j)$ whenever $j \in \llbracket -n, n \rrbracket$; that $b_j^{[m,n]} (0) = b_j$ whenever $j \in \llbracket m,n \rrbracket$; and that $a_j^{[m,n]} (0) = a_j$ whenever $j \in \llbracket m, n-1 \rrbracket$ (and $a_n^{[m,n]} (0) = 0$).		
	\end{assumption}

	\begin{prop} 
		
		\label{infiniteab} 
		
		Adopt \Cref{initialrp}. Then, for each $(j, t) \in \mathbb{Z} \times \mathbb{R}_{\ge 0}$, there exist real numbers $a_j(t)$ and $b_j (t)$ such that
		\begin{flalign}
			\label{alimitblimit} 
			\displaystyle\lim_{n \rightarrow \infty} a_j^{[-n,n]} (t) = a_j (t) = \displaystyle\lim_{n \rightarrow \infty} a_j^{(n)} (t); \qquad  \displaystyle\lim_{n \rightarrow \infty} b_j^{[-n,n]} (t) = b_j (t) = \displaystyle\lim_{n \rightarrow \infty} b_j^{(n)} (t).
		\end{flalign}
		
		\noindent Moreover, denoting $\bm{a} (t) = ( a_j (t) )_{j \in \mathbb{Z}}$ and $\bm{b}(t) = ( b_j (t) )_{j \in \mathbb{Z}}$, we have that $( \bm{a}(t); \bm{b}(t) )$ solves the Toda lattice \eqref{derivativepa} for each $(j, t) \in \mathbb{Z} \times \mathbb{R}$. 
		
	\end{prop} 
	
	\begin{proof} 
		
		Observe by \eqref{estimaterjp} and \Cref{abltt} that, for any integers $n \ge 1 \ge m$ and real number $t \in \mathbb{R}_{\ge 0}$,   
		\begin{flalign}
			\label{ajpn}
			& \displaystyle\sup_{j \in \llbracket -n, n \rrbracket} \big( | a_j^{[-n,n]} (t) | + | b_j^{[-n,n]} (t) | + | a_j^{(n)} (t) | + | b_j^{(n)} (t) | \big) \le 12 R (n+1)^{\mathfrak{p}}.
		\end{flalign} 
		
		\noindent Now let $K \ge 1$ be an integer, and let $N \ge 1$ be an integer. By \Cref{aabbk2}, with the last statement in \Cref{initialrp} (and the fact that $72R(n+1)^{\mathfrak{p}} \le 144 Rn^{\mathfrak{p}}$ for $n \ge 1$), we have
		\begin{flalign}
			\label{ajnbjn}
			\big| a_j^{(N)} (t) - a_j^{[-N, N]} (t) \big| + \big| b_j^{(N)} (t) - b_j^{[-N, N]} (t) \big| \le \displaystyle\frac{t^K}{K!} \cdot (144RN^{\mathfrak{p}})^{K+1},
		\end{flalign}
		
		\noindent for any $(j, t) \in \mathbb{Z} \times \mathbb{R}$ with $|j| + K \le N$. Thus, applying \eqref{ajnbjn} at $(N, K) = \big( n, \lfloor n / 2 \rfloor)$, we find 
		\begin{flalign*}
			\displaystyle\lim_{n \rightarrow \infty} \big(  \big| a_j^{(n)} (t) - a_j^{[-n, n]} (& t) \big| + \big| b_j^{(n)} (t) - b_j^{[-n, n]} (t) \big| \big) \\
			&  \le \displaystyle\lim_{n \rightarrow \infty} \displaystyle\frac{t^{\lfloor n/2 \rfloor}}{\lfloor n/2 \rfloor!} \cdot (144Rn^{\mathfrak{p}})^{n/2+1}= 0,
		\end{flalign*}
		
		\noindent where in the last equality we used the fact that $\mathfrak{p} < 1$. Hence, it suffices to show that the first limits in both of the pairs of equalities in \eqref{alimitblimit} hold. 
		
		To that end, fix an integer $J \ge 1$ and a real number $T \ge 0$. For $(j, t) \in \llbracket -J, J \rrbracket \times [0, T]$, we first bound $( a_j^{[-n,n]} (t), b_j^{[-n,n]} (t) )$ independently of $n$. To do this, observe (as used to obtain \eqref{ajnbjn}) by \Cref{aabbk} with the last statement in \Cref{initialrp} that, for any integers $N \ge N' \ge 1$,
		\begin{flalign}
			\label{ajnnnn}
			\big| a_j^{[-N, N]} (t) - a_j^{[-N', N']} (t) \big| + \big| b_j^{[N, N]} (t) - b_j^{[-N', N']} (t) \big| \le \displaystyle\frac{T^K}{K!} \cdot ( 144RN^{\mathfrak{p}})^{K+1}, 
		\end{flalign}
		
		\noindent for any $(j,t) \in \mathbb{Z} \times \mathbb{R}$ with $|j| + K \le N'$. Now let $N_0 > N_1 > \cdots > N_r$ be integers such that $N_0 = N$, such that $N_i / 16 \le N_{i+1} \le N_i / 4$ for each $i \in \llbracket 0, r-1 \rrbracket$, and such that $N_r \in [4J, 16J]$. Applying \eqref{ajnnnn} with the $(N, N'; K)$ there equal to $(N_i, N_{i+1}; N_{i+1} / 2)$ here, we find that there exists a constant $C_1 = C_1 (\mathfrak{p}, R, T) > 1$ such that 
		\begin{flalign*}
			\big| a_j^{[-N_i,N_i]} (t) - a_j^{[-N_{i+1},N_{i+1}]} (t ) \big| + \big| & b_j^{[-N_i, N_i]} (t) - b_j^{[-N_{i+1},N_{i+1}]} (t) \big| \\
			& \le \displaystyle\frac{T^K}{K!} \cdot (144RN_i^{\mathfrak{p}})^{K+1} \le C_1 e^{-N_i},
		\end{flalign*} 
		
		\noindent where in the last bound we used the facts that $K = N_{i+1} / 2 \ge N_i / 32$ and $\mathfrak{p} < 1$. Hence, for any integer $m \in \llbracket 1, r \rrbracket$, we have
		\begin{flalign}
			\label{ajnmn}
			\begin{aligned}
				\big| & a_j^{[-N,N]} (t) - a_j^{[-N_m, N_m]} (t) \big| + \big| b_j^{[-N, N]} (t) - b_j^{[-N_m, N_m]} (t) \big| \\
				& \le \displaystyle\sum_{i=0}^{m-1} \Big( \big| a_j^{[-N_i, N_i]} (t) - a_j^{[-N_{i+1}, N_{i+1}]} (t) \big| + \big| b_j^{[-N_i, N_i]} (t) - b_j^{[-N_{i+1}, N_{i+1}]} (t) \big| \Big) \\
				& \le C_1 \displaystyle\sum_{i=0}^{m-1} e^{-N_i} \le 2C_1 e^{-N_m}.
			\end{aligned}
		\end{flalign}
		
		\noindent Taking $m=r$ and using \eqref{ajpn} (with the $n$ there equal to $N_r \in [4J, 16J]$ here), it follows that 
		\begin{flalign}
			\label{jaestimate}
			\begin{aligned} 
				\displaystyle\max_{j \in \llbracket -J, J \rrbracket} \displaystyle\sup_{t \in [0, T]} \big( | a_j^{[-N, N]} (t) | + | b_j^{[-N, N]} (t) | \big) & \le 2C_1 e^{-N_r} + 12R(N_r+1)^{\mathfrak{p}} \\
				& \le 250RJ + 2C_1.
			\end{aligned}
		\end{flalign}
		
		By \eqref{ajnmn} and \eqref{ajnnnn}, $( a_j^{[-N,N]} (t) )$ and $( b_j^{[-N,N]} (t) )$ are Cauchy sequences over $N \in \mathbb{Z}_{\ge 0}$. Hence, these sequences admit unique limits $a_j (t)$ and $b_j (t)$, respectively, as $N$ tends to $\infty$. This establishes the first part of the proposition. 
		
		To establish the second, observe by \eqref{jaestimate} that there exists a constant $C_2 = C_2 (\mathfrak{p}, R, J, T) > 1$ such that $| a_j (t) | + | b_j (t) | \le C_2$ and $| a_j^{[-N,N]} (t) | + | b_j^{[-N,N]} (t) | \le C_2$, for all $(j, t) \in \llbracket -J, J \rrbracket \times [0, T]$ and $N \in \mathbb{Z}_{\ge 1}$. Since $( \bm{a}^{[-N,N]} (t); \bm{b}^{[-N,N]} (t) )$ satisfies \eqref{derivativepa}, these estimates imply that the first two $t$-derivatives of $a_j^{[-N,N]} (t)$ and $b_j^{[-N,N]} (t)$ are uniformly bounded, for all $(j, t) \in \llbracket -J, J \rrbracket \times [0, T]$. Thus, the $t$-derivatives of $a_j (t)$ and $b_j (t)$ are also uniformly bounded for all $(j, t) \in \llbracket -J, J \rrbracket \times [0, T]$. Together with the fact that $( \bm{a}^{[-N,N]} (t); \bm{b}^{[-N,N]} (t) )$ satisfy \eqref{derivativepa}, this implies that $(\bm{a}; \bm{b})$ satisfy \eqref{derivativepa}, thereby confirming the second statement of the proposition.
	\end{proof}

	The below corollary indicates that the Toda lattice on the full line $\mathbb{Z}$ at thermal equilibrium can be defined by taking a limit of periodic Toda lattices at thermal equilibrium $\mu_{\beta,\theta;N}$ (recall \Cref{mubeta}). Here, for any integer $N \ge 1$, we identify the torus $\mathbb{T}_{2N+1}$ with the interval $\llbracket -N, N \rrbracket$ (instead of $\llbracket 0, 2N \rrbracket$), which will index the Flaschka variables of the associated Toda lattice.
	
	\begin{cor} 
		
		\label{ajbjequation} 
		
		Fix $\beta, \theta \in \mathbb{R}_{>0}$. Let $a_1, a_2, \ldots , b_1 ,b_2, \ldots $ be mutually independent random variables, so that $(a_j, b_j)$ has law $\mu_{\beta, \theta; 1}$ for each $j \in \mathbb{Z}$. For each integer $N \ge 1$, let $( \bm{a}^{(N)} (t), \bm{b}^{(N)} (t) )$ denote the periodic Toda lattice \eqref{derivativepa} on $\mathbb{T}_{2N+1}$, with initial data $( \bm{a}^{(N)} (0); \bm{b}^{(N)} (0) )$ given by $\bm{a}^{(N)} (0) = (a_{-N}, a_{1-N}, \ldots , a_N)$ and $\bm{b}^{(N)} (0) = (b_{-N}, b_{1-N}, \ldots , b_N)$. 
		
		\begin{enumerate}
			\item For each $(j, t) \in \mathbb{Z} \times \mathbb{R}_{\ge 0}$, the limits $\lim_{N \rightarrow \infty} a_j^{(N)} (t) = a_j (t)$ and $\lim_{N \rightarrow \infty} b_j^{(N)} (t) = b_j (t)$ exist almost surely; are finite; and solve  \eqref{derivativepa}. 	
			\item Denoting $\bm{a}(t) = (a_j(t))_{j \in \mathbb{Z}}$ and $\bm{b}(t) = (b_j(t))_{j \in \mathbb{Z}}$, the law of $(\bm{a}(t); \bm{b}(t))$ coincides with that of $(\bm{a}(0); \bm{b}(0))$, for any $t \ge 0$.
			
		\end{enumerate}

	\end{cor}
	
	\begin{proof}
		
		The second statement of the corollary follows from the first, together with \Cref{betathetainvariant}. It therefore remains to show the first, to which end it suffices by \Cref{infiniteab} to verify \Cref{initialrp}, that is, to show that there almost surely exists a constant $R > 1$ such that $|a_j| + |b_j| \le R(|j|+1 )^{1/2}$ for all $j \in \mathbb{Z}$. To that end, observe that there exists a constant $C > 1$ such that, for each $j \in \mathbb{Z}$, we have
		\begin{flalign}
			\label{ajbj2j112}
			\begin{aligned} 
				\mathbb{P} \big[ |a_j| + |b_j| > 2 ( |j| + 1 )^{1/2} \big] & \le \mathbb{P} \big[ a_j > ( |j|+1 )^{1/2} \big] + \mathbb{P} \big[ |b_j| > ( |j|+1 )^{1/2} \big] \\
				& \le C |j|^{\theta+1} e^{-\beta |j| / 2},
			\end{aligned} 
		\end{flalign}
		
		\noindent where in the last inequality we used the fact that $(a_j, b_j)$ has law $\mu_{\beta, \theta; 1}$ given by \Cref{mubeta}. Since the sum of the right side of \eqref{ajbj2j112} over $j \in \mathbb{Z}$ is finite, the Borel--Cantelli lemma implies the almost sure existence of a constant $R > 1$ such that $|a_j| + |b_j| \le R ( |j|+1 )^{1/2}$ for all $j \in \mathbb{Z}$. 		
	\end{proof} 
	
	We can now establish \Cref{ajbjequation}. 
	
	\begin{proof}[Proof of \Cref{ajbjequation2}]
		
		This will follow from using \Cref{a2p2} to compare the open Toda lattice $(\bm{a}^{[-N,-N]}(t), \bm{b}^{[-N,-N]}(t))$ to one on the torus $\mathbb{T}_{2N+1}$ with the same initial data, and then applying \Cref{ajbjequation} (and its proof) to confirm the large $N$ limit of the latter. So let $(\breve{\bm{a}}(t), \breve{\bm{b}}(t))$, where $\breve{\bm{a}}(t) = (\breve{a}_{-N}(t), \breve{a}_{1-N}(t), \ldots , \breve{a}_N (t))$ and $\breve{\bm{b}}(t) = (\breve{b}_{-N}(t), \breve{b}_{1-N}(t), \ldots , \breve{b}_N(t))$, denote the periodic Toda lattice \eqref{derivativepa} on $\mathbb{T}_{2N+1}$, with initial data obtained by setting $\breve{a}_j (0) = a_j$ and $\breve{b}_j (0) = b_j$ for each $j \in \llbracket -N, N \rrbracket$. Denoting for each integer $k \ge 2$ the event 
		\begin{flalign*}
			\mathsf{E}_k = \bigcap_{j =-k}^k \Big\{ |a_j| + |b_j| \le \displaystyle\frac{1}{1600} \cdot \big( R + \log (|k|+1) \big)  \Big\}; \qquad \mathsf{E}= \bigcap_{k = N}^{\infty} \mathsf{E}_k,
		\end{flalign*}
		
		\noindent we have from the explicit densities of $a_j$ and $b_j$ that $\mathbb{P} [\mathsf{E}_k^{\complement}] \le c_1^{-1} e^{-c_1 (R+\log k)^2}$, for some constant $c_1>0$. Hence, a union bound yields $\mathbb{P}[\mathsf{E}^{\complement}] \le c_2^{-1} e^{-c_2 R^2}$ for some constant $c_2>0$ (as $R \ge \log N$). 
		
		Restricting to $\mathsf{E}$, we then apply \Cref{a2p2}, with the $(N_1, N_2, A, K)$; $(\bm{a}(t); \bm{b}(t))$; and $(\tilde{\bm{a}}(t);\tilde{\bm{b}}(t))$ there equal to $(-N,N, R/200, K)$; $(\breve{\bm{a}}(t); \breve{\bm{b}}(t))$; and $(\bm{a}(t); \bm{b}(t))$ here, respectively. This yields 
		\begin{flalign}
			\label{aajbbj}
			\displaystyle\sup_{t \in [0, T]} \displaystyle\max_{j \in \llbracket K-N,N-K \rrbracket} \big( \big|a_j^{[-N,N]} (t) - \breve{a}_j (t) \big| + |b_j^{[-N,N]} (t) - \breve{b}_j(t) \big| \big) \le e^{-K/4}.
		\end{flalign}		
		
		\noindent Letting $N$ (and then $K$) tend to $\infty$, the first two statements of the proposition therefore follow from \Cref{ajbjequation} (with the $(\bm{a}(t);\bm{b}(t))$ there equal to $(\breve{\bm{a}}(t);\breve{\bm{b}}(t))$ here).		
		
		To confirm the third, we assume that $N_2 \ge -N_1$, as the case when $-N_1 > N_2$ is entirely analogous. We will first use \eqref{ajnmn} to estimate the difference between $(\bm{a}^{[-N_2,N_2]}(t); \bm{b}^{[-N_2,N_2]}(t))$ and $(\bm{a}(t);\bm{b}(t))$, and then use \Cref{a2p2} to estimate the difference between $(\bm{a}^{[N_1, N_2]} (t); \bm{b}^{[N_1,N_2]}(t))$ and $(\bm{a}^{[-N_2,N_2]}(t); \bm{b}^{[-N_2,N_2]}(t))$. To implement the former, first observe that the assumption \eqref{estimaterjp} holds in \Cref{infiniteab} (with the $R$ there uniformly bounded for fixed $\mathfrak{p} \in (0, 1)$), since on $\mathsf{E}$ we have $|a_j| \le (R+\log j)/800 \le (j^{\mathfrak{p}} + \log j)/800 \le (|j|+1)^{\mathfrak{p}}$ for sufficiently large $j$. Therefore, \eqref{ajnmn} applies, with the $(N, N_m)$ there equal to $(\mathfrak{N}, N_2)$ here for some $\mathfrak{N} \ge 4^m N_2 \ge 4^{m-1} N$. This yields for some constant\footnote{By our assumption $RT \le K \le N^{\mathfrak{p}}$, it is quickly verified using \eqref{ajnnnn} that this constant is independent of $T$.} $C_1 = C_1 (\mathfrak{p}) > 0$ that
		\begin{flalign*}
			\displaystyle\sup_{t \in [0,T]} \big( \big| a_j^{[-\mathfrak{N},\mathfrak{N}]} (t) - a_j^{[-N_2,N_2]} (t) \big| + \big| b_j^{[-\mathfrak{N},\mathfrak{N}]} (t) - b_j^{[-N_2,N_2]} (t) \big| \big) \le C_1 e^{-N_2}.
		\end{flalign*}
		
		\noindent Letting $m$, and thus $\mathfrak{N}$, tend to $\infty$ (and using the fact that $N_2 \ge N/3$, as $N_2 \ge -N_1$), we deduce
		\begin{flalign}
			\label{aabb2} 
			\displaystyle\sup_{t \in [0,T]} \big( \big| a_j (t) - a_j^{[-N_2,N_2]} (t) \big| + \big| b_j (t) - b_j^{[-N_2,N_2]} (t) \big| \big) \le C_1 e^{-N/3}.
		\end{flalign}

		\noindent We next apply \Cref{a2p2}, with the $(N_1,N_2; \tilde{N}_1,\tilde{N}_2;A,K)$; $(\bm{a}(t);\bm{b}(t))$; and $(\tilde{\bm{a}}(t);\tilde{\bm{b}}(t))$ there equal to $(N_1,N_2;-N_2,N_2; R/200, K)$; $(\bm{a}^{[-N_2,N_2]}(t);\bm{b}^{[-N_2,N_2]}(t))$; and $(\bm{a}^{[N_1,N_2]}(t);\bm{b}^{[N_1,N_2]}(t))$ here, respectively (using our restriction to $\mathsf{E}$, and the fact that $R \ge \log N$, to verify its hypotheses). As in the derivation of \eqref{aajbbj}, this yields
		\begin{flalign*}
			\displaystyle\sup_{t \in [0, T]} & \displaystyle\max_{j \in \llbracket N_1+K,N_2+K \rrbracket} \big( \big| a^{[N_1,N_2]} (t) - a^{[-N_2,N_2]}(t) \big| \\
			& \qquad \qquad \qquad \qquad   + \big| a^{[N_1,N_2]}(t) - a^{[-N_2, N_2]}(t) \big| \big) \le e^{-K/4}.
		\end{flalign*}
		
		\noindent Together with \eqref{aabb2} (and the fact that $K \le N$), this yields the third part of the corollary. 
	\end{proof}

	\section{Localization Centers}
	
	\label{EigenvalueLocal} 
	
	Throughout this section, we fix $\beta, \theta \in \mathbb{R}_{> 0}$; the constants below may depend on them, even if not explicitly stated. We further let $N_1 \le N_2$ be integers and set $N = N_2 - N_1 + 1$.

	\subsection{Existence and Speed Bounds for Localization Centers}
	
	\label{Center}
	
	In this section we discuss localization centers (recall \Cref{ucenter}) and their properties. We begin by confirming their existence, by proving \Cref{bijectionm}. 
	
	\begin{proof}[Proof of \Cref{bijectionm}]
		
		Let $(\bm{u}_1, \bm{u}_2, \ldots , \bm{u}_N)$ denote an orthonormal eigenbasis for $\bm{M}$; assume that the rows and columns of $\bm{M}$ are indexed by $\llbracket N_1, N_2 \rrbracket$. For each index $j \in \llbracket 1, N \rrbracket$, denote $\bm{u}_j = ( u_j (N_1), u_j (N_1+1), \ldots , u_j (N_2) )$ and define the set 
		\begin{flalign*} 
			\mathcal{I}_j = \{ i \in \llbracket N_1, N_2 \rrbracket : | u_j (i) | \ge (2N)^{-1} \}.
		\end{flalign*} 
		
		\noindent We must show that there exists a bijection $\varphi : \llbracket 1, N \rrbracket \rightarrow \llbracket N_1, N_2 \rrbracket$ such that $\varphi (j) \in \mathcal{I}_j$ for each $j \in \llbracket 1, N \rrbracket$. To do this, it suffices by Hall's theorem to show for any set $\mathcal{J} \subseteq \llbracket 1, N \rrbracket$ that $|\mathcal{I}_{\mathcal{J}}| \ge |\mathcal{J}|$, where $\mathcal{I}_{\mathcal{J}} = \bigcup_{j \in \mathcal{J}} \mathcal{I}_j$. If this statement were false for some subset $\mathcal{J} \subseteq \llbracket 1, N \rrbracket$, then we would have 
		\begin{flalign*}
			|\mathcal{J}|-1 \ge |\mathcal{I}_{\mathcal{J}}| = \displaystyle\sum_{i \in \mathcal{I}_{\mathcal{J}}} \displaystyle\sum_{j=1}^N | u_j (i) |^2 & \ge  \displaystyle\sum_{j \in \mathcal{J}} \displaystyle\sum_{i \in \mathcal{I}_j} | u_j (i) |^2 \\ 
			& \ge \displaystyle\sum_{j \in \mathcal{J}} \displaystyle\sum_{i=N_1}^{N_2} | u_j (i) |^2 -|\mathcal{J}|N (2N)^{-2} \\
			& \ge |\mathcal{J}| - \displaystyle\frac{1}{4},
		\end{flalign*} 
		
		\noindent where the first bound follows from the assumption; the second from the orthonormality of the eigenbasis $(\bm{u}_1, \bm{u}_2, \ldots , \bm{u}_N)$; the third from restricting to $j \in \mathcal{J}$; the fourth from using that $| u_j (i) | < (2N)^{-1}$ for $i \notin \mathcal{I}_j$; and the fifth from the fact that each $\bm{u}_j$ is a unit vector. This is a contradiction, so a $(2N)^{-1}$-localization center bijection for $\bm{M}$ exists. 
	\end{proof}

	The next lemma concerns a random Lax matrix corresponding to the open Toda lattice, sampled under thermal equilibrium (recall \Cref{Localization}). It first states that its eigenvectors are exponentially localized around their localization centers; it then states (as a quick consequence) that such localization centers are essentially unique, up to an error of $(\log N)^2 / 2$.
	
	\begin{lem} 
		
		\label{bijectionl}
		
		There exists a constant $c>0$ such that the following holds with overwhelming probability. Adopt \Cref{lbetaeta}, but assume more generally that $\zeta \ge e^{-200(\log N)^{3/2}}$. Fix any index $j \in \llbracket 1, N \rrbracket$, and let $\varphi_j \in \llbracket N_1, N_2 \rrbracket$ denote any $\zeta$-localization center for $\bm{u}_j$. 
		
		\begin{enumerate}
			\item For any $i \in \llbracket N_1, N_2 \rrbracket$ with $| i - \varphi_j | \ge (\log N)^2/2$, we have $| u_j (i) | \le e^{-c|i-\varphi_j|}$.
			
			\item If $i \in \llbracket N_2, N_2 \rrbracket$ satisfies $| i - \varphi_j | \ge (\log N)^2/2$, then $i$ is not a $\zeta$-localization center for $\bm{u}_j$. 
		\end{enumerate}
		
	\end{lem} 
	
	\begin{proof} 
		
		The first statement of the lemma implies the second (as $c^{-1} e^{-c(\log N)^2 / 2} < \zeta$ for sufficiently large $N$), so it suffices to verify the former. To that end, observe for some $c>0$ that 
		\begin{flalign*}
			\zeta^{-1}  \cdot \mathbb{E} [ | u_j (i) | ] \le \mathbb{E} \big[ | u_j (i)  \cdot u_j ( \varphi_j ) | \big] \le c^{-1} N e^{-c|i-\varphi_j|}.
		\end{flalign*} 
		
		\noindent Here, the first inequality holds by the definition of $\varphi_j$; the second by summing \Cref{uijexponential} over $j \in \llbracket N_1, N_2 \rrbracket$ (and examining its $(k,j)$ term equal to $( j, \varphi_j)$ here). This, with a Markov estimate, a union bound, and the facts that $\zeta \ge e^{-200(\log N)^{3/2}}$ and $|i-\varphi_j| \ge (\log N)^2/2$, yields the first statement of the lemma.
	\end{proof} 
	
	The next lemma bounds the ``speed'' at which localization centers can move, for a Lax matrix initially at thermal equilibrium.

	\begin{lem}
		
		\label{centert} 
		
		There exists a constant $c>0$ such that the following holds with overwhelming probability. Adopt \Cref{lbetaeta}, and assume more generally that $\zeta \ge e^{-200(\log N)^{3/2}}$. For each $j \in \llbracket 1, N \rrbracket$, let $\varphi_j \in \llbracket N_1, N_2 \rrbracket$ denote any $\zeta$-localization center for $\bm{u}_j (0)$. Then, for each $(j, s) \in \llbracket 1, N \rrbracket \times [0, T]$, we have
		\begin{flalign}
			\label{ujms}
			| u_j (m; s) | \le e^{-c |m-\varphi_j|}, \qquad \text{for any $m \in \llbracket N_1, N_2 \rrbracket$ with $| m - \varphi_j| \ge T (\log N)^2$}.
		\end{flalign}
		
		\noindent Moreover, $m$ is not a localization center for $\bm{u}_j (s)$ whenever $| m - \varphi_j| \ge T (\log N)^2$. 
		
	\end{lem} 
	
	\begin{proof} 
		
		The first statement \eqref{ujms} of the lemma implies the second, as $e^{-c|m - \varphi_j|} < \zeta$ for $| m - \varphi_j| \ge T (\log N)^2$ if $N$ is sufficiently large; it therefore suffices to establish \eqref{ujms}. Recalling \Cref{adelta} and letting $\mathfrak{c} \in (0, 1)$ denote the constant $c$ from \Cref{bijectionl}, and define the event
		\begin{flalign} 
			\label{evente0} 
			\mathsf{E} =  \bigcap_{s \ge 0} \mathsf{BND}_{\bm{L}(s)} (\log N) \cap \bigcap_{i=N_1}^{N_2} \bigcap_{j=1}^N \big\{ \mathbbm{1}_{|i-\varphi_j| \ge (\log N)^2/2} \cdot | u_j (i; 0) | \le e^{-\mathfrak{c} |i - \varphi_j|} \big\}.
		\end{flalign}
		
		\noindent By \Cref{l0eigenvalues} and \Cref{bijectionl}, $\mathsf{E}$ is overwhelmingly probable, so we restrict to $\mathsf{E}$ for the remainder of this proof. 
		
		We next recall a fact concerning the evolution of the Lax matrix $\bm{L}(s)$. For each $s \in \mathbb{R}$, define the tridiagonal skew-symmetric matrix $\bm{P} (s) = [ P_{ij} (s) ] \in \Mat_{\llbracket N_1, N_2 \rrbracket}$ as follows. For each $i \in \llbracket N_1, N_2 - 1 \rrbracket$, set $P_{i,i+1} (s) = a_i (s)/2$ and $P_{i+1,i} (s) = -a_i (s)/2$; for all $(i, j) \in \llbracket N_1, N_2 \rrbracket^2$ not of the above form, set $P_{i,j} (s) = 0$. For each $s \in \mathbb{R}$, further let $\bm{V}(s) = [ V_{ij} (s) ] \in \Mat_{\llbracket N_1, N_2 \rrbracket}$, satisfying the ordinary differential equation $\partial_s \bm{V}(s) = \bm{P}(s) \cdot \bm{V}(s)$, with initial data $\bm{V}(0) = \Id$; the existence of such a matrix $\bm{V}(s)$ follows from the Picard--Lindel\"{o}f theorem. For any $(i, j) \in \llbracket N_1, N_2 \rrbracket^2$, its $(i,j)$-entry is more explicitly given by 
		\begin{flalign}
			\label{vijs}
			V_{ij} (s) = \mathbbm{1}_{i=j} + \displaystyle\sum_{k=1}^{\infty}  \displaystyle\sum_{i_1=N_1}^{N_2} \cdots \displaystyle\sum_{i_{k-1}=N_1}^{N_2}  \displaystyle\int_0^s \cdots \displaystyle\int_0^s \mathbbm{1}_{s_1 > s_2 > \cdots > s_k} \cdot \displaystyle\prod_{h=1}^k P_{i_{h-1}, i_h} (s_h) ds_h,
		\end{flalign} 
		
		\noindent where we have denoted $(i_0, i_k) = (i, j)$. Then, by \cite[Section 2]{FMPL}, we have $\bm{V} (s)^{-1} \cdot \bm{L}(s) \cdot \bm{V}(s) = \bm{L}(0)$. 
		
		This implies that $\bm{L}(s) = \bm{V}(s) \cdot \bm{L}(0) \cdot \bm{V}(s)^{\mathsf{T}}$, as $\bm{V}(s)$ is orthogonal (since $\bm{V}(0) = \Id$, $\partial_s \bm{V}(s) = \bm{P}(s) \cdot \bm{V}(s)$, and $\bm{P}(s)$ is skew-symmetric). Hence, letting $\bm{U}(s) = [ U_{ij} (s) ] \in \Mat_{N \times N}$ denote matrix of eigenvectors of $\bm{L}(s)$, whose $(i,j)$-entry is given by $U_{ij}(s) = u_j (i; s)$ for each $(i, j) \in \llbracket N_1, N_2 \rrbracket \times \llbracket 1, N \rrbracket$, we have $\bm{U}(s) = \bm{V}(s) \cdot \bm{U}(0)$. In particular,
		\begin{flalign}
			\label{vijsu0} 
			u_j (i; s) = \displaystyle\sum_{k = N_1}^{N_2} V_{ik}(s) \cdot u_j (k; 0).
		\end{flalign}
		
		Now observe whenever $|i-j| \ge 20 T \log N$ that 
		\begin{flalign}
			\label{u0vijs}
			| V_{ij} (s) |\le \displaystyle\sum_{k = |i-j|}^{\infty} \displaystyle\frac{s^k}{k!} \cdot (2 \log N)^k \le \displaystyle\sum_{k=|i-j|}^{\infty} \bigg( \displaystyle\frac{2e s \log N}{k} \bigg)^k \le \displaystyle\sum_{k=|i-j|}^{\infty} e^{-k} \le 2e^{-|i-j|}.
		\end{flalign}
		
		\noindent Here, in the first inequality we used \eqref{vijs}, with the facts that each $| P_{ij} (s_h) | \le \log N$ (as we have restricted to the event $\mathsf{E}$ from \eqref{evente0}) and that $P_{ij} = 0$ whenever $|i-j| \ne 1$ (meaning that there are at most two choices for each $i_h$ that gives rise to a nonzero summand in \eqref{vijs}); in the second we used the bound $k! \ge (e^{-1} k)^k$ for each $k \ge 0$; in the third we used the bound $2k^{-1} es \log N \le e^{-1}$ for $k \ge |i-j| \ge 20 T \log N$; and in the fourth we performed the sum. Hence, for $| m - \varphi_j| > T (\log N)^2$, 
		\begin{flalign*}
			| u_j (m; s) | & \le \displaystyle\sum_{k = N_1}^{N_2} \mathbbm{1}_{|k-m| \ge |m - \varphi_j| / 2} \cdot | V_{mk}(s) | + \displaystyle\sum_{k=N_1}^{N_2} \mathbbm{1}_{|k-m| \le |m - \varphi_j| / 2} \cdot | u_j (k; 0) | \\
			& \le 2 \displaystyle\sum_{k = N_1}^{N_2} \mathbbm{1}_{|k-m| \ge |m - \varphi_j| / 2} \cdot e^{-|k-m|} + \displaystyle\sum_{k=N_1}^{N_2} \mathbbm{1}_{|k-\varphi_j| \ge |m-\varphi_j|/2} \cdot \mathfrak{c}^{-1} e^{-\mathfrak{c} |k - \varphi_j|} \\ 
			& \le 16 \mathfrak{c}^{-2} e^{-\mathfrak{c} |m-\varphi_j|/2},
		\end{flalign*}
		
		\noindent which verifies \eqref{ujms}. Here, the first inequality follows from \eqref{vijsu0}, together with the facts that $| u_j (k; 0) | \le 1$ (as $\bm{u}_j$ is a unit vector) and $| V_{mk} (s) | \le 1$ (as $\bm{V}(s)$ is orthogonal); the second follows from \eqref{u0vijs} (with the fact that $T (\log N)^2 \ge 40T \log N$ for sufficiently large $N$) and the fact that we have restricted to the event $\mathsf{E}$ from \eqref{evente0} (with the fact that, if $| m - \varphi_j | \ge T (\log N)^2$ and $|k-m| \le | \varphi_j - m | / 2$, then $| k - \varphi_j | \ge | m - \varphi_j | / 2 \ge (\log N)^2/2$); and the third follows from performing the sums.
	\end{proof} 
	
	\subsection{Resolvent Perturbation Estimates} 
	
	\label{PerturbLambda}  
	
	In this section we estimate the effect of perturbing a random Lax matrix, associated with the open Toda lattice sampled from thermal equilibrium, on its eigenvalues. We first require some notation.  
	
	\begin{assumption} 
		
		\label{lmatrixl} 
		
		Sample $(\bm{a}; \bm{b})$ under the thermal equilibrium $\mu_{\beta, \theta; N-1,N}$ from \Cref{mubeta2}, where $\bm{a} = (a_{N_1}, a_{N_1+1}, \ldots , a_{N_2-1})$ and $\bm{b} = (b_{N_1}, b_{N_1+1}, \ldots , b_{N_2})$. Let $\bm{L} = [L_{ij}] \in \SymMat_{\llbracket N_1, N_2 \rrbracket}$ denote the associated Lax matrix (as in \Cref{matrixl}), and let $\tilde{\bm{L}} = [\tilde{L}_{ij}] \in \SymMat_{\llbracket N_1, N_2 \rrbracket}$ be another tridiagonal matrix. Assume that there is an index set $\mathcal{D} \subseteq \llbracket N_1, N_2 \rrbracket$ and a real number $\delta \in (0, 1)$ satisfying 
		\begin{flalign}
			\label{ll} 
			\displaystyle\max_{i,j \notin (\llbracket N_1, N_2 \rrbracket \setminus \mathcal{D})^2} |\tilde{L}_{ij}| \le 2 \log N; \qquad \displaystyle\max_{i,j \in \llbracket N_1, N_2 \rrbracket \setminus \mathcal{D}} |L_{ij} - \tilde{L}_{ij}| \le \delta.
		\end{flalign} 
		
		\noindent For any $z \in  \mathbb{C}$, denote $\bm{G}(z) = [ G_{ij} (z) ] = (\bm{L} - z)^{-1}$ and $\tilde{\bm{G}} (z) = [ \tilde{G}_{ij} (z) ] = (\tilde{\bm{L}} - z)^{-1}$. 
		
	\end{assumption}
	
	The following lemma bounds the difference between entries of the resolvents $\bm{G} (z)$ and $\tilde{\bm{G}} (z)$, in terms of the distance from their indices to the set $\mathcal{D}$ where $\bm{L}$ and $\tilde{\bm{L}}$ might substantially disagree.
	
	\begin{lem} 
		
		\label{lgdifference}

		There exists a constant $c>0$ such that the following holds with overwhelming probability. Adopt \Cref{lmatrixl}; let $\eta \in [\delta, 1]$ be a real number; and define the set $\Omega = \{ z \in \mathbb{C} : -N \le \Real z \le N, \eta \le \Imaginary z \le 1 \}$. For any integers $i, j \in \llbracket N_1, N_2 \rrbracket$, we have 
		\begin{flalign}
			\label{gkke} 
			\displaystyle\sup_{z \in \Omega} \big| G_{ij} (z) - \tilde{G}_{ij} (z) \big| \le e^{ (\log N)^2} \eta^{-2} (\delta^{1/4} + e^{-c \dist (i, \mathcal{D}) - c \dist (j, \mathcal{D})}).
		\end{flalign}
		
	\end{lem} 
	
	\begin{proof}
		
		Throughout, we suppose $\dist (i, \mathcal{D}) \ge \dist (j, \mathcal{D})$, as the proof in the alternative case is entirely analogous. We may assume that $\dist (i, \mathcal{D}) \ge 1$, as otherwise \eqref{gkke} holds deterministically, by \eqref{gijeta}.
		
		We first show a variant of \eqref{gkke} for a fixed point $z \in \Omega$. So, fix $z_0 \in \Omega$; abbreviate $G_{km} = G_{km} (z_0)$ and $\tilde{G}_{km} = \tilde{G}_{km} (z_0)$, for any $k, m \in \llbracket N_1, N_2 \rrbracket$. Recalling \Cref{adelta}, also denote  
		\begin{flalign}
			\label{eprobability} 
			\mathsf{E} = \mathsf{BND}_{\bm{L}} (\log N), \qquad \text{so that} \qquad \mathbb{P} [ \mathsf{E}^{\complement} ] \le \mathfrak{c}^{-1} e^{-\mathfrak{c} (\log N)^2},
		\end{flalign}
		
		\noindent for some constant $\mathfrak{c} > 0$, by \Cref{l0eigenvalues}. Further let $\mathcal{D}_0 = \{ k \in \llbracket N_1, N_2 \rrbracket : \dist (k, \mathcal{D}) \le 1 \}$. Then, for any real number $s \in (0, 1)$, there exist constants $c_1 = c_1 (s) > 0$ and $C_1 = C_1 (s) > 1$ such that
		\begin{flalign*}
			\mathbb{E} \big[ \mathbbm{1}_{\mathsf{E}} \cdot | G_{ij} - \tilde{G}_{ij}|^{s} \big] & \le \displaystyle\sum_{N_1 \le k, m \le N_2} \mathbb{E} \big[ \mathbbm{1}_{\mathsf{E}} \cdot | G_{ik} (L_{km} - \tilde{L}_{km}) \tilde{G}_{kj} |^s \big] \\
			& \le \eta^{-s} \displaystyle\sum_{N_1 \le k, m \le N_2} \mathbb{E} \big[ \mathbbm{1}_{\mathsf{E}} \cdot | G_{ik}|^s \cdot |L_{km} - \tilde{L}_{km}|^{s} \big] \\
			& \le \delta^{s} \eta^{-s} \displaystyle\sum_{k \notin \mathcal{D}_0} \mathbb{E} [ |G_{ik}|^{s} ] + \eta^{-s} (3 \log N)^{s} \cdot \displaystyle\sum_{k \in \mathcal{D}_0} \mathbb{E} [ |G_{ik}|^s ] \\
			& \le C_1 \delta^s \eta^{-s} \displaystyle\sum_{k \notin \mathcal{D}_0} e^{-c_1 |k-i|} + 3 C_1 \eta^{-s} (\log N)^s \cdot \displaystyle\sum_{k \in \mathcal{D}_0} e^{-c_1 |k-i|} \\
			& \le 8 c_1^{-1} C_1 \eta^{-s} (\log N)^s ( \delta^s + e^{-c_1 \dist (i, \mathcal{D}_0)} ),
		\end{flalign*}
		
		\noindent where in the first inequality we applied \eqref{ab}; in the second we applied \eqref{gijeta}; in the third we used \eqref{ll}, with the definition \eqref{eprobability} of $\mathsf{E}$ (which implies that $\mathbbm{1}_{\mathsf{E}} \cdot |L_{km} - \tilde{L}_{km}| \le 3 \log N$) and the fact that $\bm{L} - \tilde{\bm{L}}$ is tridiagonal; in the fourth we used \eqref{gijs}; and in the fifth we performed the sums. Taking $s = 2 / 3$ and using the fact that $\Imaginary z_0 \ge \eta$ (as $z_0 \in \Omega$), we therefore deduce by a Markov estimate that there exist constants $c_2 > 0$ and $C_2 > 1$ such that 
		\begin{flalign}
			\label{gkkz0}
			\begin{aligned}
				\mathbb{P} \big[ & \mathbbm{1}_{\mathsf{E}} \cdot |G_{ij} - \tilde{G}_{ij}| \ge  e^{(\log N)^2/2} \eta^{-1} (\delta + e^{-c_1 \dist (i, \mathcal{D})})^{1/4} \big] \\
				& \qquad \qquad \le e^{-(\log N)^2/3} \eta^{2/3} (\delta  + e^{-c_1 \dist (i, \mathcal{D})})^{-1/6} \cdot \mathbb{E} \big[ \mathbbm{1}_{\mathsf{E}} \cdot |G_{ij} - \tilde{G}_{ij}|^{2/3} \big] \\
				& \qquad \qquad \le C_2 e^{-c_2 (\log N)^2} (\delta + e^{-c_1 \dist (i, \mathcal{D})})^{1/2}.
			\end{aligned} 
		\end{flalign} 
		
		Now, let us extend the estimate \eqref{gkkz0} to simultaneously apply for all $z \in \Omega$. To do this, set $\delta_0 = (\delta + e^{-c_1 \dist (i, \mathcal{D})})^{1/4}$. Observe for any $z, z' \in \Omega$ with $|z-z'| \le \delta_0$ that 
		\begin{flalign}
			\label{gkkzz}
			| G_{ij} (z) - G_{ij} (z') | \le |z-z'| \displaystyle\sum_{k=1}^N | G_{ik} (z) | \cdot | \tilde{G}_{kj} (z) | \le \eta^{-2} N \delta_0,
		\end{flalign}
		
		\noindent where in the first inequality we used \eqref{ab} and in the second we used \eqref{gijeta} (with the fact that $\min \{ \Imaginary z, \Imaginary z' \} \ge \eta$ for $z, z' \in \Omega$). Let $\Omega_0 \subset \Omega$ denote a $\delta_0$-mesh of $\Omega$, so that $|\Omega_0| \le 8 N \delta_0^{-2}$. Denote the event 
		\begin{flalign*}
			\mathsf{F} = \bigg\{ \displaystyle\max_{z \in \Omega_0} | G_{ij} (z) - \tilde{G}_{ij} (z) | \le e^{(\log N)^2/2} \eta^{-1} \delta_0 \bigg\}.
		\end{flalign*} 
		
		\noindent Applying \eqref{eprobability} and \eqref{gkkz0} for all $z_0 \in \Omega$, yields constants $c_3 > 0$ and $C_3 > 1$ such that 
		\begin{flalign}
			\label{gkkevent2}
			\begin{aligned} 
				\mathbb{P} [ \mathsf{F}^{\complement} ] \le C_2 e^{-c_2 (\log N)^2} \delta_0^2 \cdot |\Omega_0| + \mathbb{P} [ \mathsf{E}^{\complement} ] & \le 8 C_2 N e^{-c_2 (\log N)^2} + \mathfrak{c}^{-1} e^{-\mathfrak{c} (\log N)^2} \\ 
				&  \le  C_3 e^{-c_3 (\log N)^2}.
			\end{aligned} 
		\end{flalign}
		
		\noindent Observe for sufficiently large $N$ that $| G_{ij} (z) - \tilde{G}_{ij} (z) | \le e^{(\log N)^2} \eta^{-2} \delta_0$ holds for all $z \in \Omega$ on the event $\mathsf{F}$, due to \eqref{gkkzz} and the estimate $e^{(\log N)^2/2} \eta^{-1} \delta_0 + 2 \eta^{-2} N \delta_0 \le e^{(\log N)^2} \eta^{-2} \delta_0$. Since $\dist (i, \mathcal{D}) \ge \dist (j, \mathcal{D})$, this yields \eqref{gkkevent2}.
	\end{proof} 
	
	Using \Cref{lgdifference} and \Cref{abgh}, we can establish the following two corollaries. The first indicates that eigenvalues of $\bm{L}$ with localization centers distant from $\mathcal{D}$ are also nearly eigenvalues of $\tilde{\bm{L}}$ (with the same localization center), and the second essentially indicates the reverse.

	\begin{cor}
		
		\label{lleigenvalues2} 
		
		There exists a constant $c>0$ such that the following holds with overwhelming probability. Adopt \Cref{lmatrixl}; assume $\delta \le e^{-10(\log N)^2}$, and let $\zeta \ge e^{-200(\log N)^{3/2}}$ be a real number. Fix $\lambda \in \eig \bm{L}$, and let $\varphi \in \llbracket N_1, N_2 \rrbracket$ denote a $\zeta$-localization center of $\lambda$ with respect to $\bm{L}$, satisfying  
		\begin{flalign} 
			\label{dn2} 
			\dist (\varphi, \mathcal{D}) \ge (\log N)^3.
		\end{flalign}
		
		\noindent Then there exists an eigenvalue $\tilde{\lambda} \in \eig \tilde{\bm{L}}$ such that
		\begin{flalign}
			\label{lambdalambda} 
			|\lambda - \tilde{\lambda}| \le e^{(\log N)^2} (\delta^{1/8} + e^{-c \dist (\varphi, \mathcal{D})}), 
		\end{flalign}
		
		\noindent and $\varphi$ is an $N^{-1} \zeta$-localization center for $\tilde{\lambda}$ with respect to $\tilde{\bm{L}}$. 
		
	\end{cor}
	
	\begin{proof}
		
		Let $\mathfrak{c} \in (0, 1)$ denote the constant $c$ from \Cref{lgdifference}, and set
		\begin{flalign}
			\label{etadelta}
			\delta_0 = \delta^{1/4} + e^{-\mathfrak{c} \dist (\varphi, \mathcal{D})}; \qquad \eta = (e^{(\log N)^2} \delta_0)^{1/2}. 
		\end{flalign}
		
		\noindent Recalling the event $\mathsf{E}$ from \eqref{eprobability}, also define event 
		\begin{flalign}
			\label{eventf} 
			\begin{aligned}
				\mathsf{F} & = \mathsf{E} \cap \bigg\{ \displaystyle\sup_{E \in [-N, N]} \big| G_{\varphi \varphi} (E + \mathrm{i} \eta) - \tilde{G}_{\varphi \varphi} (E + \mathrm{i}\eta) \big| \le e^{(\log N)^2} \eta^{-2} \delta_0 \bigg\}.
			\end{aligned}
		\end{flalign}
		
		\noindent By \eqref{ll} and the min-max principle, on $\mathsf{F}$, we have $\eig \bm{L} \cup \eig \tilde{\bm{L}} \subseteq [-N, N]$ for sufficiently large $N$.
		
		By \eqref{eprobability} and \Cref{lgdifference}, we may restrict to the e vent $\mathsf{F}$. We then apply \Cref{abgh}, with the parameters $(\varphi; \eta, \zeta, \delta; \bm{A}, \bm{B})$ there equal to $(\varphi; \eta, \zeta, e^{(\log N)^2} \eta^{-2} \delta_0; \bm{L}, \tilde{\bm{L}})$ here. To verify \eqref{deltaetachi}, observe that the first estimate there follows for sufficiently large $N$ from the fact that 
		\begin{flalign}
			\label{eta0} 
			(2\eta) \cdot e^{(\log N)^2} \eta^{-2} \delta_0 \le 2e^{(\log N)^2/2} \cdot (\delta^{1/4} + e^{-\mathfrak{c} \dist (\varphi, \mathcal{D})})^{1/2} \le e^{-400 (\log N)^{3/2}} \le \zeta^2,
		\end{flalign}
		
		\noindent where in the first bound we used \eqref{etadelta}; in the second we used the facts that $\delta \le e^{-10(\log N)^2}$, that $\dist (\varphi, \mathcal{D}) \ge (\log N)^3$, and that $N$ is sufficiently large; and in the third we used the fact that $\zeta \ge e^{-200 (\log N)^{3/2}}$. The second estimate in \eqref{deltaetachi} follows from the fact that $\varphi$ is a $\zeta$-localization center for $\lambda$ with respect to $\bm{L}$, and the third follows from our restriction to $\mathsf{F}$. 
		
		Thus, \Cref{abgh} yields some $\tilde{\lambda} \in \eig \tilde{\bm{L}}$ such that $|\lambda - \tilde{\lambda}| \le 3N \zeta^{-2} \eta$, which by \eqref{etadelta} (and again the fact that $\zeta \ge e^{-200 (\log N)^{3/2}}$) confirms \eqref{lambdalambda}. Since $(6N)^{-1/2} \zeta \ge N^{-1} \zeta$, \Cref{abgh} also indicates that $\varphi$ is an $N^{-1} \zeta$-localization center $\tilde{\lambda}$ with respect to $\tilde{\bm{L}}$, establishing the corollary. 
	\end{proof}

	\begin{cor}
		
		\label{lleigenvalues} 
		
		There exists a constant $c>0$ such that the following holds with overwhelming probability. Adopt \Cref{lmatrixl}; assume that $\delta \le e^{-10(\log N)^2}$; and let $\zeta \ge N e^{-200(\log N)^{3/2}}$ be a real number. Fix $\tilde{\lambda} \in \eig \tilde{\bm{L}}$, and let $\tilde{\varphi} \in \llbracket N_1, N_2 \rrbracket$ denote a $\zeta$-localization center of $\tilde{\lambda}$ with respect to $\tilde{\bm{L}}$; suppose that 
		\begin{flalign} 
			\label{dn} 
			\dist (\tilde{\varphi}, \mathcal{D}) \ge (\log N)^3.
		\end{flalign}
		
		\begin{enumerate} 
			
			\item There exists a unique eigenvalue $\lambda \in \eig \bm{L}$ such that $|\lambda - \tilde{\lambda}| \le e^{(\log N)^2} (\delta^{1/8} + e^{-c \dist (\tilde{\varphi}, \mathcal{D})})$.
			
			\item We have that $\tilde{\varphi}$ is an $N^{-1} \zeta$-localization center of $\lambda$ with respect to $\bm{L}$, and any $N^{-1} \zeta$-localization center $\varphi \in \llbracket N_1, N_2 \rrbracket$ satisfies $|\varphi - \tilde{\varphi}| \le (\log N)^2 / 2$. 
			
		\end{enumerate} 
		
	\end{cor}
	
	\begin{proof}
		
		Let $\mathfrak{c} \in (0, 1)$ denote the constant $c$ from \Cref{lgdifference}, and set $\delta_0$ and $\eta$ as in \eqref{etadelta}. Recalling \Cref{adelta} and the event $\mathsf{F}$ from \eqref{eventf}, define the event $\mathsf{G} = \mathsf{F} \cap \mathsf{SEP}_{\bm{L}} (e^{-(\log N)^2/2})$. By \eqref{eprobability}, \Cref{eigenvalues0}, and \Cref{lgdifference}, $\mathsf{G}$ is overwhelmingly probable. By \eqref{ll} and the min-max principle we have on $\mathsf{G}$ that $\eig \bm{L} \cup \eig \tilde{\bm{L}} \subseteq [-N, N]$ for sufficiently large $N$.
		
		Restricting to the event $\mathsf{G}$, we then apply \Cref{abgh}, with the parameters $(\varphi; \eta, \zeta, \delta; \bm{A}, \bm{B})$ there equal to $(\tilde{\varphi}; \eta, \zeta, e^{(\log N)^2} \eta^{-2} \delta_0; \tilde{\bm{L}}, \bm{L})$ here. As in the proof of \Cref{lleigenvalues2}, the first bound in \eqref{deltaetachi} is verified by \eqref{eta0}; the second  by the fact that $\tilde{\varphi}$ is a $\zeta$-localization center for $\tilde{\lambda}$ with respect to $\tilde{\bm{L}}$; and the third by our restriction to $\mathsf{G}$. Thus, \Cref{abgh} gives some $\lambda \in \eig \bm{L}$ such that $|\lambda - \tilde{\lambda}| \le 3N \zeta^{-2} \eta$, which by \eqref{etadelta} (and again the bound $\zeta \ge e^{-200 (\log N)^{3/2}}$) yields some $\lambda$ satisfying the first statement of the corollary. As any two distinct eigenvalues of $\bm{L}$ differ by at least $e^{-(\log N)^2/2}$ by our restriction to $\mathsf{G}$, and $6 N \zeta^{-2} \eta \le e^{-(\log N)^2}$ (by \eqref{etadelta}, \eqref{dn}, and the facts that $\zeta \ge e^{-200(\log N)^{3/2}}$ and $\delta \le e^{-10 (\log N)^2}$), such an eigenvalue $\lambda \in \eig \bm{L}$ satisfying these hypotheses is unique. 
		
		Since $(6N)^{-1/2} \zeta \ge N^{-1} \zeta$, \Cref{abgh} also indicates that $\varphi$ is an $N^{-1} \zeta$-localization center $\tilde{\lambda}$ with respect to $\tilde{\bm{L}}$. Thus, the second part of \Cref{bijectionl} yields that, with overwhelming probability, any $N^{-1} \zeta$-localization center $\varphi \in \llbracket N_1, N_2 \rrbracket$ of $\lambda$ with respect to $\bm{L}$ satisfies $|\varphi - \tilde{\varphi}| \le (\log N)^2/2$. This confirms the second statement of the corollary.
	\end{proof}

	\section{Eigenvector Analysis for Lax Matrices}
	
	\label{Gamma} 
	
	In this section we analyze the eigenvectors of random Lax matrices. We begin in \Cref{ExponentMatrix} by proving a result, \Cref{uk12}, that approximates the exponential rate of decay of the first entries of these eigenvectors, by reducing it to \Cref{ltmatrix} below; the latter indicates that eigenvalues of the Lax matrix are close to those of certain truncations of it. We then show \Cref{ltmatrix} in \Cref{MatrixEigenvalue0} and \Cref{ProofMatrix0}. Throughout this section, we adopt \Cref{lbetaeta}.

	\subsection{Eigenvector Decay for Lax Matrices}
	
	\label{ExponentMatrix}

	The following proposition estimates the behavior of the first entries of eigenvectors of the Lax matrix $\bm{L}(t)$. Recall that we adopt \Cref{lbetaeta} throughout.

	\begin{prop}
		
		\label{uk12} 
		
		For any real number $t \in [0, T]$, the following holds with overwhelming probability. For any integer $k \in \llbracket 1, N \rrbracket$ satisfying 
		\begin{flalign}
			\label{n1k} 
			N_1 + T (\log N)^5 + N^{1/100} \le \varphi_t(k) \le N_2 - T (\log N)^5 - N^{1/100},
		\end{flalign}
		
		\noindent we have 
		\begin{flalign}
			\label{ukn1t}
			\Bigg| \log | u_k (N_1; t) | - \displaystyle\sum_{j=N_1}^{\varphi_t(k)-1} \log L_{j,j+1} (t) + \displaystyle\sum_{i : \varphi_t (i) < \varphi_t(k)} \log |\lambda_i - \lambda_k| \Bigg| \le (\log N)^6 .
		\end{flalign}
		
	\end{prop}

	To establish \Cref{uk12}, we  use the following lower bounds on $\log | u_k (N_1; t) |$ and $\log | u_k (N_2; t) |$. 
	
	\begin{lem}
		
		\label{uk1lower}
		
		Adopting the notation and assumptions of \Cref{uk12}, we have with overwhelming probability that
		\begin{flalign}
			\label{ukn1n2} 
			\begin{aligned}
				& \log | u_k (N_1; t) | \ge  \displaystyle\sum_{j=N_1}^{\varphi_t (k) - 1} \log L_{j,j+1} (t) - \displaystyle\sum_{i:\varphi_t(i) < \varphi_t(k)} \log |\lambda_i - \lambda_k| - \displaystyle\frac{1}{2} \cdot (\log N)^6; \\ 
				& \log | u_k (N_2; t) | \ge \displaystyle\sum_{j = \varphi_t(k)}^{N_2-1} \log L_{j,j+1} (t) - \displaystyle\sum_{i:\varphi_t(i) > \varphi_t(k)} \log |\lambda_i - \lambda_k| - \displaystyle\frac{1}{2} \cdot (\log N)^6.
			\end{aligned} 
		\end{flalign}
	\end{lem}

	Given \Cref{uk1lower}, we can quickly establish \Cref{uk12} using \Cref{n1n2u}.
	
	\begin{proof}[Proof of \Cref{uk12}]
		
		Let $\mathsf{F}$ denote the event on which \eqref{ukn1n2} holds. By \Cref{uk1lower}, $\mathsf{F}$ is overwhelmingly probable, so it suffices to show that \eqref{ukn1t} holds for sufficiently large $N$ on the event $\mathsf{F}$; assume to the contrary that this is false. Since the first bound in \eqref{ukn1n2} holds but \eqref{ukn1t} does not hold, we must have 
		\begin{flalign*}
			\log | u_k (N_1; t) | \ge \displaystyle\sum_{j=N_1}^{\varphi_t(k)-1} \log L_{j,j+1} (t) - \displaystyle\sum_{i: \varphi_t(i) < \varphi_t(k)} \log |\lambda_i - \lambda_k| + (\log N)^6.
		\end{flalign*}
		
		\noindent Together with the second bound in \eqref{ukn1n2}, this yields
		\begin{flalign*}
			\log | u_k (N_1; t) | + \log | u_k (N_2; t) | \ge \displaystyle\sum_{j=N_1}^{N_2-1} \log L_{j,j+1} (t) - \displaystyle\sum_{i \ne k} \log |\lambda_i - \lambda_k| + \displaystyle\frac{1}{2} \cdot (\log N)^6,
		\end{flalign*}
		
		\noindent which contradicts \Cref{n1n2u}. This verifies the proposition.
	\end{proof}

	To prove \Cref{uk1lower}, we will make use of \Cref{smatrixu}, which expresses the eigenvector entries of $\bm{L}(t)$ in terms of transfer matrices. The latter admit an explicit form, given by \Cref{sij}, which involves the eigenvalues of truncations of $\bm{L}(t)$. So, we require an estimate on how the eigenvalues of the Lax matrix $\bm{L}(t)$ change after setting one of its rows and columns to $0$. This is provided by the first part of the following proposition. Its second part states that the localization centers of these eigenvalues cannot differ by too much, and its third improves the bound on this difference (making it independent of $T$) if one of the localization centers is not too close to an endpoint of the domain $\llbracket N_1, N_2 \rrbracket$. The proof of this proposition will appear in \Cref{ProofMatrix0} below. 
	
	\begin{prop}
		
		\label{ltmatrix}
		
		For any real number $t \in [0, T]$, the following holds with overwhelming probability, for some constant $c>0$. Let $\ell \in \llbracket N_1, N_2 \rrbracket$ be an integer, and assume that 
		\begin{flalign}
			\label{tn1n2}
			N_1 + T (\log N)^4 + N^{1/100} \le \ell \le N_2 - T (\log N)^4 - N^{1/100}.
		\end{flalign}
		
		\noindent Set $\bm{P} = \bm{L}(t)$ and $\bm{M} = \bm{P}^{(\ell)}$. Let $\mu \in \eig \bm{M}$ be any eigenvalue of $\bm{M}$; let $\Phi \in \llbracket N_1, N_2 \rrbracket$ be any $\zeta$-localization center for $\mu$ with respect to $\bm{M}$. Suppose that $|\Phi - \ell| \ge (\log N)^4$.
		
		\begin{enumerate} 
			
			\item There exists a unique eigenvalue $\lambda \in \eig \bm{P}$ such that $|\lambda - \mu| \le e^{- c \min \{ |\Phi-\ell|, N^{1/100} \}}$. 
			\item If $\varphi \in \llbracket N_1, N_2 \rrbracket$ is a $\zeta$-localization center for $\lambda$ with respect to $\bm{P}$, then we have $|\varphi - \Phi| \le T (\log N)^3$. 
			\item Assume that 
			\begin{flalign}
				\label{n1n2center} 
				N_1 + T (\log N)^3 + \displaystyle\frac{1}{2} \cdot N^{1/100} \le \Phi \le N_2 - T (\log N)^3 - \displaystyle\frac{1}{2} \cdot N^{1/100}.
			\end{flalign}
			
			\noindent If $\varphi \in \llbracket N_1, N_2 \rrbracket$ is a $\zeta$-localization center for $\lambda$ with respect to $\bm{P}$, then $|\varphi - \Phi| \le (\log N)^2$.
		\end{enumerate} 
	\end{prop} 
	
	Given \Cref{ltmatrix}, we can establish \Cref{uk1lower}.
	
	\begin{proof}[Proof of \Cref{uk1lower}]
		
		We only establish the first bound in \eqref{ukn1n2}, as the second would then follow from symmetry. Throughout, we set $\bm{P} = [P_{ij}] = \bm{L}(t)$ and $v_j = u_k (j; t)$ for each $j \in \llbracket N_1, N_2 \rrbracket$. 
		
		Recall from \Cref{EigenvectorM} the notation on the transfer matrices $\bm{S}_{\mathcal{K}}$, defined by \eqref{ks} and \eqref{k2s}; abbreviate $\bm{S} = [S_{ij}] = \bm{S}_{\llbracket N_1, \varphi_t(k) \rrbracket} ( \lambda_k; \bm{P}) \in \Mat_{2 \times 2}$, whose entries are indexed by $i, j \in \{ 1, 2 \}$. Also denote $\bm{Q} = \bm{P}^{[N_1, \varphi_t (k)-1]}$, where we recall from above \Cref{sij} that this is the matrix obtained by restricting $\bm{P}$ to rows and columns indexed by $\llbracket N_1, \varphi_t (k) - 1 \rrbracket$. Then \Cref{smatrixu} yields $\bm{S} \cdot (0, v_{N_1}) = (v_{\varphi_t(k)}, v_{\varphi_t(k)+1})$, so $S_{12} \cdot v_{N_1} = v_{\varphi_t(k)}$. We thus deduce for sufficiently large $N$ that
		\begin{flalign*}
			\log |v_{N_1}| & = \log |v_{\varphi_t(k)}| - \log |S_{12}| \\ 
			& \ge \log \zeta - \log |S_{12}| \ge \displaystyle\sum_{j=N_1}^{\varphi_t(k)-1} \log P_{j,j+1} - \displaystyle\sum_{\mu \in \eig \bm{Q}} \log |\mu - \lambda_k| - (\log N)^2,
		\end{flalign*}
		
		\noindent where in the second statement we used the fact that $\varphi_t$ is a $\zeta$-localization center bijection for $\bm{P}$, and in the third we used \Cref{sij} and the fact that $\zeta \ge e^{-100 (\log N)^{3/2}} \ge e^{-(\log N)^2}$. To verify the first bound in \eqref{ukn1n2}, it therefore suffices to show that with overwhelming probability we have 
		\begin{flalign}
			\label{muq} 
			\displaystyle\sum_{\mu \in \eig \bm{Q}} \log |\mu - \lambda_k| \le \displaystyle\sum_{i : \varphi_t(i) < \varphi_t(k)} \log |\lambda_i - \lambda_k| + \displaystyle\frac{1}{3} \cdot (\log N)^6.
		\end{flalign}
		
		\noindent To that end, recalling \Cref{adelta}, define the event $\mathsf{E}_1 = \mathsf{BND}_{\bm{P}} (\log N) \cap \mathsf{SEP}_{\bm{P}} (e^{-(\log N)^2})$. Then $\mathsf{E}_1$ is overwhelmingly probable, due to \Cref{l0eigenvalues} and \Cref{eigenvalues0}. 
		
		We next apply \Cref{ltmatrix}, with the $\ell$ there equal to $\varphi_t(k)$ here; observe that the estimate \eqref{tn1n2} assumed in that proposition holds by \eqref{n1k}. Thus, \Cref{ltmatrix} yields an overwhelmingly probable event $\mathsf{E}_2$, on which the following holds. Let $\mu \in \eig \bm{Q}$ be any eigenvalue that admits a $\zeta$-localization center $\Phi \in \llbracket N_1, \varphi_t(k) - 1 \rrbracket$ with respect to $\bm{Q}$ satisfying $\Phi \le \varphi_t (k) - (\log N)^4$. 
		\begin{enumerate} 
			\item There exists a unique eigenvalue $\lambda = \lambda(\mu) \in \eig \bm{P}$ such that $|\lambda - \mu| \le e^{-(\log N)^3}$. 
			\item Letting $\lambda = \lambda_{\psi}$ for some $\psi \in \llbracket 1, N \rrbracket$, we have $| \varphi_t(\psi) - \Phi | \le T (\log N)^3$.
			\item If $\Phi \ge N_1 + T (\log N)^3 + N^{1/100} / 2$, then we further have $| \varphi_t(\psi) - \Phi | \le (\log N)^2$.
		\end{enumerate}
		
		\noindent In what follows, we restrict to $\mathsf{E}_1 \cap \mathsf{E}_2$ and show \eqref{muq} holds. 
		
		To that end, denote $\eig \bm{Q} = (\mu_1, \mu_2, \ldots , \mu_{\varphi_t(k) - N_1})$, and let $\Phi : \llbracket 1, \varphi_t (k) - N_1 \rrbracket \rightarrow \llbracket N_1, \varphi_t (k) - 1 \rrbracket$ denote a $(2N)^{-1}$-localization center bijection for $\bm{Q}$, which is guaranteed to exist by \Cref{bijectionm}. Let $\mathcal{J}$ denote the set of indices $j \in \llbracket 1, \varphi_t (k) - N_1 \rrbracket$ such that $\Phi (j) \le \varphi_t (k) - (\log N)^4$. For each $j \in \mathcal{J}$, let $\psi_j \in \llbracket 1, N \rrbracket$ denote the unique index such that 
		\begin{flalign}
			\label{lambdapsimu} 
			|\lambda_{\psi_j} - \mu_{\Phi(j)}| \le e^{-(\log N)^3}, \qquad \text{where $j$ satisfies} \qquad \varphi_t (\psi_j) < \varphi_t (k).
		\end{flalign}
		
		\noindent Such an index $\psi_j$ satisfying the first bound in \eqref{lambdapsimu} exists and is unique by our restriction to $\mathsf{E}_2$. To verify that $\psi_j$ satisfies the second observe that, if $\Phi(j) \ge N_1 + T (\log N)^3 + N^{1/100} / 2$, then 
		\begin{flalign*} 
			\varphi_t(\psi_j) \le \Phi(j) + (\log N)^2 \le \varphi_t(k) - (\log N)^4 + (\log N)^2 < \varphi_t(k),
		\end{flalign*} 
		
		\noindent the first bound since we restricted to $\mathsf{E}_2$, the second since $j \in \mathcal{J}$, and the third since $N > 1$. If instead $\Phi (j) < N_1 + T (\log N)^3 + N^{1/100} / 2$, then 
		\begin{flalign*} 
			\varphi_t (\psi_j) \le \Phi (j) + T (\log N)^3 \le N_1 + 2T(\log N)^4 + \displaystyle\frac{1}{2} \cdot N^{1/100}  < \varphi_t (k), 
		\end{flalign*} 
		
		\noindent the first bound by our restriction to $\mathsf{E}_2$, the second by our assumption on $\Phi(j)$, and the third by \eqref{n1k}. This confirms \eqref{lambdapsimu}. We also have that 
		\begin{flalign}
			\label{muhlambdak}
			\begin{aligned}
				& \log |\mu_h - \lambda_k| \le \log N \qquad \qquad \qquad \qquad \qquad \quad \text{for all $h \in \llbracket 1, \varphi_t (k) - N_1 - 1 z\rrbracket$}; \\
				& \log |\mu_j - \lambda_k| \le \log |\lambda_{\psi_j} - \lambda_k| + e^{-(\log N)^2}, \qquad \text{for all $j \in \mathcal{J}$}.
			\end{aligned}
		\end{flalign}
		
		\noindent Here, the first bound holds since $|\mu_i - \lambda_i| \le 2 \log N$ (by our restriction to $\mathsf{E}_1$ and the fact that $\mathsf{E}_1 \subseteq \mathsf{BND}_{\bm{Q}} (\log N)$, by \Cref{mmj}); the second holds since $|\lambda_{\psi_j} - \mu_j| \le e^{-(\log N)^3}$, with the bound $|\lambda_{\psi_j} - \lambda_k| \ge e^{-(\log N)^2}$ (by our restriction to $\mathsf{E}_1$ and the fact from \eqref{lambdapsimu} that $\psi_j \ne k$). Therefore,  
		\begin{flalign*}
			\displaystyle\sum_{\mu \in \bm{Q}} \log |\mu - \lambda_k| & \le \displaystyle\sum_{j \in \mathcal{J}} \log |\mu_h - \lambda_k| + (\log N)^5 \\
			& \le  \displaystyle\sum_{j \in \mathcal{J}} \log |\lambda_{\psi_j} - \lambda_k| + 2 (\log N)^5 \\ 
			& \le \displaystyle\sum_{i : \varphi_t (i) < \varphi_t(k)} \log |\lambda_i - \lambda_k| + 3 (\log N)^5.
		\end{flalign*}
		
		\noindent Here, in the first estimate, we used the fact that at most $(\log N)^4$ indices $h \notin \mathcal{J}$ exist, with the first bound in \eqref{muhlambdak} for each such $h$. In the second, we used the second bound in \eqref{muhlambdak}. In the third we used the fact that $\varphi_t(\psi_j) < \varphi_t(k)$ for each $j \in \mathcal{J}$ (by \eqref{lambdapsimu}); the bound $\log |\lambda_i - \lambda_k| \le \log (2 \log N) \le \log N$ for all indices $i$ with $\varphi_t (i) \in \llbracket N_1, \varphi_t (k) - 1 \rrbracket$ but that are not of the form $\psi_j$ for some $j \in \mathcal{J}$; and the fact that at most $(\log N)^4$ such indices $i$ exist (as is quickly verified from the injectivity of $\varphi_t$ and $\psi$, the second statement in \eqref{lambdapsimu}, and the fact that there are at most $(\log N)^4$ indices not in $\mathcal{J}$). This establishes the second bound in \eqref{muq} and thus the lemma.
	\end{proof}

	\subsection{Eigenvalues of Truncated Lax Matrices}
	
	\label{MatrixEigenvalue0} 
	
	In this section we begin the proof of \Cref{ltmatrix}; we adopt the notation of that proposition throughout. We first address it when \eqref{n1n2center} holds, that is, when $\Phi$ is not too close to $N_1$ or $N_2$. 
	
	\begin{prop}
		
		\label{ltmatrix2} 
		
		Proposition \ref{ltmatrix} holds assuming \eqref{n1n2center}. 
		
	\end{prop} 
	
	To prove \Cref{ltmatrix2}, we first apply \Cref{ltl0}, which yields a random matrix $\bm{Q} = [Q_{ij}] \in \SymMat_{\llbracket N_1, N_2 \rrbracket}$ with the same law as $\bm{L}(0)$, and an overwhelmingly probable event $\mathsf{E}$, on which we have   
	\begin{flalign}
		\label{pq} 
		\displaystyle\max_{i, j \in \llbracket N_1 + K, N_2 - K \rrbracket} | P_{ij} - Q_{ij} | \le e^{-K/5}, \qquad \text{for any $K \ge T \log N$}.
	\end{flalign}
	
	\begin{lem} 
		
		\label{ltmatrix3}
		
		If \eqref{n1n2center} holds, then there exists a constant $c>0$ such that the following two statements hold with overwhelming probability. 
		
		\begin{enumerate} 
			\item There exists an eigenvalue $\kappa \in \eig \bm{Q}$ such that $|\mu - \kappa| \le e^{-c \min \{ |\Phi - \ell|, N^{1/100} \}}$. 
			\item The index $\Phi \in \llbracket N_1, N_2 \rrbracket$ is an $N^{-2} \zeta$-localization center for $\kappa$ with respect to $\bm{Q}$. 
		\end{enumerate}
	\end{lem} 
	
	\begin{proof} 
		
		Define $\bm{R} = [R_{ij}] \in \SymMat_{\llbracket N_1, N_2 \rrbracket}$ by setting $\bm{R} = \bm{Q}^{(\ell)}$. Throughout this proof, we assume that $\Phi \le \ell$, as the proof when $\Phi \ge \ell$ is entirely analogous. Then denote by $\bm{P}' = [P_{ij}']$ and $\bm{Q}' = [Q_{ij}']$ the top $(\ell-N_1) \times (\ell-N_1)$ corners of $\bm{P}$ and $\bm{Q}$, respectively. In this way, their rows and columns are indexed by $i, j \in \llbracket N_1, \ell-1 \rrbracket$, and we have $P_{ij}' = P_{ij} = M_{ij}$ and $Q_{ij}' = Q_{ij} = R_{ij}$ for each such $(i, j)$; observe that $\eig \bm{P}' \subseteq \eig \bm{M}$ and $\eig \bm{Q}' \subseteq \eig \bm{R}$. Since $\Phi \le \ell$, we also have $\mu \in \eig \bm{P}'$. 
		
		To establish the lemma, we will first use \Cref{lleigenvalues} to show $\mu$ is close to some $\nu \in \eig \bm{R}$. To that end, recalling \Cref{adelta}, define the quantities $K, \delta > 0$ and event $\mathsf{E}_1$ by 
		\begin{flalign}
			\label{kdeltae} 
			\begin{aligned}
				& K = \displaystyle\frac{1}{10} \cdot \min  \{ |\Phi - \ell| , N^{1/100} \} + \frac{T}{2} \cdot (\log N)^3; \\ 
				& \quad \delta = e^{-K/4}; \quad \mathsf{E}_1 = \bigcap_{s \ge 0} \mathsf{BND}_{\bm{L}(s)} \Big( \displaystyle\frac{\log N}{600} \Big). 
			\end{aligned}
		\end{flalign} 
		
		\noindent By \Cref{l0eigenvalues}, $\mathsf{E}_1$ holds with overwhelming probability, so we restrict to $\mathsf{E}_1$ in what follows.
		
		Applying \Cref{lleigenvalues} requires (through \Cref{lmatrixl}) a bound on $|P_{ij}' - Q_{ij}'|$ for $i,j$ sufficiently far from $N_1$. This will follow from \eqref{pq}; indeed, observe that $N_2 - \ell \ge T(\log N)^4 + N^{1/100} \ge K$, where the first bound holds by \eqref{tn1n2} and the second holds by \eqref{kdeltae}. Thus $N_2 - K \ge \ell$, so since $(P_{ij}', Q_{ij}') = (P_{ij}, Q_{ij})$ for $i,j \le \ell$, it follows from \eqref{pq} that 
		\begin{flalign}
			\label{pq2} 
			\displaystyle\max_{i,j \in \llbracket N_1+K, \ell \rrbracket} |P_{ij}' - Q_{ij}'| \le \delta.
		\end{flalign} 
		
		Using this, we next verify \Cref{lmatrixl} with the $(\bm{L}, \tilde{\bm{L}})$ there given by $(\bm{Q}', \bm{P}')$ here. To that end, observe that $|P_{ij}| \le (\log N) / 50 \le 2 \log (\ell - N_1)$, where the first bound follows from \Cref{abltt} and our restriction to $\mathsf{E}_1$, and the second follows from \eqref{tn1n2}. This, together with \eqref{pq2} and the fact that $\bm{Q}$ has the same law as $\bm{L}(0)$, implies that \Cref{lmatrixl} holds, with the $(\bm{L}, \tilde{\bm{L}}; \mathcal{D}; N)$ there equal to $(\bm{Q}', \bm{P}'; \llbracket N_1, N_1 + K-1 \rrbracket; \ell-N_1 )$ here. Since \eqref{n1n2center} and \eqref{kdeltae} together yield
		\begin{flalign}
			\label{n1k2} 
			\Phi - N_1 \ge T (\log N)^3 + \displaystyle\frac{1}{2} \cdot N^{1/100} \ge 2K + (\log N)^3,
		\end{flalign} 	
		
		\noindent we deduce that $\Phi \ge N_1 + K + (\log N)^3$, verifying \eqref{dn}. Since $\mu \in \eig \bm{P}'$, the first statement of \Cref{lleigenvalues} yields a constant $c_1 > 0$ such that, with overwhelming probability, there exists an eigenvalue $\nu \in \eig \bm{Q}' \subseteq \eig \bm{R}$ satisfying 
		\begin{flalign}
			\label{numu} 
			|\nu - \mu| \le e^{(\log N)^2} (\delta^{1/16} + e^{-2c_1 |\Phi - N_1 - K|}) \le e^{-c_1 K},
		\end{flalign} 
		
		\noindent where in the last inequality we used \eqref{n1k2} and the bound $K \ge (\log N)^4 / 10$ (as $|\Phi - \ell| \ge (\log N)^4$). The second statement of \Cref{lleigenvalues} further impies that $\Phi$ is an $N^{-1} \zeta$-localization center of $\nu$ with respect to $\bm{Q}'$, and thus with respect to $\bm{R}$. We further restrict to $\mathsf{E}_2$ in what follows.
		
		We now use \Cref{lleigenvalues} again to show that $\nu$ is close to some $\kappa \in \eig \bm{Q}$; this will proceed similarly to above. Since $R_{ij} = Q_{ij}$ unless $\ell \in \{ i, j \}$ (and $\bm{Q}$ has the same law as $\bm{L}(0)$), we have by our restriction to $\mathsf{E} \cap \mathsf{E}_1$ that \Cref{lmatrixl} holds with the $(\bm{L}, \tilde{\bm{L}}; \delta; \mathcal{D})$ there equal to $( \bm{Q}, \bm{R}; 0; \{ \ell \} )$ here. Since $|\Phi - \ell| \ge  (\log N)^4$, the first statement of \Cref{lleigenvalues} yields a constant $c_2 > 0$ such that, with overwhelming probability, there exists an eigenvalue $\kappa \in \eig \bm{Q}$ such that 
		\begin{flalign}
			\label{kappanu} 
			|\kappa - \nu| \le e^{(\log N)^2} \cdot e^{-2c_2 |\Phi - \ell|} \le e^{-c_2 |\Phi - \ell|},
		\end{flalign} 
		
		\noindent where in the last inequality we used the fact that $|\Phi - \ell| \ge (\log N)^4$. Thus, \eqref{numu} and \eqref{kappanu} imply the first statement of the lemma. The second statement of \Cref{lleigenvalues} further implies that $\Phi$ is an $N^{-2}\zeta$-localization center for $\kappa$ with respect to $\bm{Q}$, verifying the second statement of the lemma.
	\end{proof}
	
	\begin{proof}[Proof of \Cref{ltmatrix2}]
		
		Throughout this proof, we adopt the notation and assumptions from \Cref{ltmatrix3} and its proof. In particular, we recall the quantities $K, \delta > 0$ and the event $\mathsf{E}_1$ from \eqref{kdeltae}. Recalling \Cref{adelta}, we further define the event
		\begin{flalign}
			\label{eventf1} 
			\mathsf{F}_1 = \mathsf{SEP}_{\bm{P}} (e^{-(\log N)^2}) \cap \mathsf{SEP}_{\bm{Q}} (e^{-(\log N)^2}).
		\end{flalign}
		
		\noindent By \Cref{ltt}, \Cref{eigenvalues0}, and a union bound, $\mathsf{F}_1$ is overwhelmingly probable. We restrict to $\mathsf{E}_1 \cap \mathsf{F}_1$, and further to the event that \Cref{ltmatrix3} holds, in what follows. 
		
		We will use \Cref{lleigenvalues2} to show that $\kappa$ is close to some $\lambda \in \eig \bm{P}$. To that end, first observe from \eqref{pq} and our restriction to $\mathsf{E} \cap \mathsf{E}_1$ (with the fact that $\bm{Q}$ has the same law as $\bm{L}(0)$) that \Cref{lmatrixl} holds, with the $(\bm{L}, \tilde{\bm{L}}; \mathcal{D})$ there equal to $( \bm{Q}, \bm{P}; \llbracket N_1, N_2 \rrbracket \setminus \llbracket N_1 + K, N_2 + K \rrbracket)$ here. Also, by \eqref{n1n2center}, $\min \{ \Phi - N_1 - K, N_2 - \Phi - K \} \ge N^{1/100} / 2 \ge (\log N)^3$, which verifies \eqref{dn2}. Therefore, \Cref{lleigenvalues2} applies and yields a constant $c_1 > 0$ and an overwhelmingly probable event $\mathsf{F}_2$, on which the following holds. There exists $\lambda \in \eig \bm{P}$ such that $\Phi$ is a $N^{-3} \zeta$-localization center of $\lambda$ with respect to $\bm{P}$, and 
		\begin{flalign}
			\label{lambdakappa} 
			|\lambda - \kappa| \le e^{(\log N)^2} (\delta^{1/8} + e^{-2c_1 \min \{ \Phi - N_1-K, N_2-K-\Phi \}}) \le e^{-c_1 K},
		\end{flalign}
		
		\noindent where in the last inequality we used the facts that $\min \{ \Phi - N_1, N_2 - \Phi \} \ge T(\log N)^3 + N^{1/100} / 2 \ge 2K + (\log N)^3$ and $|\Phi - \ell| \ge (\log N)^4$. Together with the first statement in \Cref{ltmatrix3}, this implies the existence of $\lambda$ satisfying the conditions in the first statement of \Cref{ltmatrix}, upon additionally restricting to $\mathsf{F}_2$. The uniqueness of such a $\lambda \in \eig \bm{P}$ follows follows from our restriction to the event $\mathsf{F}_1$ from \eqref{eventf1} (with the fact that $|\Phi - \ell| \ge (\log N)^4$).
		
		It remains to establish the third part of the proposition (which implies the second, as we have assumed that \eqref{n1n2center} holds), to which we must use \Cref{lleigenvalues} (instead of \Cref{lleigenvalues2}). Fix any $\zeta$-localization center $\varphi \in \llbracket N_1, N_2 \rrbracket$ for $\lambda$ with respect to $\bm{P}$. Let a $\bm{u}(t) = ( u(N_1; t), u(N_1+1; t), \ldots , u(N_2; t) ) \in \mathbb{R}^N$ be a unit eigenvector of $\bm{P}$ with eigenvalue $\lambda$. Then, by the definition of $\varphi$ and \Cref{lleigenvalues2}, we have 
		\begin{flalign}
			\label{upsi} 
			| u(\varphi; t) | \ge \zeta, \qquad \text{and} \qquad | u(\Phi; t) | \ge N^{-3} \zeta.
		\end{flalign}
		
		\noindent We must show \eqref{dn}, to which end we claim that, with high probability,
		\begin{flalign}
			\label{n1n2k}  
			\min \{ \varphi - N_1, N_2 - \varphi \} \ge K + (\log N)^3.
		\end{flalign} 
		
		\noindent To do so, observe that \Cref{centert} yields an overwhelmingly probable event $\mathsf{F}_3$, on which the following holds. An index $m \in \llbracket N_1, N_2 \rrbracket$ can only be a $N^{-3} \zeta$-localization center for $\lambda$ with respect to $\bm{P}$ if $|m-\varphi| \le T (\log N)^2$. Restricting to $\mathsf{F}_3$, the second bound in \eqref{upsi} then implies that we must have $|\Phi - \varphi| < T  (\log N)^2$ for sufficiently large $N$. Together with the estimate  (by \eqref{n1n2center} and \eqref{kdeltae})
		\begin{flalign*} 
			\min \{ \Phi - N_1, N_2 - \Phi \} & \ge  T (\log N)^3 + \displaystyle\frac{1}{2} \cdot N^{1/100} \ge  K + T(\log N)^2 + (\log N)^3,
		\end{flalign*} 
		
		\noindent this confirms \eqref{n1n2k}. 
		
		The bound \eqref{n1n2k} (with \eqref{pq}) verifies \eqref{dn}, with the parameters $(\bm{L}, \tilde{\bm{L}}; \tilde{\lambda}, \tilde{\varphi}; \delta; \mathcal{D})$ there equal to $( \bm{Q}, \bm{P}; \lambda, \varphi; \delta; \llbracket N_1, N_2 \rrbracket \setminus \llbracket N_1+K, N_2-K \rrbracket )$ here. Thus, \Cref{lleigenvalues} applies and yields a constant $c_2 > 0$ such that, with overwhelming probability, there exists $\kappa' \in \eig \bm{Q}$ with 
		\begin{flalign*} 
			|\kappa' - \lambda| \le e^{(\log N)^2} (\delta^{1/8} + e^{-c_2 \min \{ \varphi - N_1 - K, N_2 - K - \varphi \}}) \le e^{-c_2 (\log N)^3},
		\end{flalign*} 
		
		\noindent where in the last inequality we used the definition \eqref{kdeltae} of $\delta$ (with the fact that $|\Phi-\ell| \ge (\log N)^4$) with \eqref{n1n2k}. Together with \eqref{lambdakappa}, the fact that $|\Phi-\ell| \ge (\log N)^4$, and our restriction to the event $\mathsf{F}_1$ from \eqref{eventf1}, this implies that $\kappa = \kappa'$. Hence, the second part of \Cref{lleigenvalues} implies that $\varphi$ is a $N^{-1} \zeta$-localization center for $\kappa$ with respect to $\bm{Q}$. Since $\Phi$ is an $N^{-2} \zeta$-localization center for $\kappa$ with respect to $\bm{Q}$ (by \Cref{ltmatrix3}), it also implies that $|\varphi - \Phi| \le (\log N)^2$, verifying the third (and thus also the second) statement of the proposition.
	\end{proof}

	\subsection{Proof of \Cref{ltmatrix}} 
	
	\label{ProofMatrix0} 
	
	In this section we establish \Cref{ltmatrix}, adopting its notation and assumptions throughout. For any $z \in \mathbb{C}$, denote the resolvents of $\bm{M}$ and $\bm{P}$ by $\breve{\bm{G}} (z) = [ \breve{G}_{ij} (z) ] = (\bm{M} - z)^{-1}$ and $\check{\bm{G}} (z) = [ \check{G}_{ij} (z) ] = (\bm{P} - z)^{-1}$, respectively. The following lemma estimates the $(\Phi, \Phi)$ entry of $\check{\bm{G}} - \breve{\bm{G}}$, assuming that \eqref{n1n2center} does not hold. Its proof is similar to that of \Cref{lgdifference} (by using \Cref{centert} in place of \Cref{gijexponential}), so we only outline it.

	\begin{lem}
		
		\label{zomega} 
		
		There exists a constant $c > 0$ so that the following holds with overwhleming probability. Suppose that \eqref{n1n2center} does not hold, and let $\eta \in \mathbb{R}$ satisfy $e^{-cN^{1/100}} \le \eta \le 1$. Denoting $\Omega = \{ z \in \mathbb{C} : -N \le \Real z \le N, \eta \le \Imaginary z \le 1 \}$, we have
		\begin{flalign}
			\label{zomega2} 
			\displaystyle\sup_{z \in \Omega} \big| \check{G}_{\Phi \Phi} (z) - \breve{G}_{\Phi \Phi} (z) \big| \le e^{-c N^{1/100}}.
		\end{flalign}
		
	\end{lem} 
	
	\begin{proof}[Proof of \Cref{zomega} (Outline)]
		
		Throughout this proof, recalling \Cref{adelta}, we set
		\begin{flalign}
			\label{eevent} 
			\mathsf{F} = \mathsf{BND}_{\bm{P}} (\log N), \qquad \text{so that} \qquad \mathbb{P} [ \mathsf{F}^{\complement} ] \le c_1^{-1} e^{-c_1 (\log N)^2},
		\end{flalign}
		
		\noindent for some constant $c_1 \in (0, 1)$, by \Cref{l0eigenvalues}. We begin by using \Cref{centert} to estimate the eigenvectors of $\bm{P}$. To that end, first observe, by \Cref{ltt} and the fact that $\bm{P} = \bm{L}(t)$, that $\eig \bm{P} = (\lambda_1, \lambda_2, \ldots , \lambda_N)$. For each $j \in \llbracket 1, N \rrbracket$, recall that $\bm{u}_j (t) = (u_j (N_1; t), u_j (N_1+1;t), \ldots , u_j (N_2;t)) \in \mathbb{R}^N$ denotes the nonnegatively normalized, unit eigenvector of $\bm{P}$ with eigenvalue $\lambda_j$. For each $j \in \llbracket 1, N \rrbracket$, also let $\psi_j \in \llbracket N_1, N_2 \rrbracket$ denote a $N^{-1/2}$-localization center for the eigenvector $\bm{u}_j$ of $\bm{L}$ of eigenvalue $\lambda_j$ (which is guaranteed to exist, as $\bm{u}$ is a unit vector). Since \eqref{n1n2center} does not hold, \eqref{tn1n2} implies for sufficiently large $N$ that 
		\begin{flalign*}
			|\ell - \Phi| \ge K_0, \qquad \text{where} \qquad K_0 = \displaystyle\frac{1}{2} \cdot \big( T(\log N)^4 + N^{1/100} \big).
		\end{flalign*} 
		
		\noindent Thus, for any $i \in \{ \ell-1, \ell, \ell+1\}$ and $k \in \llbracket N_1, N_2 \rrbracket$, we either have $|\Phi - \psi_j| \ge K_0 / 2 \ge T (\log N)^2$ or $|i - \psi_j| \ge K_0 / 2 \ge T (\log N)^2$. Therefore, \Cref{centert} yields a  constant $\mathfrak{c} \in (0, c_1 / 25)$ and an event $\mathsf{E}_0$ with $\mathbb{P} [ \mathsf{E}_0^{\complement} ] \le \mathfrak{c}^{-1} e^{-25\mathfrak{c} (\log N)^2}$, such that 
		\begin{flalign}
			\label{vv} 
			\mathbbm{1}_{\mathsf{E}_0} \cdot \displaystyle\max_{i:|i-\ell| \le 1} \displaystyle\sum_{k \in \llbracket 1, N \rrbracket} | u_k (\Phi; t) | \cdot | u_k (i; t) | \le N \cdot e^{-25 \mathfrak{c} K_0} \le e^{-12 \mathfrak{c} N^{1/100}}.		
		\end{flalign} 
		
		\noindent Set $\mathsf{E} = \mathsf{E}_0 \cap \mathsf{F}$, which satisfies $\mathbb{P} [ \mathsf{E}^{\complement} ] \le \mathfrak{c}^{-1} e^{-20 \mathfrak{c} (\log N)^2}$, by \eqref{eevent} and a union bound.
		
		Now, let us show that \eqref{zomega2} holds with high probability for a fixed point $z \in \Omega$. Fix $z_0 \in \Omega$; abbreviate $\check{G}_{ij} = \check{G}_{ij} (z_0)$ and $\breve{G}_{ij} = \breve{G}_{ij} (z_0)$ for each $i, j \in \llbracket N_1, N_2 \rrbracket$. For $\eta \ge e^{-\mathfrak{c} N^{1/100}}$, we have
		\begin{flalign}
			\label{ge}
			\begin{aligned} 
				\mathbb{E} \big[ \mathbbm{1}_{\mathsf{E}} \cdot |\check{G}_{\Phi \Phi} - \breve{G}_{\Phi \Phi} | \big] & \le \displaystyle\sum_{\substack{i: |i-\ell| \le 1 \\ j: |j-\ell| \le 1}} \mathbb{E} \big[ \mathbbm{1}_{\mathsf{E}} \cdot | \check{G}_{\Phi i} \cdot P_{ij} \cdot \breve{G}_{j\Phi}| \big] \\
				& \le 9 \eta^{-1} \cdot \log N \cdot \displaystyle\max_{i : |i-\ell| \le 1} \mathbb{E} [ \mathbbm{1}_{\mathsf{E}} \cdot |\check{G}_{\Phi i}| ] \\
				& \le 9 \eta^{-2} \cdot \log N \cdot \displaystyle\max_{i: |i-\ell| \le 1} \displaystyle\sum_{k=1}^N \mathbb{E} \big[ | u_k (\Phi; t) | \cdot | u_k (i; t) | \big] \\
				& \le 9 \eta^{-2} \cdot \log N \cdot e^{- 12 \mathfrak{c} N^{1/100}} \le e^{-9\mathfrak{c} N^{1/100}},
			\end{aligned} 
		\end{flalign}
		
		\noindent where the first inequality follows from \eqref{ab}, together with the fact that $P_{ij} - M_{ij} \in \{ P_{ij}, 0 \}$ with $P_{ij} - M_{ij} = 0$ unless $i, j \in \{ \ell - 1, \ell, \ell + 1 \}$; the second from \eqref{gijeta} (with the fact that $\Imaginary z_0 \ge \eta$, as $z_0 \in \Omega$) and the definition \eqref{eevent} of $\mathsf{F}$; the third from \eqref{gsum}, together with the bound $|\lambda_k - z_0| \ge \Imaginary z_0 \ge \eta$ for any $k \in \llbracket 1, N \rrbracket$; the fourth from \eqref{vv}; and the fifth from the fact that $\eta \ge e^{-\mathfrak{c} N^{1/100}}$. The estimate \eqref{ge}, together with a Markov bound, verifies \eqref{zomega2} at $z = z_0$. 
		
		To extend it to all $z \in \Omega$, we first use a union bound to apply \eqref{ge} on an $e^{-3 \mathfrak{c} N^{1/100}}$-mesh $\Omega_0 \subset \Omega$. Then, using \eqref{ab} and \eqref{gijeta}, we approximate the resolvents of $\bm{M}$ and $\bm{P}$ at an arbitrary points $z \in \Omega$ by those at the nearest point $z_0 \in \Omega_0$. This is very similar to what was done at the end of the proof of \Cref{lgdifference}, so we omit further details.
	\end{proof}

	\begin{proof}[Proof of \Cref{ltmatrix}] 
		
		In what follows, we assume that \eqref{n1n2center} does not hold, as \Cref{ltmatrix2} establishes the proposition when it does. Letting $\mathfrak{c} > 0$ denote the constant $c / 10$ from \Cref{zomega}, set $\eta = e^{-2\mathfrak{c} N^{1/100}}$ and define the event (recalling \Cref{adelta})
		\begin{flalign*}
			\mathsf{G} & = \mathsf{BND}_{\bm{P}} (\log N) \cap \mathsf{SEP}_{\bm{P}} (e^{-(\log N)^2}) \\
			& \qquad \cap  \bigg\{ \displaystyle\sup_{E \in [-N, N]} \big| \check{G}_{\Phi \Phi} (E + \mathrm{i} \eta) - \breve{G}_{\Phi \Phi} (E + \mathrm{i} \eta) \big| \le e^{-10\mathfrak{c} N^{1/100}} \bigg\}.
		\end{flalign*}
		
		\noindent By \Cref{l0eigenvalues}, \Cref{eigenvalues0} (with \Cref{ltt}), and \Cref{zomega}, $\mathsf{G}$ is overwhelmingly probable. 
		
		Restricting to $\mathsf{G}_1$, we next apply \Cref{abgh} with the parameters $(\lambda, \varphi; \eta, \zeta, \delta; \bm{A}, \bm{B})$ there equal to $(\mu, \Phi; \eta, \zeta, e^{-10\mathfrak{c} N^{1/100}}; \bm{M}, \bm{P})$ here. The first bound in \eqref{deltaetachi} holds for sufficiently large $N$ by the facts that $\zeta \ge e^{-100 (\log N)^{3/2}}$ and $\eta = e^{-2\mathfrak{c} N^{1/100}}$; the second holds since $\Phi$ is a $\zeta$-localization center of $\mu$ with respect to $\bm{M}$; and the third holds by our restriction to $\mathsf{G}_1$. Hence, \Cref{abgh} yields an eigenvalue $\lambda \in \eig \bm{P}$ such that $|\lambda-\mu| \le 3N \zeta^{-2} \eta \le e^{-\mathfrak{c} N^{1/100}}$ (again as $\zeta \ge e^{-100 (\log N)^{3/2}}$ and $\eta = e^{-2\mathfrak{c} N^{1/100}}$). This shows the first statement of the proposition; it remains to verify the second.
		
		To that end, observe that \Cref{abgh} further indicates that $\Phi$ is an $N^{-1}\zeta$-localization center for $\lambda$ with respect to $\bm{P}$. Letting $\Psi \in \llbracket N_1, N_2 \rrbracket$ denote an $N^{-1} \zeta$-localization center of $\lambda$ with respect to $\bm{L}$, \Cref{centert} implies that the following holds with overwhelming probability. We have $|m - \Psi| \le T (\log N)^2$ for any $N^{-1} \zeta$-localization center $m \in \llbracket N_1, N_2 \rrbracket$ of $\lambda$ with respect to $\bm{P}$. Applying this for $m \in \{ \Phi, \varphi \}$ then yields $|\varphi - \Phi| \le 2T (\log N)^2 \le  T (\log N)^3$ for sufficiently large $N$; this confirms the second statement of the proposition. 
	\end{proof}

	\section{Properties of Localization Centers}
	
	\label{ProofsCenter}

	We establish \Cref{centerdistance} in \Cref{ProofCenterDistance} and \Cref{currentestimate} in \Cref{ProofCurrent}, after showing several estimates on spacings between particles in the Toda lattice in \Cref{SpacingQ}. Throughout this section, we adopt \Cref{lbetaeta}.

	\subsection{Proof of \Cref{centerdistance}} 
	
	\label{ProofCenterDistance}
	
	In this section we show \Cref{centerdistance} by proving the following generalization of it. Recall we adopt \Cref{lbetaeta} throughout.

	\begin{prop} 
		
		\label{centerdistance2} 
		
		Assume more generally that $\zeta \ge N^3 e^{-200 (\log N)^{3/2}}$. For any real number $\mathfrak{d}>0$, there exists a constant $c = c(\mathfrak{d}) >0$ such that the following holds with probability at least $1-c^{-1} e^{-c(\log N)^2}$. Fix real numbers $t, \tilde{t} \in [0, T]$; an eigenvalue $\lambda \in \eig \bm{L}$; and $N^{-1} \zeta$-localization centers $\varphi, \tilde{\varphi} \in \llbracket N_1, N_2 \rrbracket$ of $\lambda$ with respect to $\bm{L}(t)$ and $\bm{L}(\tilde{t})$, respectively. If $\varphi$ satisfies \eqref{tzeta} and $|t-\tilde{t}| \le e^{-\mathfrak{d}(\log N)^2}$ holds, then $|\varphi  - \tilde{\varphi} | \le (\log N)^3$. 
		
	\end{prop}

	\begin{proof} 
		
		We may assume in what follows that $T \ge 3$ and that $\mathfrak{d}$ is sufficiently small, in a way to be determined later. We first restrict to several events.
		
		Let $\mathcal{T} \subseteq [0, T]$ denote an $e^{-\mathfrak{d}(\log N)^2}$-mesh of $[0,T]$, and fix $s \in \mathcal{T}$. By \Cref{ltl0}, there exists a constant $c_1>0$; an event $\mathsf{E}_1(s)$ with $\mathbb{P}[\mathsf{E}_1 (s)^{\complement}] \le c_1^{-1} e^{-c_1 (\log N)^2}$; and a random matrix $\bm{M} (s) = [M_{ij} (s)] \in \SymMat_{\llbracket N_1, N_2 \rrbracket}$ with the same law as $\bm{L}(0)$, such that the following holds on $\mathsf{E}_1 (s)$. Denoting $N_1' = N_1 + T (\log N)^{5/2}$ and $N_2' = N_2 - T(\log N)^{5/2}$, we have 
		\begin{flalign}
			\label{mijlijs} 
			\displaystyle\max_{i,j \in \llbracket N_1', N_2' \rrbracket}  |M_{ij} (s) - L_{ij}(s)| \le e^{-(\log N)^{5/2}/5}.
		\end{flalign} 
		
		\noindent Corollary \ref{lleigenvalues} (with the $(\bm{L}, \tilde{\bm{L}}; \mathcal{D}; \delta; \tilde{\lambda})$ there equal to $(\bm{M}, \bm{L}(s); \llbracket N_1, N_2 \rrbracket \setminus \llbracket N_1', N_2' \rrbracket; e^{-(\log N)^{5/2}/5}; \lambda)$ here) therefore yields a constant $c_2 > 0$ and an event $\mathsf{E}_2 (s)$ for each $s \in \mathcal{T}$, with $\mathbb{P}[\mathsf{E}_2 (s)^{\complement}] \le c_2^{-1} e^{-c_2(\log N)^2}$, such that the following holds on $\mathsf{E}_2 (s)$. There exists a unique eigenvalue $\mu = \mu(s) \in \eig \bm{M}(s)$ so that, for any $N^{-2} \zeta$-localization center $\varphi_s \in \llbracket N_1' + (\log N)^3, N_2' - (\log N)^3 \rrbracket$ for $\lambda$ with respect to $\bm{L}(s)$ and any $N^{-3} \zeta$-localization center $\psi_s \in  \llbracket N_1, N_2 \rrbracket$ for $\mu(s)$ with respect to $\bm{M}(s)$, we have $|\varphi_s - \psi_s| \le (\log N)^2$.
		
		Set $\mathfrak{d}' = \mathfrak{d} / 8$. Recalling \Cref{adelta}, let $\mathsf{E}_3 = \bigcap_{r \ge 0} \mathsf{BND}_{\bm{L}(r)} (\log N) \cap \mathsf{SEP}_{\bm{L}(0)} (e^{-\mathfrak{d}'(\log N)^2})$, so \Cref{l0eigenvalues} and \Cref{eigenvalues0} give a constant $c_3 = c_3 (\mathfrak{d}) >0$ with $\mathbb{P}[\mathsf{E}_3^{\complement}] \le c_3^{-1} e^{-c_3 (\log N)^2}$. Further let $\mathsf{E}_4$ denote the event on which \Cref{centert} holds, which satisfies $\mathbb{P}[\mathsf{E}_4^{\complement}] \le c_4^{-1} e^{-c_4 (\log N)^2}$ for some $c_4 > 0$. We restrict to the event $\mathsf{E} = \bigcap_{s \in \mathcal{T}} (\mathsf{E}_1 (s) \cap \mathsf{E}_2 (s)) \cap \mathsf{E}_3 \cap \mathsf{E}_4$ in what follows, which for $\mathfrak{d} < \min \{ c_1/2, c_2/2 \}$ satisfies $\mathbb{P} [\mathsf{E}^{\complement}] \le c_5^{-1} e^{-c_5 (\log N)^2}$ for some $c_5 >0$, by a union bound.
		
		By \Cref{centert} (and our restriction to $\mathsf{E}_4$), we have for any $N^{-1} \zeta$-localization center $\varphi_0$ for $\lambda$ with respect to $\bm{L}(0)$ that $|\varphi - \varphi_0| \le T(\log N)^2$ and $|\tilde{\varphi} - \varphi_0| \le T(\log N)^2$. Therefore, $|\varphi-\tilde{\varphi}| \le 2T(\log N)^2$, so \eqref{tzeta} implies that 
		\begin{flalign}
			\label{tzeta2} 
			N_1 + \displaystyle\frac{T}{2} \cdot (\log N)^3 \le  \tilde{\varphi} \le N_2 - \displaystyle\frac{T}{2} \cdot (\log N)^3.
		\end{flalign}
		
		Now, fix $t \in [0, T]$, and let $s \in \mathcal{T}$ satisfy $|s-t| \le e^{-\mathfrak{d} (\log N)^2}$, so $|s-t'| \le 2 e^{-\mathfrak{d} (\log N)^2}$. We will first apply \Cref{abgh}, with the $(\bm{A}; \bm{B})$ there equal to $(\bm{L}(t); \bm{L}(s))$ here, to show that $\varphi$ and $\tilde{\varphi}$ are localization centers for $\lambda$ with respect to $\bm{L}(s)$. Then, we will apply our restriction to $\mathsf{E}_2 (s)$ to deduce that $\varphi$ and $\tilde{\varphi}$ are close to a localization center of $\mu(s)$ with respect to $\bm{M}(s)$, and thus are close to each other.
		
		To implement this, we require the third estimate in \eqref{deltaetachi}. So, for any $r \in [0, T]$ and $z \in \mathbb{C}$, denote the resolvent $G(z;r) = [G_{ij} (z;r)] = (\bm{L}(r)-z)^{-1} \in \Mat_{\llbracket N_1, N_2 \rrbracket}$. Observe that 
		\begin{flalign*}
			\displaystyle\max_{i,j \in \llbracket N_1, N_2 \rrbracket} | L_{ij} (s) - L_{ij} (t) | \le 2e^{-\mathfrak{d} (\log N)^2} \cdot \displaystyle\max_{i,j \in \llbracket N_1, N_2 \rrbracket} | L_{ij}' (t) | \le  e^{-7\mathfrak{d}' (\log N)^2},
		\end{flalign*}
		
		\noindent where in the second inequality we used \eqref{qtpt}, \eqref{derivativepa}, \Cref{matrixl}, and our restriction to $\mathsf{E}_3$ (which together imply that $|L_{ij}' (t)| \le 2 (\log N)^2$). With \eqref{gijeta} and \eqref{ab}, this implies that 
		\begin{flalign*}
			\displaystyle\sup_{z \in \Omega} \displaystyle\max_{i,j \in \llbracket N_1, N_2 \rrbracket} | G_{ij} (z;t) - G_{ij} (z;s) | \le N^2 \cdot e^{4\mathfrak{d}' (\log N)^2} \cdot e^{-7 \mathfrak{d}' (\log N)^2} \le e^{-2\mathfrak{d}' (\log N)^2},
		\end{flalign*}
		
		\noindent where we have denoted $\Omega = \{ z \in \mathbb{C} : e^{-2\mathfrak{d}' (\log N)^2} \le \Imaginary z \le 1 \}$.

		Now we apply \Cref{abgh}, with the parameters $(\bm{A},\bm{B}; \lambda; \varphi; \eta, \delta)$ there equal to the parameters $(\bm{L}(t),\bm{L}(s); \lambda; \varphi; e^{-2\mathfrak{d}' (\log N)^2}, e^{-2\mathfrak{d}'(\log N)^2})$ here. This yields an eigenvalue $\lambda' \in \eig \bm{L}(s) = \eig \bm{L}(t)$ such that $|\lambda - \lambda'| \le 3N^2 \zeta^{-1} e^{-2\mathfrak{d}' (\log N)^2} < e^{-\mathfrak{d}' (\log N)^2}$ and $\varphi$ is a $N^{-2} \zeta$-localization center for $\lambda'$ with respect to $\bm{L}(s)$. Due to our restriction to $\mathsf{E}_3 \subseteq \mathsf{SEP}_{\bm{L}(s)} (e^{-\mathfrak{d}' (\log N)^2})$, we have that $\lambda = \lambda'$, meaning that $\varphi$ is an $N^{-2} \zeta$-localization center for $\lambda$ with respect to $\bm{L}(s)$. By similar reasoning, $\tilde{\varphi}$ is an $N^{-2}\zeta$-localization center for $\lambda$ with respect to $\bm{L}(s)$. 
		
		By our restriction to $\mathsf{E}_2 (s)$, and the fact from \eqref{tzeta} and \eqref{tzeta2} that $\varphi, \tilde{\varphi} \in \llbracket N_1' + (\log N)^3, N_2' - (\log N)^3 \rrbracket$, it follows that for any $N^{-3} \zeta$-localization center $\psi_s \in \llbracket N_1, N_2 \rrbracket$ of $\mu(s)$ with respect to $\bm{M}(s)$ we have $|\varphi - \psi_s| \le (\log N)^2$ and $|\tilde{\varphi} - \psi_s| \le (\log N)^2$. Hence, $|\varphi - \tilde{\varphi}| \le 2(\log N)^2$, confirming the proposition.
	\end{proof}

	\begin{proof}[Proof of \Cref{centerdistance}]
		
		This follows from the $\tilde{t}=t$ case of \Cref{centerdistance2}.
	\end{proof}

	\subsection{Spacing Bounds for the Toda Particles} 
	
	\label{SpacingQ}
	
	In this section we prove the following lemma approximating the distances between the Toda particles $q_j(s)$ under thermal equilibrium. Recall we adopt \Cref{lbetaeta} throughout.

	\begin{lem} 
		
		\label{qijsalpha} 
		
		The following two statements hold with overwhelming probability. 
		
		\begin{enumerate} 
			
			\item For any $s \in [0,T]$ and $i, j \in \llbracket N_1 + T(\log N)^3, N_2 - T (\log N)^3 \rrbracket$, we have
			\begin{flalign}
				\label{qiqjs4}
				\big| q_i (s) - q_j (s) - \alpha (i-j) \big| \le |i-j|^{1/2} (\log N)^2.
			\end{flalign}
			
			\item For any $s \in [0,T]$ and $i, j \in \llbracket N_1, N_2 \rrbracket$ with $|i-j| \ge T (\log N)^5$, we have 
			\begin{flalign}
				\label{qiqjs5}
				\big(  q_i (s) - q_j (s) \big) \cdot \sgn (\alpha i - \alpha j) \ge \displaystyle\frac{|\alpha|}{2} \cdot |i-j|.
			\end{flalign}
			
		\end{enumerate} 
		
	\end{lem} 
	
	\begin{proof}
		
		Observe that it suffices to show for any fixed $s \in [-T, T]$ that the below two statements hold with overwhelming probability. First, for any $i, j \in \llbracket N_1 + T (\log N)^3, N_2 - T(\log N)^3 \rrbracket$,  
		\begin{flalign}
			\label{qiqjs2} 
			|q_i (s) - q_j (s) - \alpha(i-j) | \le \displaystyle\frac{1}{2} \cdot |i-j|^{1/2} (\log N)^2.
		\end{flalign} 
		
		\noindent Second, for any $i \in \llbracket N_1, N_2 \rrbracket$ with $|i-j| \ge T (\log N)^5$,  
		\begin{flalign}
			\label{qiqjs3} 
			\big(q_i (s) - q_j (s) \big) \cdot \sgn (\alpha i - \alpha j) \ge \displaystyle\frac{3|\alpha|}{4} \cdot |i-j|.
		\end{flalign} 
		
		\noindent Indeed, if this were true then applying a union bound, we may restrict to the event $\mathsf{E}_1$ on which \eqref{qiqjs2} and \eqref{qiqjs3} hold for all $s \in \mathcal{T}$, where $\mathcal{T} \subset [-T,T]$ is an $N^{-20}$-mesh of $[-T,T]$. Recalling \Cref{adelta}, we may further by \Cref{l0eigenvalues} restrict to the event $\mathsf{E}_2 = \bigcap_{s \ge 0} \mathsf{BND}_{\bm{L}(s)} (\log N)$. In view of \eqref{qtpt}, \eqref{abr}, \Cref{matrixl}, and our restriction to $\mathsf{E}_2$, we have $|q_k'(s)| \le \log N$ for all $s \ge 0$ and $k \in \llbracket N_1, N_2 \rrbracket$. Then, the two statements of the lemma hold for any $s \in [-T, T]$. Indeed, letting $s' \in \mathcal{T}$ be such that $|s-s'| \le N^{-20}$, we have from \eqref{qiqjs2} and the above bound on $|q_k'(s)|$ that
		\begin{flalign*}
			| q_i (s) - q_j (s) - \alpha(i-j) | & \le | q_i (s') - q_j (s') - \alpha (i-j) | + 2 N^{-20} \log N \\
			& \le \displaystyle\frac{1}{2} \cdot |i-j|^{1/2} (\log N)^2 + 2N^{-20} \log N \le |i-j|^{1/2} (\log N)^2,
		\end{flalign*} 
		
		\noindent which confirms \eqref{qiqjs4}. The verification of \eqref{qiqjs5} from \eqref{qiqjs3} is entirely analogous and thus omitted.
		
		Hence, it remains to show \eqref{qiqjs2} and \eqref{qiqjs3}. Throughout the remainder of this proof, we set $K = T (\log N)^3$. By \Cref{ltl0}, there is a family of random variables $\tilde{\bm{a}} = (\tilde{a}_{N_1}, \tilde{a}_{N_1+1}, \ldots , \tilde{a}_{N_2-1}) \in \mathbb{R}^{N-1}$ and $\tilde{\bm{b}} = (\tilde{b}_{N_1}, \tilde{b}_{N_1+1}, \ldots , \tilde{b}_{N_2}) \in \mathbb{R}^N$, such that $(\tilde{\bm{a}}; \tilde{\bm{b}})$ has the same law as $(\bm{a}; \bm{b})$ and the following holds. There exists an overwhelmingly probable event $\mathsf{E}_1$, on which 
		\begin{flalign}
			\label{iaa} 
			\displaystyle\max_{i \in \llbracket N_1 + K, N_2 - K \rrbracket} | a_i (s) - \tilde{a}_i | + \displaystyle\max_{i \in \llbracket N_1 + K, N_2 - K \rrbracket} | b_i (s) - \tilde{b}_i | \le 2e^{-(\log N)^3/5}.
		\end{flalign}
		
		\noindent We restrict to $\mathsf{E}_1$ in what follows and let $(\tilde{\bm{p}}; \tilde{\bm{q}}) \in \mathbb{R}^N \times \mathbb{R}^N$ denote the Toda state space initial data associated with $(\tilde{\bm{a}}; \tilde{\bm{b}})$, as described in \Cref{Open}. Further define event 
		\begin{flalign*}
			\mathsf{E}_2 = \bigcap_{i, i' \in \llbracket N_1, N_2 \rrbracket} \Big\{ | q_i (0) - q_{i'} (0) - \alpha (i-i') | + | \tilde{q}_i - \tilde{q}_{i'} - \alpha (i-i')| \le \displaystyle\frac{1}{4} \cdot |i-i'|^{1/2} (\log N)^2 \Big\},
		\end{flalign*} 
		
		\noindent observing that $\mathsf{E}_2$ is overwhelmingly probable, by \Cref{qij}. Recalling \Cref{adelta}, we also define the event 
		\begin{flalign*} 
			\mathsf{E}_3 =  \bigcap_{r \ge 0} \mathsf{BND}_{\bm{L}(r)} (\log N) \cap \bigcap_{i=N_1}^{N_2} \{ a_i (0) \ge e^{-(\log N)^2} \} \cap \{ \tilde{a}_i \ge e^{-(\log N)^2} \},
		\end{flalign*} 
		
		\noindent which is overwhelmingly probable, by \Cref{l0eigenvalues} and the explicit form of the density $\mu_{\beta,\theta;N-1,N}$ for $\bm{a}(0)$ (from \Cref{mubeta2}). Denoting the event $\mathsf{E} = \mathsf{E}_1 \cap \mathsf{E}_2 \cap \mathsf{E}_3$, we restrict to $\mathsf{E}$ in what follows; it then suffices to show that \eqref{qiqjs2} and \eqref{qiqjs3} hold. 
		
		To that end, first observe that, for any $i, i' \in \llbracket N_1+K, N_2-K \rrbracket$, we have
		\begin{flalign*}
			| q_i (s) - q_{i'} (s) - (\tilde{q}_i - \tilde{q}_{i'}) | & \le 2 \displaystyle\sum_{k=i}^{i'-1} | \log a_i (s) - \log \tilde{a}_i | \\
			& \le 2N \cdot \displaystyle\max_{k \in \llbracket i, i' \rrbracket} | a_i (s) - \tilde{a}_i | \cdot ( | a_i (s) |^{-1} + |\tilde{a}_i|^{-1} ) \le N^{-1}, 
		\end{flalign*}
		
		\noindent where in the first bound we used \eqref{q00}, and in the second and third we used \eqref{iaa} and our restriction to $\mathsf{E}_2 \cap \mathsf{E}_3$. This, together with our restriction to the event $\mathsf{E}_2$, establishes \eqref{qiqjs2} and thus the first statement of the lemma. To show the second, we suppose that $\alpha > 0$ and $i \ge j + T (\log N)^5$, as the proof when $\alpha < 0$ or $i \le j - T (\log N)^5$ is entirely analogous. Then, observe 
		\begin{flalign*}
			q_i (s) - q_j (s) & \ge q_i (0) - q_j (0) - | q_i (s) - q_i (0) | - | q_j (s) - q_j (0) | \\ 
			& \ge \alpha (i-j) - |i-j|^{1/2} (\log N)^2 - 2 s \cdot \displaystyle\max_{|s'| \le s} | b_i (s) | \\
			& \ge \displaystyle\frac{4\alpha}{5} \cdot (i-j) - 2 T \log N \ge \displaystyle\frac{\alpha}{2} \cdot (i-j),
		\end{flalign*}
		
		\noindent where in the first and second statements we used our restriction to $\mathsf{E}_2 \cap \mathsf{E}_3$, the first equality in \eqref{qtpt}, and the fact that $p_i (t) = b_i (t)$ by \eqref{abr}; in the third we used the fact that $i-j \ge T (\log N)^5$ and our restriction to $\mathsf{E}_2$; and in the fourth we again used the fact that $i-j \ge T (\log N)^5$. This confirms \eqref{qiqjs3} and thus the second statement of the lemma when $i \ge j + T (\log N)^5$; since (as mentioned above) the proof when $i \le j - T (\log N)^5$ is entirely analogous, this establishes the lemma.
	\end{proof} 
	
	\subsection{Proof of \Cref{currentestimate}}
	
	\label{ProofCurrent} 
	
	In this section we establish \Cref{currentestimate}. We begin with the following lemma proving a variant of \eqref{currentsum}, when the second sum appearing there is over $i$ such that $\varphi_t (i)$ (as opposed to $q_i (t)$ or $Q_i (t)$) is in a prescribed interval. Recall we adopt \Cref{lbetaeta} throughout.
	
	\begin{prop}
		
		\label{currentestimate2} 
		
		There exists a constant $c >0$ such that the following holds with probability at least $1 - c^{-1} e^{-c(\log N)^2}$. Let $N_1', N_2' \in \llbracket N_1, N_2 \rrbracket$ be indices satisfying
		\begin{flalign}
			\label{n1n20} 
			N_1 + T (\log N)^5 + N^{1/100} \le N_1' \le N_2' \le N_2 - T (\log N)^5 - N^{1/100}.
		\end{flalign}
		
		\noindent Then, the following two statements hold for any real number $t \in [0, T]$ and integer $m \in \llbracket 0, \log N \rrbracket$. 
		
		\begin{enumerate}
			\item We have 
			\begin{flalign}
				\label{currentestimate3} 
				\Bigg| \displaystyle\sum_{i = N_1'}^{N_2'} \mathfrak{k}_i^{[m]} (t) - \displaystyle\sum_{i : \varphi_t (i) \in \llbracket N_1', N_2' \rrbracket} \lambda_i^m \Bigg| \le 12 (\log N)^{m+3}.
			\end{flalign}
			\item For any $k \in \llbracket N_1', N_2' \rrbracket$, we have 
			\begin{flalign}
				\label{currentestimate4}
				\Bigg| \displaystyle\sum_{i < k} \mathfrak{k}_i^{[m]} (t) - \displaystyle\sum_{i : \varphi_t (i) \le k} \lambda_i^m \Bigg| \le 12 (\log N)^{m+3}.
			\end{flalign}
		\end{enumerate} 
	\end{prop}
	
	\begin{proof} 
		
		The proofs of \eqref{currentestimate3} and \eqref{currentestimate4} are very similar, so we only focus on the former. We may assume in what follows that $(N_1', N_2')$ is fixed, by a union bound. Let us first establish \eqref{currentestimate3} for any deterministic $t \in [0, T]$, and then we will show it holds for all $t \in [0, T]$. 
		
		So, fix $t \in [0, T]$. We first apply \Cref{ltl0} to deduce the existence of a random matrix $\bm{P} = [P_{ij}] \in \SymMat_{\llbracket N_1, N_2 \rrbracket}$ with the same law as $\bm{L}(0)$, and an overwhelmingly probable event $\mathsf{E}_1$, on which we have 
		\begin{flalign}
			\label{ltp} 
			\displaystyle\max_{i,j \in \llbracket N_1' - m, N_2' + m \rrbracket} | P_{ij} - L_{ij} (t) | \le e^{-(\log N)^3}.
		\end{flalign}
		
		\noindent where we used the fact that $\llbracket N_1' - m, N_2' + m \rrbracket \subseteq \llbracket N_1 + T(\log N)^4, N_2 - T (\log N)^4 \rrbracket$. Further set $N' = N_2' - N_1' + 1$ and define the $N' \times N'$ symmetric matrix $\bm{Q} = [Q_{ij}] \in \SymMat_{\llbracket N_1, N_2' \rrbracket}$ by setting $Q_{ij} = P_{ij}$ whenever $i, j \in \llbracket N_1', N_2' \rrbracket$. Set $\eig \bm{Q} = (\mu_1, \mu_2, \ldots , \mu_{N'})$, and let $\psi : \llbracket 1, N' \rrbracket \rightarrow \llbracket N_1', N_2' \rrbracket$ denote a $\zeta$-localization center bijection for $\bm{Q}$. Recalling \Cref{adelta}, we define the event $\mathsf{E}_2 = \mathsf{BND}_{\bm{L}(t)} (\log N) \cap \mathsf{BND}_{\bm{P}} (\log N) \cap \mathsf{BND}_{\bm{Q}} (\log N) \cap \mathsf{SEP}_{\bm{L}(0)} (e^{-(\log N)^2})$, which by \Cref{l0eigenvalues} and \Cref{eigenvalues0} is overwhelmingly probable. We restrict to $\mathsf{E}_1 \cap \mathsf{E}_2$ in the below. 
		
		We will proceed by comparing the first sum on the left side of \eqref{currentestimate3} to $\Tr \bm{Q} = \sum_{\mu \in \eig \bm{Q}} \mu^m$; use \Cref{lleigenvalues} to approximate the eigenvalues of $\bm{Q}$ by those of $\bm{L}(t)$; and compare the resulting expression to the second sum on the left side of \eqref{currentestimate3}. To implement the first task, observe that 
		\begin{flalign}
			\label{sumkappalt} 
			\begin{aligned}
				\Bigg| \displaystyle\sum_{i = N_1'}^{N_2'} \mathfrak{k}_i^{[m]} (t) - \Tr \bm{Q}^m \Bigg| & = \Bigg| \displaystyle\sum_{i=N_1'}^{N_2'} [\bm{L}(t)^m]_{ii} - \Tr \bm{Q}^m \Bigg| \\
				& \le \Bigg| \displaystyle\sum_{i=N_1'}^{N_2'} [\bm{P}^m]_{ii} - \Tr \bm{Q}^m \Bigg| + N e^{-(\log N)^3} m (\log N)^m,
			\end{aligned} 
		\end{flalign}
		
		\noindent where we have denoted the $(i, j)$ entry of any matrix $\bm{M}$ by $[\bm{M}]_{ij}$; here, the first statement holds by \Cref{current} and the second by \eqref{ltp} and our restriction to $\mathsf{E}_2$. Next, since $\bm{P}$ and $\bm{Q}$ are tridiagonal and satisfy $P_{ij} = Q_{ij}$ whenever $i, j \in \llbracket N_1', N_2' \rrbracket$, observe that $[\bm{P}^m]_{ii} = [\bm{Q}^m]_{ii}$ for each $i \in \llbracket N_1' +m, N_2' - m \rrbracket$. We further have by our restriction to $\mathsf{E}_2$ that each entry of $\bm{P}$ and $\bm{Q}$ is bounded by $\log N$, meaning (again since $\bm{P}$ and $\bm{Q}$ are tridiagonal) that each entry of $\bm{P}^m$ and $\bm{Q}^m$ is bounded by $(3 \log N)^m$. It follows that 
		\begin{flalign}
			\label{sumpq} 
			\Bigg| \displaystyle\sum_{i=N_1'}^{N_2'} [\bm{P}^m]_{ii} - \Tr \bm{Q}^m \Bigg| \le \displaystyle\sum_{\substack{i \in \llbracket N_1', N_2' \rrbracket \\ i \notin \llbracket N_1'+m, N_2'-m \rrbracket}} \big( |(\bm{P}^m)_{ii}| + |(\bm{Q}^m)_{ii}| \big) \le 4m (3 \log N)^m.
		\end{flalign} 
		
		\noindent Next, denote $N_1'' = N_1' + (\log N)^3$ and $N_2'' = N_2' - (\log N)^3$. Then,
		\begin{flalign}
			\label{sumqmu} 
			\Bigg| \Tr \bm{Q}^m - \displaystyle\sum_{j: \psi (j) \in \llbracket N_1'', N_2'' \rrbracket} \mu_{j}^m \Bigg| \le 2 (\log N)^3 \cdot \displaystyle\max_{\mu \in \eig \bm{Q}} |\mu|^m \le 2 (\log N)^{m+3},
		\end{flalign}
		
		\noindent where in the first bound we used that $\Tr \bm{Q}^m = \sum_{\mu \in \eig \bm{Q}} \mu^m$ and $|\llbracket N_1', N_2' \rrbracket \setminus \llbracket N_1'', N_2'' \rrbracket| \le 2(\log N)^3$, and in the second we used the fact that $|\mu| \le \log N$ for any $\mu \in \eig \bm{Q}$ (by our restriction to $\mathsf{E}_2$). 
		
		Now, we apply \Cref{centerdistance2} and \Cref{lleigenvalues2}, the latter with the $(\bm{L}; \tilde{\bm{L}}; \delta; \mathcal{D})$ there equal to $(\bm{Q}; \bm{L}(t); e^{-(\log N)^3}; \llbracket N_1, N_2 \rrbracket \setminus \llbracket N_1', N_2' \rrbracket)$ here. This by \eqref{ltp} yields a constant $c > 0$, such that the following holds with overwhelming probability. There exists a function $\kappa : \llbracket 1, N' \rrbracket \rightarrow \llbracket 1, N \rrbracket$ such that, for each $j \in \llbracket 1, N' \rrbracket$ with $\psi(j) \in \llbracket N_1'', N_2'' \rrbracket$, 
		\begin{flalign}
			\label{mulambdakappa} 
			|\mu_j - \lambda_{\kappa(j)}| \le e^{-c_1 (\log N)^3}, \quad \text{and} \quad \big| \psi (j) - \varphi_t (\kappa(j)) \big| \le (\log N)^3.
		\end{flalign} 
		
		\noindent The second statement in \eqref{mulambdakappa} holds since \Cref{lleigenvalues2} implies that $\psi(j)$ is an $N^{-1} \zeta$-localization center of $\lambda_{\kappa(j)}$ with respect to $\bm{L}(t)$; therefore, since $\varphi_t (\kappa(j))$ is as well,  \Cref{centerdistance2} implies (as $\psi(j) \in \llbracket N_1'', N_2'' \rrbracket$) that $|\psi(j) - \varphi_t (\kappa(j))| \le (\log N)^3$. Further observe that, due our restriction to $\mathsf{E}_2$, we have $|\lambda_i - \lambda_{i'}| \ge e^{-(\log N)^2} \ge 2e^{-c_1 (\log N)^3}$ for any distinct $i, i' \in \llbracket 1, N \rrbracket$; thus, $\kappa(j) \ne \kappa (j')$ for any distinct $j, j' \in \llbracket 1, N' \rrbracket$ with $\psi (j), \psi(j') \in \llbracket N_1'', N_2'' \rrbracket$. Therefore,   
		\begin{flalign}
			\label{summulambda2} 
			\begin{aligned} 
				\Bigg| & \displaystyle\sum_{j: \psi (j) \in \llbracket N_1'', N_2'' \rrbracket} \mu_{j}^m - \displaystyle\sum_{j : \varphi_t (j) \in \llbracket N_1', N_2' \rrbracket} \lambda_j^m \Bigg| \\ 
				&  \le \displaystyle\sum_{j: \psi (j) \in \llbracket N_1'', N_2'' \rrbracket} |\mu_{j}^m - \lambda_{\kappa(j)}^m| + 2 (\log N)^3 \cdot \displaystyle\max_{\lambda \in \eig \bm{L}(t)} |\lambda|^m \le 4 (\log N)^{m+3},
			\end{aligned} 
		\end{flalign} 
		
		\noindent where the first statement holds by changing variables from $j$ to $\kappa(j)$ in the second sum there whenever $\psi (j) \in \llbracket N_1'', N_2'' \rrbracket$, and using \eqref{mulambdakappa} with the facts that $\llbracket N_1'' - (\log N)^3, N_2 + (\log N)^3 \rrbracket = \llbracket N_1', N_2' \rrbracket$ and that $|\llbracket N_1', N_2' \rrbracket \setminus \llbracket N_1'', N_2'' \rrbracket| \le 2 (\log N)^3$; the second follows from \eqref{mulambdakappa} with the fact that $|\lambda| \le \log N$ for each $\lambda \in \eig \bm{L}(t)$ (as we restricted to $\mathsf{E}_2$). Summing \eqref{sumkappalt}, \eqref{sumpq}, \eqref{sumqmu}, and \eqref{summulambda2}, we deduce that 
		\begin{flalign}
			\label{sum1} 
			\Bigg|\displaystyle\sum_{i=N_1'}^{N_2'} \mathfrak{k}_i^{[m]} (t) - \displaystyle\sum_{i:\varphi_t(i) \in \llbracket N_1', N_2' \rrbracket} \lambda_i^m \Bigg| \le 7 (\log N)^{m+3}.
		\end{flalign} 
		
		This verifies that \eqref{currentestimate3} holds with overwhelming probability for a fixed $t \in [0, T]$; it remains to show it holds with high probability for all $t \in [0, T]$ simultaneously. To that end, for any $t \in [0, T]$, let $\mathsf{E}(t)$ denote the event on which \eqref{sum1} holds. We then have that $\mathbb{P}[\mathsf{E}(t)^{\complement}] \le c_2^{-1} e^{-c_2 (\log N)^2}$ for some constant $c_2 > 0$. Denoting $\mathfrak{c} = c_2/2$, it follows from \Cref{centerdistance2} that there exists a constant $c_3>0$ and an event $\mathsf{E}_4$ with $\mathbb{P}[\mathsf{E}_4^{\complement}] \le c_3^{-1} e^{-c_3 (\log N)^2}$ such that the following holds on $\mathsf{E}_4$. For any $s, s' \in [0, T]$ with $|s-s'| \le e^{-\mathfrak{c}(\log N)^2}$, we have $|\varphi_s (i) - \varphi_{s'} (i)| \le (\log N)^3$, whenenever $\varphi = \varphi_s (i)$ satisfies \eqref{tzeta}. Letting $\mathcal{T}$ denote an $e^{-\mathfrak{c}(\log N)^2}$-mesh of $[0,T]$, we restrict to the event $\bigcap_{s \in \mathcal{T}} \mathsf{E}(s) \cap \mathsf{E}_4 \cap \bigcap_{r \ge 0} \mathsf{BND}_{\bm{L}(r)} (\log N)$, which we may by a union bound (and \Cref{l0eigenvalues}).
		
		Now let $t \in [0, T]$ be arbitrary, and let $s \in \mathcal{T}$ satisfy $|t-s| \le e^{-\mathfrak{c} (\log N)^2}$. Then, we claim that 
		\begin{flalign}
			\label{sum3}
			\begin{aligned}
				& \Bigg| \displaystyle\sum_{i=N_1'}^{N_2'} \big| \mathfrak{k}_i^{[m]} (t) - \mathfrak{k}_i^{[m]} (s) \big| \Bigg| \le 1; \\
				& \Bigg| \displaystyle\sum_{i: \varphi_t (i) \in \llbracket N_1', N_2' \rrbracket} \lambda_i^m - \displaystyle\sum_{i : \varphi_s (i) \in \llbracket N_1', N_2' \rrbracket} \lambda_i^m \Bigg| \le 4 (\log N)^{m+3}.  
			\end{aligned} 
		\end{flalign}
		
		\noindent which together with \eqref{sum1} (with the $t$ there equal to $s$ here) would imply the proposition. 
		
		To verify the first bound in \eqref{sum3}, observe for any $i, j \in \llbracket N_1, N_2 \rrbracket$ that 
		\begin{flalign*} 
			|L_{ij} (t) - L_{ij} (s)| \le |s-t| \cdot \displaystyle\sup_{r \in [s,t]} |L_{ij}' (r)| \le e^{-\mathfrak{c}(\log N)^2} \cdot 2 (\log N)^2 \le e^{-\mathfrak{c} (\log N)^2 / 2},
		\end{flalign*} 
		
		\noindent where in the second inequality we used \Cref{matrixl}, \eqref{abr}, \eqref{derivativepa}, and our restriction to the event $\bigcap_{r \ge 0} \mathsf{BND}_{\bm{L}(r)} (\log N)$. Since the same event implies that each entry of $\bm{L}(s)$ and $\bm{L}(t)$ is bounded by $\log N$, and both of these matrices are tridiagonal, it follows for any $i \in \llbracket N_1, N_2 \rrbracket$ that 
		\begin{flalign*}
			\big| \mathfrak{k}_i^{[m]} (t) - \mathfrak{k}_i^{[m]} (s) \big| & = \big| [\bm{L}(t)^m]_{ii} - [\bm{L}(s)^m]_{ii} \big|  \\
			& \le \displaystyle\max_{i,j \in \llbracket N_1, N_2 \rrbracket} |L_{ij} (s) - L_{ij} (t)| \cdot m (3 \log N)^m \\
			& \le e^{-\mathfrak{c}(\log N)^2/2} \cdot m (3 \log N)^m \le N^{-1}.
		\end{flalign*}
		
		\noindent Summing over $i \in \llbracket N_1', N_2' \rrbracket$ confirms the first bound in \eqref{sum3}. The second bound in \eqref{sum3} follows from the fact that $|\lambda| \le \log N$ for any $\lambda \in \eig \bm{L}(t)$ (by our restriction to $\mathsf{BND}_{\bm{L}(t)} (\log N)$), with the fact that there are at most $4 (\log N)^3$ indices $i$ for which $\varphi_s (i) \in \llbracket N_1', N_2' \rrbracket$ and $\varphi_t (i) \notin \llbracket N_1', N_2' \rrbracket$ or for which $\varphi_s (i) \notin \llbracket N_1', N_2' \rrbracket$ and $\varphi_t (i) \in \llbracket N_1', N_2' \rrbracket$ (as $|\varphi_t (i) - \varphi_s (i)| \le (\log N)^3$, by our restriction to $\mathsf{E}_4$). This proves \eqref{sum3} and thus \eqref{currentestimate3}.
		
		As mentioned previously, the proof of \eqref{currentestimate4} is entirely analogous, and is obtained by replacing the application of \Cref{ltl0} and \Cref{lleigenvalues} above with \Cref{ltmatrix} (with the $\ell$ there equal to $k$ here); we omit further details.
	\end{proof} 
	
	We can now prove \Cref{currentestimate}.

	\begin{proof}[Proof of \Cref{currentestimate}]
		
		We assume that $\alpha > 0$ in what follows, as the proof when $\alpha < 0$ is entirely analogous. The proofs of \eqref{currentsum} and \eqref{currentsum2} are very similar, so we only focus on the former. Recalling \Cref{adelta}, define the event $\mathsf{E}_1 = \bigcap_{s \in [0, t]} \mathsf{BND}_{\bm{L}(s)} (\log N)$. By \Cref{l0eigenvalues}, $\mathsf{E}_1$ holds with overwhelming probability, so we restrict to $\mathsf{E}_1$ in what follows. We further restrict to the event $\mathsf{E}_2$ on which \Cref{qijsalpha} holds, with the $T$ there both equal to $0$ and $T$ here (we may take the $T$ there to be $0$ by the $R = (\log N)^2 |i-j|^{1/2}$ case of \Cref{qij}). 
		
		Fix $t \in [0, T]$, and let $N_1' = \min \{ i : q_i (t) \in \mathcal{J} \}$ and $N_2' = \max \{ i : q_i (t) \in \mathcal{J} \}$. We claim that $(N_1', N_2')$ satisfy \eqref{n1n20}. To verify this, assume to the contrary, so that either $N_1' < N_1 + T (\log N)^5 + N^{1/100}$ or $N_2' > N_2 - T(\log N)^5 - N^{1/100}$ holds. The analysis of these two cases is entirely analogous, so we assume the former. Then, 
		\begin{flalign*} 
			q_{N_1'} (0) - q_0 (0) \ge q_{N_1'} (t) - T \log N & \ge \alpha N_1 + (T + |N_1|^{1/2} ) (\log N)^5 - T \log N \\
			&  \ge \alpha N_1 + T (\log N)^5 + \displaystyle\frac{1}{2} \cdot |N_1|^{1/2}  (\log N)^5,
		\end{flalign*} 
		
		\noindent where the first statement holds since $q_0 (0) = 0$ and $|q_i'(t)| = |p_i(t)| \le \log N$ (by \eqref{qtpt}, \eqref{abr}, \Cref{matrixl}, and our restriction to $\mathsf{E}_1$); the second holds by \eqref{j} and the inclusion $q_{N_1'} (t) \in \mathcal{J}$; and the third holds for $N$ sufficiently large, by \eqref{n1n2zetat}. This contradicts \eqref{qiqjs4} (at $s=T=0$), confirming \eqref{n1n20}. 
		
		Therefore, \Cref{currentestimate2} applies and implies that  \eqref{currentestimate3} holds with overwhelming probability. Then, since $q_i (t) \in \mathcal{J}$ implies that $i \in \llbracket N_1', N_2' \rrbracket$, it suffices (by \eqref{qjs2}) to show that 
		\begin{flalign*}
			\displaystyle\sum_{i \in \llbracket N_1', N_2' \rrbracket : q_i (t) \notin \mathcal{J}} |\mathfrak{k}_i^{[m]} (t)|  + \displaystyle\sum_{i : \varphi_t (i) \in \llbracket N_1', N_2' \rrbracket, q_{\varphi_t(i)} (t) \notin \mathcal{J}} |\lambda_i|^m \le (3 \log N)^{m+5}.
		\end{flalign*} 
		
		\noindent Changing variables in the second sum from $i$ to $\varphi_t^{-1} (i)$, this is equivalent to 
		\begin{flalign*}
			\displaystyle\sum_{i \in \llbracket N_1', N_2' \rrbracket : q_i (t) \notin \mathcal{J}} |\mathfrak{k}_i^{[m]} (t)|  + \displaystyle\sum_{i \in \llbracket N_1', N_2' \rrbracket: q_{i} (t) \notin \mathcal{J}} |\lambda_{\varphi_t^{-1} (i)}|^m \le (3 \log N)^{m+5}.
		\end{flalign*} 
		
		\noindent Using the facts that each entry of $\bm{L}(t)^m$ is bounded by $(3 \log N)^m$ (as $\bm{L}(t)$ is tridiagonal and each of its entries is bounded by $\log N$) and that $|\lambda| \le \log N$ for each $\lambda \in \eig \bm{L}(t)$, it remains to verify 
		\begin{flalign}
			\label{qin1n2} 
			& \# \{ i \in \llbracket N_1', N_2' \rrbracket : q_i (t) \notin \mathcal{J} \} \le 2 (\log N)^{9/2}. 
		\end{flalign}
		
		To show this, observe that if $i \in \llbracket N_1', N_2' \rrbracket$ satisfies $q_i (t) \notin \mathcal{J}$, then either $q_i (t) < q_{N_1'} (t)$ or $q_i (t) > q_{N_2'} (t)$. In the former case, we have 
		\begin{flalign*} 
			0 > q_i (t) - q_{N_1'} (t) \ge \alpha \cdot (i -N_1') - |i-N_1'|^{1/2} \cdot (\log N)^2,
		\end{flalign*}
		
		\noindent where the last inequality holds by our restriction to $\mathsf{E}_2$. Hence $i < N_1' + (\log N)^{9/2}$, and there are at most $(\log N)^{9/2}$ such indices in $\llbracket N_1', N_2' \rrbracket$. By similar reasoning, there are at most $(\log N)^{9/2}$ indices $i \in \llbracket N_1', N_2' \rrbracket$ such that $q_i (t) > q_{N_2'} (t)$. Summing these bounds yields \eqref{qin1n2} and thus \eqref{currentsum}.
		
		As mentioned previously, the proof of \eqref{currentsum2} is entirely analogous (using \eqref{currentestimate4} in place of \eqref{currentestimate3}).
	\end{proof}

	\section{Proof of the Asymptotic Scattering Relation} 
	
	\label{InverseAsymptotic}

	\subsection{Proof of \Cref{ztlambda}} 
	
	\label{YEstimate}
	
	In this section we prove \Cref{ztlambda}, which will follow from the following variant of it that replaces the $Q_t (i)$-dependent sums in \eqref{lambdak22} with ones that depend on $\varphi_t (i)$. The latter will be shown in \Cref{ProofQ2} below.  We adopt \Cref{lbetaeta} throughout.

	\begin{thm}
		
		\label{ztlambda2} 
		
		The following holds with overwhelming probability. Let $k \in \llbracket 1, N \rrbracket$ satisfy \eqref{n1n2k02}. Then, for any $t \in [0, T]$, we have
		\begin{flalign}
			\label{lambdak2}
			\Bigg| \lambda_k t - Q_k(t) + Q_k (0) - 2 \displaystyle\sum_{i: \varphi_t (i) < \varphi_t (k)}  \log |\lambda_{k} - \lambda_{i}| + 2 \displaystyle\sum_{i: \varphi_0 (i) < \varphi_0 (k)} \log |\lambda_k - \lambda_i| \Bigg| \le 2 (\log N)^{12}.
		\end{flalign} 
		
	\end{thm}

	\begin{proof}[Proof of \Cref{ztlambda}]
		
		Throughout this proof, we assume that $\alpha > 0$, as the proof when $\alpha < 0$ is entirely analogous; we also fix $k \in \llbracket 1, N \rrbracket$ satisfying \eqref{n1n2k02}. We begin by restricting to several events. Recalling \Cref{adelta}, first restrict to the event $\mathsf{E}_1 = \mathsf{BND}_{\bm{L}(0)} (\log N) \cap \mathsf{SEP}_{\bm{L}(0)} (e^{-(\log N)^2})$, which we may by \Cref{l0eigenvalues}, \Cref{eigenvalues0}, and a union bound. Further restrict to the event $\mathsf{E}_2$ on which \Cref{ztlambda2} holds; to the event $\mathsf{E}_3$ on which \Cref{qijsalpha} holds; and to the event $\mathsf{E}_4$ on which \Cref{centert} holds. 
		
		In view of \eqref{lambdak2} (and our restriction to $\mathsf{E}_2$), we must show that 
		\begin{flalign*}
			2 \displaystyle\sum_{i = N_1}^{N_2} \big| (\mathbbm{1}_{\varphi_t (i) < \varphi_t (k)} - \mathbbm{1}_{\varphi_0 (i) < \varphi_0 (k)}) - (& \mathbbm{1}_{Q_i (t) < Q_k (t)} - \mathbbm{1}_{Q_i (0) < Q_k (0)} ) \big| \cdot \big| \log |\lambda_k - \lambda_i| \big| \\
			& \le (\log N)^{15} - 2 (\log N)^{12}.
		\end{flalign*}
		
		\noindent Since by our restriction to $\mathsf{E}_1$ we that $|\log |\lambda_k - \lambda_i|| \le (\log N)^2$ for each $i \ne k$, it therefore suffices to show that 
		\begin{flalign}
			\label{sum0}
			\displaystyle\sum_{i = N_1}^{N_2} \big| (\mathbbm{1}_{\varphi_t (i) < \varphi_t (k)} - \mathbbm{1}_{\varphi_0 (i) < \varphi_0 (k)}) - (\mathbbm{1}_{Q_i (t) < Q_k (t)} - \mathbbm{1}_{Q_i (0) < Q_k (0)}) \big|  \le (\log N)^8.
		\end{flalign}
		
		We first claim that the summand on the left side of \eqref{sum0} is equal to $0$ if $i$ satisfies $|\varphi_0 (i) - \varphi_0 (k)| > 2 T (\log N)^5$. Indeed, fix such an index $i$, and assume $\varphi_0 (i) < \varphi_0 (k) - 2T (\log N)^5$, as the verification in the alternative case is entirely analogous. Then, \eqref{ujms} (with our restriction to $\mathsf{E}_4$) gives $|\varphi_t (i) - \varphi_0 (i)| \le T (\log N)^2$ and $|\varphi_t (k) - \varphi_0 (k)| \le T (\log N)^2$, so for sufficiently large $N$ we have $\varphi_t (i) < \varphi_t (k) - T (\log N)^5$. As such, $\mathbbm{1}_{\varphi_t (i) < \varphi_t (k)} = 1 = \mathbbm{1}_{\varphi_0 (i) < \varphi_0 (k)}$. By \eqref{qiqjs5} (and our restriction to $\mathsf{E}_3$) whose assumption holds by the above bounds, we have for each $s \in \{ 0, t \}$ that
		\begin{flalign*}
			Q_k (s) - Q_i (s) = q_{\varphi_s (k)} - q_{\varphi_s (i)} \ge \displaystyle\frac{\alpha}{2} \cdot \big( \varphi_s (k) - \varphi_s (i) \big) > 0.
		\end{flalign*}
		
		\noindent So, it follows that $\mathbbm{1}_{Q_i (t) < Q_k (t)} = 1 = \mathbbm{1}_{Q_i (0) < Q_k (0)}$, meaning that the summand on the left side of \eqref{sum0} associated with $i$ is equal to $0$. 
		
		Therefore, we may restrict the sum on the left side of \eqref{sum0} to indices $i$ satisfying $|\varphi_0 (i) - \varphi_0 (k)| \le 2T (\log N)^5$. Fix such an index $i$; we claim for each $s \in \{ 0, t \}$ that
		\begin{flalign}
			\label{ik2} 
			\mathbbm{1}_{\varphi_s (i) < \varphi_s (k)} = \mathbbm{1}_{Q_i (s) < Q_k (s)}, \qquad \text{unless $|\varphi_s (i) - \varphi_s (k)| \le (\log N)^5$},
		\end{flalign} 
		
		\noindent which would imply \eqref{sum0} and thus the theorem. To confirm \eqref{ik2}, fix $s \in \{ 0, t \}$, and assume that $|\varphi_s (i) - \varphi_s (k)| > (\log N)^5$; we will further suppose that $\varphi_s (i) < \varphi_s (k) - (\log N)^5$, as the proof in the alternative case is entirely analogous. By our assumptions \eqref{n1n2k02} and $|\varphi_0 (i) - \varphi_0 (k)| \le 2T (\log N)^5$, observe (by \eqref{ujms}, using our restriction to $\mathsf{E}_4$) that $\varphi_s (i), \varphi_s (k) \in \llbracket N_1 + T(\log N)^5, N_2 - T(\log N)^5 \rrbracket$, so \eqref{qiqjs4} holds (by our restriction to $\mathsf{E}_4$) with the $(i, j)$ there equal to $(\varphi_s (i), \varphi_s (k))$ here. Hence, 
		\begin{flalign*} 
			Q_k (s) - Q_i (s) & = q_{\varphi_k (s)} (s) - q_{\varphi_i (s)} (s) \\
			& \ge \alpha \cdot \big( \varphi_k (s) - \varphi_i (s) \big) - |\varphi_k (s) - \varphi_i (s)|^{1/2} \cdot (\log N)^2 > 0,
		\end{flalign*} 
		
		\noindent where the last bound holds since $\varphi_k (s) - \varphi_i (s) \ge (\log N)^5$. Therefore, $\mathbbm{1}_{Q_i (s) < Q_k (s)} = 1 = \mathbbm{1}_{\varphi_s (i) < \varphi_s (k)}$, proving \eqref{ik2} and thus the theorem.
	\end{proof}
	
	\subsection{Proof of \Cref{ztlambda2} if $T^2 \le N$}
	
	\label{ProofQ}
	
	In this section we establish the following case of \Cref{ztlambda2}, which assumes that $T^2 \le N$ (and a slightly stronger condition \eqref{n1n2k022} on $\varphi_0 (k)$ than imposed in \eqref{n1n2k02}) and addresses a fixed time $t \in [0, T]$. Recall we adopt \Cref{lbetaeta} throughout.

	\begin{prop}
		
		\label{ztlambda3}
		
		For any $t \in [0, T]$, the following holds with overwhelming probability. Let $k \in \llbracket 1, N \rrbracket$ satisfy 
		\begin{flalign}
			\label{n1n2k022} 
			N_1 + T (\log N)^6 + N^{1/100} \le \varphi_0 (k) \le N_2 - T (\log N)^6 - N^{1/100}.
		\end{flalign}
		
		\noindent If $T^2 \le N$, then we have
		\begin{flalign}
			\label{lambdak3}
			\Bigg| \lambda_k t - Q_k(t) + Q_k (0) - 2 \displaystyle\sum_{i: \varphi_t (i) < \varphi_t (k)}  \log |\lambda_{k} - \lambda_{i}| + 2 \displaystyle\sum_{i: \varphi_0 (i) < \varphi_0 (k)} \log |\lambda_k - \lambda_i| \Bigg| \le 2 (\log N)^{11}.
		\end{flalign} 
		
	\end{prop}
	
	Proposition \ref{ztlambda3} is a quick consequence of the below two lemmas. In what follows, we denote
	\begin{flalign}
		\label{ct} 
		\mathfrak{C}(s) = \log \displaystyle\sum_{j=1}^N e^{-\lambda_j s} u_j (N_1; 0)^2.
	\end{flalign}

	\begin{lem} 
		
		\label{ytlambda}
		
		For any $t \in [0, T]$, the following holds with overwhelming probability. For any $k \in \llbracket 1, N \rrbracket$ satisfying \eqref{n1n2k022}, we have 
		\begin{flalign}
			\label{qkt2} 
			\begin{aligned} 
				\Bigg| Q_k(t) -  Q_k (0) + 2 \displaystyle\sum_{i: \varphi_t (i) < \varphi_t (k)} \log | & \lambda_k - \lambda_i|  - 2 \displaystyle\sum_{i: \varphi_0 (i) < \varphi_0 (k)} \log |\lambda_k - \lambda_i| - \lambda_k t \\
				& - \mathfrak{C}(t) - q_{N_1} (t) + q_{N_1} (0) \Bigg| \le 2 (\log N)^6.
			\end{aligned} 
		\end{flalign}
		
	\end{lem}

	\begin{lem} 
		
		\label{ctqt} 
		
		Fix $t \in [0, T]$, and assume that $T^2 \le N$. There exists a constant $c>0$ such that 
		\begin{flalign*}
			\mathbb{P} \big[ | q_{N_1} (0) - q_{N_1} (t) - \mathfrak{C}(t) | \le (\log N)^{11} \big] \ge 1 - c^{-1} e^{-c(\log N)^2}.
		\end{flalign*}
		
	\end{lem}
	
	\begin{proof}[Proof of \Cref{ztlambda3}] 
		This follows from \Cref{ytlambda} and \Cref{ctqt}. 
	\end{proof} 
	
	We now show \Cref{ytlambda} and \Cref{ctqt}.

	\begin{proof}[Proof of \Cref{ytlambda}]

		Fix $k \in \llbracket 1, N \rrbracket$ satisfying \eqref{n1n2k022}, and observe by \Cref{uknt} that 
		\begin{flalign}
			\label{ukt0} 
			-2 \log | u_k (N_1; t) | = \lambda_k t - 2 \log | u_k (N_1; 0) | + \mathfrak{C} (t).
		\end{flalign}
		
		\noindent  By \Cref{centert}, there is an overwhelmingly probable event $\mathsf{E}_1$, on which we have $| \varphi_t (k) - \varphi_0 (k) | \le T (\log N)^2$. Observe on $\mathsf{E}_1$ that $N_1 + T (\log N)^5 + N^{1/100} \le \varphi_s (k) \le N_2 - T (\log N)^5 - N^{1/100}$ for each $s \in \{ 0, t \}$, by \eqref{n1n2k022}. We restrict to $\mathsf{E}_1$ in what follows.
		
		By \Cref{uk12}, there exists an overwhelmingly probable event $\mathsf{E}_2$, on which the following holds. For any real number $s \in \{ 0, t \}$, we have 
		\begin{flalign*}
			\Bigg| 2 \log | u_k (N_1; s) | - 2 \displaystyle\sum_{i=N_1}^{\varphi_s (k)-1} \log L_{i,i+1} (s) + 2 \displaystyle\sum_{i: \varphi_s (i) < \varphi_s (k)} \log |\lambda_{i} - \lambda_{k}| \Bigg| \le 2 (\log N)^6,
		\end{flalign*}
		
		\noindent where we used that $N_1 + T(\log N)^5 + N^{1/100} \le \varphi_s (k) \le N_2 - T(\log N)^5 - N^{1/100}$. We further restrict to $\mathsf{E}_2$ below. Since \Cref{matrixl}, \eqref{abr}, and \Cref{lbetaeta} together imply that 
		\begin{flalign*} 
			2 \displaystyle\sum_{i=N_1}^{\varphi_s (k)-1} \log L_{i,i+1} (s) = 2 \displaystyle\sum_{i=N_1}^{\varphi_s (k)-1} \log a_i (s) = q_{N_1} (s) - q_{\varphi_s (k)} (s) = q_{N_1} (s) -  Q_k (s),
		\end{flalign*} 
		
		\noindent we obtain by our restriction to $\mathsf{E}_1 \cap \mathsf{E}_2$ that, for each $s \in \{ 0 ,t \}$,   
		\begin{flalign*}
			\Bigg| 2 \log | u_{k} (N_1; s) | + Q_k (s) - q_{N_1} (s) + 2 \displaystyle\sum_{i: \varphi_s (i) < \varphi_s (k)} \log |\lambda_{i} - \lambda_{k}| \Bigg| \le (\log N)^6.
		\end{flalign*}
		
		\noindent Subtracting this bound at $s = 0$ from it at $s = t$, and applying \eqref{ukt0}, then yields the lemma.
	\end{proof}

	\begin{proof}[Proof of \Cref{ctqt}]
		
		We will establish this proposition by averaging an estimate on the left side of \eqref{qkt2} over all $k \in \llbracket N_1, N_2 \rrbracket$. To that end, let $\mathsf{E}_1$ denote the event on which \eqref{qkt2} holds for each $k \in \llbracket N_1, N_2 \rrbracket$ satisfying \eqref{n1n2k022}. By \Cref{ytlambda}, $\mathsf{E}_1$ is overwhelmingly probable. We will define additional events on which we can bound the left side of \eqref{qkt2} for $k \in \llbracket N_1, N_2 \rrbracket$ not necessarily satisfying \eqref{n1n2k022}. Specifically, recalling \Cref{adelta}, set 
		\begin{flalign*} 
			& \mathsf{E}_2 = \bigcap_{r \ge 0} \mathsf{BND}_{\bm{L}(r)} (\log N) \cap \mathsf{SEP}_{\bm{L}(0)} (e^{-(\log N)^2}); \\ 
			& \mathsf{E}_3 = \bigcap_{k=N_1}^{N_2} \big\{ | \varphi_t (k) - \varphi_0 (k) | \le T (\log N)^2 \big\}. 
		\end{flalign*} 
		
		\noindent Then $\mathsf{E}_2$ is overwhelmingly probable, by \Cref{l0eigenvalues} and \Cref{eigenvalues0}, and $\mathsf{E}_3$ is overwhelmingly probable, by \Cref{centert}. We further define the event
		\begin{flalign*}
			\mathsf{E}_4 = \bigcap_{|i-j| \le T (\log N)^2} \{ | q_i (0) - q_j (0) | \le T(\log N)^3 \},
		\end{flalign*}
		
		\noindent which by \Cref{qij} and a union bound is overwhelmingly probable. Defining the event $\mathsf{E} = \mathsf{E}_1 \cap \mathsf{E}_2 \cap \mathsf{E}_3 \cap \mathsf{E}_4$, we may by a union bound restrict to $\mathsf{E}$ in what follows.
		
		Then, for any $k \in \llbracket N_1, N_2 \rrbracket$, we have 
		\begin{flalign}
			\label{qkt3}
			\begin{aligned} 
				\Bigg| & Q_k (t) - Q_k (0) - q_{N_1} (t) + q_{N_1} (0) - \lambda_k t - \mathfrak{C}(t)  \\
				&  + 2 \displaystyle\sum_{i:\varphi_t (i) < \varphi_t (k)} \log |\lambda_i - \lambda_k| - 2 \displaystyle\sum_{i:\varphi_0 (i) < \varphi_0 (k)} \log |\lambda_i - \lambda_k| \Bigg| \\ 
				& \quad \le | q_{\varphi_t (k)} (0) - q_{\varphi_0 (k)} (0) | + | q_{\varphi_t (k)} (t) - q_{\varphi_t (k)} (0) | + | q_{N_1} (t) - q_{N_1} (0) | \\
				& \qquad + T \log N + |\mathfrak{C}(t)|  + 2 (\log N)^2 \cdot \# \{ i \in \llbracket N_1, N_2 \rrbracket : \varphi_t (i) < \varphi_t (k), \varphi_0 (i) > \varphi_0 (k) \} \\
				& \qquad  + 2 (\log N)^2 \cdot \# \{ i \in \llbracket N_1, N_2 \rrbracket : \varphi_t (i) > \varphi_t (k), \varphi_0 (i) < \varphi_0 (k) \},
			\end{aligned} 
		\end{flalign}
		
		\noindent where we used the definition $Q_k (s) = q_{\varphi_s (k)} (s)$ (from \eqref{qjs2}) and the facts that $|\lambda_k| \le \log N$ and that $| \log |\lambda_i - \lambda_k || \le (\log N)^2$ (both by our restriction to $\mathsf{E}_2$). We next bound the terms on the right side of \eqref{qkt3}. Observe by our restriction to $\mathsf{E}_4$ that 
		\begin{flalign}
			\label{qt4} 
			| q_{\varphi_t (k)} (0) - q_{\varphi_0 (k)} (0) | \le T (\log N)^3,
		\end{flalign} 
		
		\noindent since $| \varphi_t (k) - \varphi_0 (k) | \le T(\log N)^2$, by our restriction to $\mathsf{E}_3$. Additionally, 
		\begin{flalign}
			\label{qt5}
			\displaystyle\max_{j \in \llbracket N_1, N_2 \rrbracket} | q_j (t) - q_j (0) | \le  \displaystyle\int_0^t  |b_j (s)|  ds \le  T \log N,
		\end{flalign}
		
		\noindent where in the first inequality we used \eqref{qtpt} and \eqref{abr}; in the second we used the fact that $t \in [0, T]$ and our restriction to $\mathsf{E}_2$. Moreover,
		\begin{flalign}
			\label{qt6}
			- 2 T \log N \le -t \max_{\lambda \in \eig \bm{L}} |\lambda| -  \log N \le \mathfrak{C}(t) \le t \max_{\lambda \in \eig \bm{L}} |\lambda| +  \log N \le 2T \log N,
		\end{flalign}
		
		\noindent where the first and fourth inequalities follow from our restriction to $\mathsf{E}_2$; the second from the definition \eqref{ct} of $\mathfrak{C}(t)$, with the fact that there exists at least one index $j \in \llbracket 1, N \rrbracket$ for which $u_j (N_1; 0)^2 \ge N^{-1}$; and the third from the definition \eqref{ct} of $\mathfrak{C}(t)$, with the fact that $u_j (N_1;0)^2 \le 1$ for each $j \in \llbracket 1, N \rrbracket$. We further have that 
		\begin{flalign}
			\label{ki2}
			\begin{aligned}
				& \# \{ i \in \llbracket N_1, N_2 \rrbracket : \varphi_t (i) < \varphi_t (k), \varphi_0 (i) > \varphi_0 (k) \} \le 2 T (\log N)^2; \\ 
				& \# \{ i \in \llbracket N_1, N_2 \rrbracket: \varphi_t (i) > \varphi_t (k), \varphi_0 (i) < \varphi_0 (k) \} \le 2T (\log N)^2.
			\end{aligned}
		\end{flalign}
		
		\noindent Indeed, to verify the first bound in \eqref{ki2}, observe (by our restriction to the event $\mathsf{E}_3$) that if $\varphi_t (i) < \varphi_t (k) - 2T (\log N)^2$ then $\varphi_0 (i) \le \varphi_t (i) + T (\log N)^2 < \varphi_t (k) - T(\log N)^2 \le \varphi_0 (k)$. In particular, we can only have $\varphi_t (i) < \varphi_t (k)$ and $\varphi_0 (i) > \varphi_0 (k)$ if $\varphi_t (i) \in \llbracket \varphi_t (k) - 2T (\log N)^2, \varphi_t (k) - 1 \rrbracket$, meaning that there are at most $2T(\log N)^2$ such indices $i$. This shows the first statement in \eqref{ki2}; the second is confirmed entirely analogously.
		
		Inserting \eqref{qt4}, \eqref{qt5}, \eqref{qt6}, and \eqref{ki2} into \eqref{qkt3}, we obtain for any $k \in \llbracket N_1, N_2 \rrbracket$ that 
		\begin{flalign*}
			\Bigg| Q_k & (t) - Q_k (0) - q_{N_1} (t) + q_{N_1} (0) - \lambda_k t - \mathfrak{C}(t)  \\
			&  + 2 \displaystyle\sum_{i:\varphi_t (i) < \varphi_t (k)} \log |\lambda_i - \lambda_k| - 2 \displaystyle\sum_{i:\varphi_0 (i) < \varphi_0 (k)} \log |\lambda_i - \lambda_k| \Bigg| \\ 
			& \qquad \qquad \qquad \le T (\log N)^3 + 5 T \log N +8T (\log N)^4 \le 9T (\log N)^4.
		\end{flalign*} 
		
		\noindent Summing this inequality over $k \in \llbracket N_1, N_2 \rrbracket$ with $\varphi_0 (k) \notin \llbracket N_1 + T (\log N)^6 + N^{1/100}, N_2 - T (\log N)^6 - N^{1/100} \rrbracket$ with \eqref{qkt2} over $k \in \llbracket N_1, N_2 \rrbracket$ with $\varphi_0 \in \llbracket N_1 + T(\log N)^6 + N^{1/100}, N_2 - T (\log N)^6 - N^{1/100} \rrbracket$, we obtain 
		\begin{flalign}
			\begin{aligned}
				\label{qt7} 
				\Bigg| & \displaystyle\sum_{k=N_1}^{N_2} ( q_k (t) - q_k (0) ) - t \displaystyle\sum_{j=1}^N \lambda_j + 2 \displaystyle\sum_{k=N_1}^{N_2} \displaystyle\sum_{i:\varphi_t (i) < \varphi_t (k)} \log |\lambda_i - \lambda_k| \\
				& \quad - 2 \displaystyle\sum_{k=N_1}^{N_2} \displaystyle\sum_{i:\varphi_0 (i) < \varphi_t (k)} \log |\lambda_i - \lambda_k| + N \cdot \big( q_{N_1} (0) - q_{N_1} (t) -  \mathfrak{C}(t) \big) \Bigg| \\
				& \qquad \qquad \qquad \qquad \le 2N(\log N)^6 + 18 T^2 (\log N)^{10} + 18T N^{1/100} (\log N)^4.
			\end{aligned} 
		\end{flalign} 
		
		\noindent The difference between of the third and fourth terms on the left side of \eqref{qt7} is equal to $0$, since both are equal to $\sum_{i \ne j} \log |\lambda_i - \lambda_j|$. The difference between the first and the second is also equal to $0$, as 
		\begin{flalign*}
			\displaystyle\sum_{k=N_1}^{N_2} ( q_k (t) - q_k (0) ) = \displaystyle\sum_{k=N_1}^{N_2} \displaystyle\int_0^t p_k (s) ds = \displaystyle\int_0^t \Tr \bm{L}(s) ds = t \cdot \Tr \bm{L}(0) = t \displaystyle\sum_{j=1}^N \lambda_j,
		\end{flalign*} 
		
		\noindent where the first statement follows from \eqref{qtpt}; the second from \eqref{abr} and \Cref{matrixl}; and the third and fourth from \Cref{ltt}. It follows from \eqref{qt7} that 
		\begin{flalign*}
			\big| q_{N_1} & (0) - q_{N_1} (t) - \mathfrak{C} (t) \big| \\ 
			& \le 2 (\log N)^6 + 18T^2 N^{-1} (\log N)^{10} + 18T N^{-99/100} (\log N)^4 \le (\log N)^{11},
		\end{flalign*} 
		
		\noindent where we used the fact that $T^2 \le N$. This establishes the lemma.
	\end{proof}

	\subsection{Proof of \Cref{ztlambda2}}
	
	\label{ProofQ2}

	In this section we establish \Cref{ztlambda2}. We first prove the following variant of it that fixes the time $t \in [0, T]$. Recall we adopt \Cref{lbetaeta} throughout.

	\begin{thm}
		
		\label{ztlambda4}
		
		For any $t \in [0, T]$, the following holds with overwhelming probability. For any $k \in \llbracket 1, N \rrbracket$ satisfying \eqref{n1n2k02}, we have
		\begin{flalign}
			\label{lambdak4}
			\Bigg| \lambda_k t - Q_k(t) + Q_k (0) - 2 \displaystyle\sum_{i: \varphi_t (i) < \varphi_t (k)}  \log |\lambda_{k} - \lambda_{i}| + 2 \displaystyle\sum_{i: \varphi_0 (i) < \varphi_0 (k)} \log |\lambda_k - \lambda_i| \Bigg| \le (\log N)^{12}.
		\end{flalign} 
		
	\end{thm}
	
	\begin{proof} 
		
		By \Cref{ztlambda3}, \eqref{lambdak3} holds with overwhelming probability, if $T^2 \le N$. It suffices to show such a bound continues to hold for larger values of $T \le N (\log N)^{-7}$ (and to weaken the constraint \eqref{n1n2k022} on $k$ to \eqref{n1n2k02}). 
		
		To that end, we first apply \eqref{lambdak3} on a Toda lattice on a larger interval (at thermal equilibrium), and then use comparison estimates (such as \Cref{lleigenvalues} and \Cref{a2p2}) to approximate the original Toda lattice by the enlarged one. To implement this, first let $\tilde{N}_1 \le \tilde{N}_2$ be integers satisfying 
		\begin{flalign} 
			\label{nn} 
			\tilde{N}_1 + T^2 N \le N_1 \le N_2 \le \tilde{N}_2 - T^2 N, \quad \text{and} \quad \tilde{N} \le N^5,
		\end{flalign}  
		
		\noindent where $\tilde{N} = \tilde{N}_2 - \tilde{N}_1 + 1$. Let $(\tilde{\bm{a}} (s); \tilde{\bm{b}} (s)) \in \mathbb{R}^{\tilde{N}} \times \mathbb{R}^{\tilde{N}}$ denote the Flaschka variables for a Toda lattice on $\llbracket \tilde{N}_1, \tilde{N}_2 \rrbracket$; letting $\tilde{\bm{a}} (s) = (\tilde{a}_{\tilde{N}_1} (s), \tilde{a}_{\tilde{N}_1+1} (s), \ldots , \tilde{a}_{\tilde{N}_2} (s))$ and $\tilde{\bm{b}} (s) = (\tilde{b}_{\tilde{N}_1} (s), \tilde{b}_{\tilde{N}_1+1} (s), \ldots , \tilde{b}_{\tilde{N}_2}(s))$, they satisfy $\tilde{a}_{\tilde{N}_2} (s) = 0$, and \eqref{derivativepa} holds for each $(j, t) \in \llbracket \tilde{N}_1, \tilde{N}_2 \rrbracket \times \mathbb{R}$. We sample the initial data $(\tilde{\bm{a}}(0);\tilde{\bm{b}}(0))$ according to the thermal equilibrium $\mu_{\beta,\theta;\tilde{N}-1,\tilde{N}}$ of \Cref{mubeta2}; we couple $(\tilde{\bm{a}}(0);\tilde{\bm{b}}(0))$ with $(\bm{a}(0);\bm{b}(0))$ so that $(\tilde{a}_i(0),\tilde{b}_i(0)) = (a_i(0), b_i(0))$ for all $i \in \llbracket N_1, N_2 - 1 \rrbracket$. 
		
		For any $s \in \mathbb{R}$, denote the Lax matrix associated with $(\tilde{\bm{a}}(s);\tilde{\bm{b}}(s))$ (as in \Cref{matrixl}) by $\tilde{\bm{L}}(s) = [\tilde{L}_{ij}(s)] \in \SymMat_{\tilde{N}\times\tilde{N}}$. Set $\eig \tilde{\bm{L}}(s) = (\tilde{\lambda}_1, \tilde{\lambda}_2, \ldots , \tilde{\lambda}_{\tilde{N}})$, which does not depend on $s$ (by \Cref{ltt}). For each $s \in \mathbb{R}_{\ge 0}$, let $\tilde{\varphi}_s : \llbracket 1, \tilde{N} \rrbracket \rightarrow \llbracket \tilde{N}_1, \tilde{N}_2 \rrbracket$ denote a $\zeta$-localization center bijection for $\tilde{\bm{L}}(s)$. Further let $(\tilde{\bm{p}} (s); \tilde{\bm{q}}(s)) \in \mathbb{R}^{\tilde{N}} \times \mathbb{R}^{\tilde{N}}$ denote the Toda state space variables associated with $(\tilde{\bm{a}}(s); \tilde{\bm{b}}(s))$, as in \Cref{Open}, where we have indexed $\tilde{\bm{p}} (s) = (\tilde{p}_{\tilde{N}_1} (s), \tilde{p}_{\tilde{N}_1+1} (s), \ldots , \tilde{p}_{\tilde{N}_2} (s))$ and $\tilde{\bm{q}} (s) = (\tilde{q}_{\tilde{N}_1} (s), \tilde{q}_{\tilde{N}_1+1} (s), \ldots , \tilde{q}_{\tilde{N}_2} (s))$. For each $s \in \mathbb{R}$ and $i \in \llbracket 1, \tilde{N} \rrbracket$, denote $\tilde{Q}_i (s) = \tilde{q}_{\tilde{\varphi}_s (i)} (s)$.
		
		By \eqref{lambdak3} and the fact that $T^2 \le \tilde{N} \le N^5$, there exists an overwhelmingly probable event $\mathsf{E}_1$, on which we have 
		\begin{flalign}
			\label{lambda22} 
			\begin{aligned} 
				\Bigg| \tilde{\lambda}_m t - \tilde{Q}_m (t) + \tilde{Q}_m (0) - 2 \displaystyle\sum_{i: \tilde{\varphi}_t (i) < \tilde{\varphi}_t (m)} \log |\tilde{\lambda}_m - \tilde{\lambda}_i| + 2&  \displaystyle\sum_{i: \tilde{\varphi}_0 (i) < \tilde{\varphi}_0 (m)} \log |\tilde{\lambda}_m - \tilde{\lambda}_i| \Bigg| \\
				& \qquad \qquad \le 5^{12} (\log N)^{11},
			\end{aligned} 
		\end{flalign}
		
		\noindent for any $m \in \llbracket 1, \tilde{N} \rrbracket$ satisfying 
		\begin{flalign}
			\label{2m} 
			\tilde{N}_1 + T (\log \tilde{N})^6 + \tilde{N}^{1/100} \le \tilde{\varphi}_0 (m) \le \tilde{N}_2 - T(\log \tilde{N})^6 - \tilde{N}^{1/100}.
		\end{flalign}
		
		\noindent In what follows, we restrict to $\mathsf{E}_1$. We must therefore approximate the parameters $(\tilde{Q}_j (s), \tilde{\lambda}_j, \tilde{\varphi}_s (j))$ by $(Q_j (s), \lambda_j, \varphi_s (j))$, which we will do using \Cref{a2p2} and \Cref{lleigenvalues} (with the fact that $\bm{q}(0)$ and $\tilde{\bm{q}}(0)$ coincide near the origin). 
		
		To do so, we restrict to several additional events. Recalling \Cref{adelta}, restrict to the event 
		\begin{flalign*} 
			\mathsf{E}_2 & = \mathsf{SEP}_{\bm{L}(0)} (e^{-(\log N)^2}) \cap \mathsf{SEP}_{\tilde{\bm{L}}(0)} (e^{-(\log N)^2}) \\
			& \qquad \quad \cap \bigcap_{r \ge 0} \mathsf{BND}_{\bm{L}(r)} \Big( \displaystyle\frac{\log N}{1600} \Big) \cap \mathsf{BND}_{\tilde{\bm{L}}(r)} \Big( \displaystyle\frac{\log N}{1600} \Big),
		\end{flalign*} 
		
		\noindent which we may by \Cref{l0eigenvalues} and \Cref{eigenvalues0}. Further restrict to the event $\mathsf{E}_3$ on which \Cref{qijsalpha} holds, with the $\bm{q}$ there equal to both $\bm{q}$ here and $\tilde{\bm{q}}$ here. Additionally restrict to the event $\mathsf{E}_4$ on which \Cref{centert} holds, with the $(\bm{L};\varphi_j)$ there equal to both $(\bm{L}(0); \varphi_0 (j))$ and $(\tilde{\bm{L}}(0); \tilde{\varphi}_0 (j))$ here, and on which \Cref{centerdistance} holds, with the $\bm{L}(t)$ there equal to both $\bm{L}(t)$ and $\tilde{\bm{L}}(t)$ here.
		
		Now, setting $K = T(\log N)^3/2$, we have by (the $A = \log N / 400$ case of) \Cref{a2p2} that 
		\begin{flalign}
			\label{l11} 
			\displaystyle\sup_{t \in [0, T]} \displaystyle\max_{i \in \llbracket N_1 + K, N_2 - K \rrbracket} \big( |a_i (t) - \tilde{a}_i (t)| + |b_i (t) - \tilde{b}_i (t)| \big) \le e^{-(\log N)^3/10}. 
		\end{flalign}  
		
		\noindent Next, by \Cref{ltl0}, there are random matrices $\bm{M} = [M_{ij}] \in \SymMat_{\llbracket N_1, N_2 \rrbracket}$ and $\tilde{\bm{M}} = [\tilde{M}_{ij}] \in \SymMat_{\llbracket \tilde{N}_1, \tilde{N}_2 \rrbracket}$ with the same laws as $\bm{L}(0)$ and $\tilde{\bm{L}}(0)$, respectively, and an overwhelmingly probable event $\mathsf{E}_5$, on which we have
		\begin{flalign}
			\label{lm11}
			\begin{aligned}  
				& \displaystyle\max_{i,j \in \llbracket N_1+K, N_2-K \rrbracket} | M_{ij} - L_{ij} (t) | \le e^{-(\log N)^3 / 10}; \\
				& \displaystyle\max_{i,j \in \llbracket N_1+K, N_2-K \rrbracket} | \tilde{M}_{ij} - \tilde{L}_{ij} (t) | \le e^{-(\log N)^3/10}.
			\end{aligned}
		\end{flalign} 
		
		\noindent Restricting to $\mathsf{E}_5$, we may therefore (by \eqref{l11} and \eqref{lm11}) further restrict to the event $\mathsf{E}_6$ on which \Cref{lleigenvalues2} and \Cref{lleigenvalues} hold, with the $(\delta; \mathcal{D})$ there equal to $(2e^{-(\log N)^3/10}; \llbracket \tilde{N}_1, \tilde{N}_2 \rrbracket \setminus \llbracket N_1+K, N_2-K \rrbracket)$ here, and the $(\bm{L}, \tilde{\bm{L}})$ equal to any of $(\bm{L}(0), \tilde{\bm{L}}(0))$, $(\bm{M},\bm{L}(t))$, $(\tilde{\bm{M}},\tilde{\bm{L}}(t))$, and $(\tilde{\bm{M}}, \bm{M})$ here (where we view $\bm{M}$ and $\bm{L}(0)$ as $\tilde{N} \times \tilde{N}$ matrices by setting $M_{ij} = L_{ij} (0) = 0$ if $(i,j) \in \llbracket \tilde{N}_1, \tilde{N}_2 \rrbracket \setminus \llbracket N_1, N_2 \rrbracket$). 
		
		Now, let $i \in \llbracket 1, N \rrbracket$ be any index satisfying 
		\begin{flalign}
			\label{0j}  
			N_1 + 3 T(\log \tilde{N})^3 \le \varphi_0 (i) \le N_2 -  3 T (\log \tilde{N})^3.
		\end{flalign} 
		Our restriction to $\mathsf{E}_4$ implies from \eqref{ujms} that $N_1 +  2 T(\log \tilde{N})^3 \le \varphi_t (i) \le N_2 - 2 T(\log \tilde{N})^3$. Recalling that $T \ge N^{1/2} \ge 250$ (as otherwise the theorem follows from \Cref{ztlambda3}), this will enable us to use our restriction to $\mathsf{E}_6$ to apply \Cref{lleigenvalues} twice, first with the $(\bm{L}, \tilde{\bm{L}})$ there equal to $(\bm{M}; \bm{L}(t))$ here and then with them equal to $(\tilde{\bm{M}}; \bm{M})$, and also to apply \Cref{lleigenvalues2} with the $(\bm{L}; \tilde{\bm{L}})$ there equal to $(\tilde{\bm{M}}; \tilde{\bm{L}}(t))$. The first yields a constant $c_2>0$ and an eigenvalue $\mu \in \eig \bm{M}$ such that $|\mu - \lambda_i| \le c_2^{-1} e^{-c_2 (\log N)^3}$ and $\varphi_t (i)$ is an $N^{-1} \zeta$-localization center for $\mu$ with respect to $\bm{M}$. The second yields an eigenvalue $\tilde{\mu} \in \eig \tilde{\bm{M}}$ such that $|\mu - \tilde{\mu}| \le c_2^{-1} e^{-c_2 (\log N)^3}$ and $\varphi_t (i)$ is an $N^{-2} \zeta$-localization center for $\tilde{\mu}$ with respect to $\tilde{\bm{M}}$. The third yields an index $\sigma(i) \in \llbracket 1, \tilde{N} \rrbracket$ such that $|\tilde{\mu} - \tilde{\lambda}_{\sigma(i)}| \le c_2^{-1} e^{-c_2 (\log N)^3}$ and $\varphi_t (i) $ is an $N^{-3} \zeta$-localization center of $\tilde{\lambda}_{\sigma(i)}$ with respect to $\tilde{\bm{L}}(t)$. By \Cref{centerdistance} (and our restriction to $\mathsf{E}_4$), we have $|\varphi_t (i) - \tilde{\varphi}_t (\sigma(i))| \le (\log \tilde{N})^3$, and hence 
		\begin{flalign}
			\label{lambdalambdaj} 
			\begin{aligned} 
				|\lambda_i - \tilde{\lambda}_{\sigma(i)}  | \le 3c_2^{-1} & e^{-c_2 (\log N)^3}; \quad \big| \varphi_t (i) - \tilde{\varphi}_t (\sigma(i)) \big| \le  (\log \tilde{N})^3; \\
				& \big| \varphi_0 (i) - \tilde{\varphi}_0 (\sigma(i)) \big| \le  (\log \tilde{N})^3,
			\end{aligned}
		\end{flalign} 
		
		\noindent where the last bound follows again from similar reasoning to the second (taken at $t=0$). Observe by our restriction to $\mathsf{E}_2 \subseteq \mathsf{SEP}_{\tilde{\bm{L}}(t)} (e^{-(\log N)^2})$ and the first inequality in \eqref{lambdalambdaj} that the map from $i$ to $\sigma(i)$ is injective. 
		
		Then, by our restriction to $\mathsf{E}_1$ on which \eqref{lambda22} holds (taking the $m$ there to be $\sigma(k)$, which satisfies \eqref{2m} by \eqref{n1n2k02} and \eqref{lambdalambdaj}), and the fact that $\llbracket N_1 + (\log N)^6, N_2 - (\log N)^6 \rrbracket \subseteq \llbracket N_1 + 3 (\log \tilde{N})^3, N_2 - 3 (\log \tilde{N})^3 \rrbracket \subseteq \llbracket \tilde{N}_1 + T (\log \tilde{N})^6 + \tilde{N}^{1/100}, \tilde{N}_2 - T(\log \tilde{N})^6 - \tilde{N}^{1/100} \rrbracket$ (by \eqref{nn}), it suffices to show that 
		\begin{flalign}
			\label{2lambdak2} 
			\begin{aligned} 
				& \qquad \qquad \qquad \qquad \quad |\lambda_k - \tilde{\lambda}_{\sigma(k)}| \le T^{-1}; \\
				&  |Q_k (0) - \tilde{Q}_{\sigma(k)} (0)| \le (\log N)^5; \qquad |Q_k (t) - \tilde{Q}_{\sigma(k)} (t)| \le (\log N)^5,
			\end{aligned} 
		\end{flalign} 
		
		\noindent and 
		\begin{flalign}
			\label{sum2} 
			\begin{aligned}
				\Bigg| & \displaystyle\sum_{i = \tilde{N}_1}^{\tilde{N}_2} ( \mathbbm{1}_{\tilde{\varphi}_t (i) < \tilde{\varphi}_t (\sigma(k))} - \mathbbm{1}_{\tilde{\varphi}_0 (i) < \tilde{\varphi}_0 (\sigma(k))}) \cdot \log |\tilde{\lambda}_{\sigma(k)} - \tilde{\lambda}_i| \\ 
				& \qquad - \displaystyle\sum_{i = N_1}^{N_2} (\mathbbm{1}_{\varphi_t (i) < \varphi_t (k)} - \mathbbm{1}_{\varphi_0 (i) < \varphi_0 (k)}) \cdot \log |\lambda_k - \lambda_i| \Bigg| \le (\log N)^9.
			\end{aligned} 
		\end{flalign}

		Taking $i = k$ in the first statement of \eqref{lambdalambdaj} (using \eqref{n1n2k02} to verify \eqref{0j}) yields the first bound in \eqref{2lambdak2}. The proofs of the second and third are entirely analogous to each other, so we only verify the third. To do so, observe by \eqref{qtpt}, \eqref{abr}, and \eqref{l11} that, for any $i \in \llbracket N_1+K,N_2-K \rrbracket$, 
		\begin{flalign}
			\label{q3} 
			| q_i (t) - \tilde{q}_i (t) | \le \displaystyle\int_0^t | b_i (s) - \tilde{b}_i (s) | ds \le 2T \cdot e^{-(\log N)^3/10} \le 1,
		\end{flalign} 
		
		\noindent as $q_i (0) = \tilde{q}_i (0)$. Additionally, \eqref{qiqjs4} implies for $i, j \in \llbracket N_1 + T(\log N)^3, N_2 - T (\log N)^3 \rrbracket$ that 
		\begin{flalign}
			\label{alphaq} 
			| q_i (t) - q_j (t) | \le \alpha \cdot |i-j| + |i-j|^{1/2} \cdot (\log N)^2,
		\end{flalign} 
		
		\noindent by our restriction to $\mathsf{E}_3$. Combining this with \eqref{q3}, \eqref{lambdalambdaj}, and the definition \eqref{qjs2} of $Q_k$ yields 
		\begin{flalign*}
			| Q_k (t) - \tilde{Q}_{\sigma(k)} (t) | & \le | q_{\varphi_t (k)} (t) - q_{\tilde{\varphi}_t (\sigma(k))} (t) | + | \tilde{q}_{\varphi_t (k)} (t) - \tilde{q}_{\tilde{\varphi}_t (\sigma(k))} (t) | \\
			& \le \alpha \cdot | \varphi_t (k) - \tilde{\varphi}_t (\sigma(k))| +  | \varphi_t (k) - \tilde{\varphi}_t (\sigma(k)) |^{1/2} \cdot (\log N)^2 + 1 \\ 
			&  \le  (\log N)^5,
		\end{flalign*}
		
		\noindent which confirms the third bound in \eqref{2lambdak2}. 
		
		It therefore remains to verify \eqref{sum2}. To that end, let 
		\begin{flalign*} 
			& \mathcal{S} = \big\{ i \in \llbracket 1, N \rrbracket : \varphi_0 (i) \in \llbracket N_1 + 2T(\log N)^6, N_2 - 2T(\log N)^6 \rrbracket \big\}; \\
			& \tilde{\mathcal{S}} = \big\{ i \in \llbracket 1, \tilde{N} \rrbracket : \tilde{\varphi}_0 (i) \in \llbracket N_1 + 2T(\log N)^6, N_2 - 2T(\log N)^6 \rrbracket \big\},
		\end{flalign*} 
		
		\noindent and $\sigma (\mathcal{S}) = \{ \sigma(i): i \in \mathcal{S} \}$. First suppose that $i \in \llbracket 1, N \rrbracket \setminus \mathcal{S}$. Then, since \eqref{ujms} (and our restriction to $\mathsf{E}_4$) implies that $|\varphi_t (i) - \varphi_0 (i)| \le T (\log N)^2$, we have from \eqref{n1n2k02} that $\mathbbm{1}_{\varphi_t (i) < \varphi_t (k)} = \mathbbm{1}_{\varphi_0 (i) < \varphi_0 (k)}$. Hence, we may restrict the second sum in \eqref{sum2} to $i \in \mathcal{S}$. Similarly, we may restrict the first sum in \eqref{sum2} to $i \in \tilde{\mathcal{S}}$. Therefore, we have 
		\begin{flalign}
			\label{sum4}  
			\begin{aligned}
				& \Bigg| \displaystyle\sum_{i = \tilde{N}_1}^{\tilde{N}_2} ( \mathbbm{1}_{\tilde{\varphi}_t (i) < \tilde{\varphi}_t (\sigma(k))} - \mathbbm{1}_{\tilde{\varphi}_0 (i) < \tilde{\varphi}_0 (\sigma(k))}) \cdot \log |\tilde{\lambda}_{\sigma(k)} - \tilde{\lambda}_i| \\ 
				& \qquad \qquad -  \displaystyle\sum_{i = N_1}^{N_2} (\mathbbm{1}_{\varphi_t (i) < \varphi_t (k)} - \mathbbm{1}_{\varphi_0 (i) < \varphi_0 (k)}) \cdot \log |\lambda_k - \lambda_i| \Bigg|  \le A + B + C.
			\end{aligned} 
		\end{flalign} 
		
		\noindent where 
		\begin{flalign*} 
			A & =  \displaystyle\sum_{i \in \mathcal{S}} \big| \log |\tilde{\lambda}_{\sigma(k)} - \tilde{\lambda}_{\sigma(i)}| - \log |\lambda_k - \lambda_i| \big|; \\
			B & = \displaystyle\sum_{i \in \mathcal{S}} \big( \big| \log| \tilde{\lambda}_{\sigma(k)} - \tilde{\lambda}_{\sigma(i)}| \big| + \big| \log |\lambda_k - \lambda_i| \big| \big) \\
			& \qquad \qquad \qquad \times \big| \big( \mathbbm{1}_{\tilde{\varphi}_t (\sigma(i)) < \tilde{\varphi}_t (\sigma(k))} - \mathbbm{1}_{\varphi_0 (\sigma(i)) < \tilde{\varphi}_0 (\sigma(k))} \big) - \big( \mathbbm{1}_{\varphi_t (i) < \varphi_t (k)} - \mathbbm{1}_{\varphi_0 (i) < \varphi_0 (k)} \big) \big|; \\
			C & = \# \big( (\tilde{\mathcal{S}} \cup \sigma(\mathcal{S})) \setminus (\tilde{\mathcal{S}} \cap \sigma(\mathcal{S})) \big) \cdot \displaystyle\max_{i \ne \sigma(k)} \big|\log |\tilde{\lambda}_{\sigma(k)} - \tilde{\lambda}_i| \big|.
		\end{flalign*}
		
		By our restriction to the event $\mathsf{E}_2$, we have 
		\begin{flalign} 
			\label{lambda3} 
			e^{-(\log N)^2} \le  |\lambda_i - \lambda_j|  \le 2 \log N, \quad \text{and} \quad e^{-(\log N)^2} \le |\tilde{\lambda}_i - \tilde{\lambda}_j| \le 2 \log N, \quad \text{whenever $i \ne j$}.
		\end{flalign}
		
		\noindent  Together with the first bound in \eqref{lambdalambdaj}, this implies that $|\log |\tilde{\lambda}_{\sigma(k)} - \tilde{\lambda}_{\sigma(i)}| - \log |\lambda_k-\lambda_i|| \le e^{-(\log N)^2}$ for each $i \in \mathcal{S}$, and so 
		\begin{flalign} 
			\label{a} 
			A \le N e^{-(\log N)^2} \le 1. 
		\end{flalign} 
		
		\noindent By the second and third bounds in \eqref{lambdalambdaj}, that are at most $500 (\log N)^3$ indices $i \in \mathcal{S}$ for which 
		\begin{flalign*} 
			\big( \mathbbm{1}_{\tilde{\varphi}_t (\sigma(i)) < \tilde{\varphi}_t (\sigma(k))} - \mathbbm{1}_{\varphi_0 (\sigma(i)) < \tilde{\varphi}_0 (\sigma(k))} \big) - \big( \mathbbm{1}_{\varphi_t (i) < \varphi_t (k)} - \mathbbm{1}_{\varphi_0 (i) < \varphi_0 (k)} \big)  \ne 0.
		\end{flalign*} 
		
		\noindent Since this quantity is always bounded above by $2$, \eqref{lambda3} implies that 
		\begin{flalign}
			\label{b} 
			B \le 2000 (\log N)^5.
		\end{flalign} 
		
		\noindent By the last inequality in \eqref{lambdalambdaj} and the injectivity of $\sigma$, it is quickly verified that $ \# ( (\tilde{\mathcal{S}} \cup \sigma(\mathcal{S})) \setminus (\tilde{\mathcal{S}} \cap \sigma(\mathcal{S})) ) \le 500 (\log N)^3$, from which we deduce again by \eqref{lambda3} that $C \le 500 (\log N)^5$. This, with \eqref{sum4}, \eqref{a}, and \eqref{b}, yields \eqref{sum2} and thus the theorem.				
	\end{proof}
	
	Next, establish \Cref{ztlambda2}, which proves the bound in \Cref{ztlambda4} uniformly in $t \in [0, T]$. 
	
	\begin{proof}[Proof of \Cref{ztlambda2} (Outline) ]
		
		This theorem will follow from applying \Cref{ztlambda4} on a mesh of times $\mathcal{T} \subset [0, T]$, and using continuity bounds to show that it continues to hold on all of $[0, T]$. Since the proof is similar to the discussions at the end of the proofs of \Cref{ztlambda} and \Cref{ztlambda4}, we only outline it. 
		
		Let $\mathfrak{c}$ denote the constant $c$ from \Cref{ztlambda4} and, for any $t \in [0, T]$, let $\mathsf{E}(t)$ denote the event on which \eqref{lambdak4} holds for all $k \in \llbracket 1, N \rrbracket$ satisfying \eqref{n1n2k02}. Denoting $\mathfrak{c}' = \mathfrak{c}/2$, let $\mathcal{T} \subset [0, T]$ denote a $e^{-\mathfrak{c}'(\log N)^2}$-mesh of $[0, T]$; by a union bound, we may restrict to the event $\mathsf{F}_1 = \bigcap_{s \in \mathcal{T}} \mathsf{E}(s)$. We further restrict to the event $\mathsf{F}_2$ on which \Cref{centerdistance2} holds, with the $\mathfrak{d}$ there equal to $\mathfrak{c}'$ here; we additionally restrict to the intersection $\mathsf{F}_3$ of all of the events from the proof of \Cref{ztlambda4}.
		
		Now fix $t \in [0, T]$, and let $s \in \mathcal{S}$ be such that $|s-t| \le e^{-\mathfrak{c}' (\log N)^2}$. By our restriction to $\mathsf{F}_1$, we have from \eqref{lambdak4} that 
		\begin{flalign*}
			\Bigg| \lambda_k s - Q_k (s) + Q_k (0) - 2 \displaystyle\sum_{i: \varphi_s (i) < \varphi_s (k)} \log |\lambda_k - \lambda_i| + 2 \displaystyle\sum_{i: \varphi_0 (i) < \varphi_0 (k)} \log |\lambda_k - \lambda_i| \Bigg| \le  (\log N)^{12}.
		\end{flalign*}
		
		\noindent Since $|\lambda_k| \cdot |s-t| \le \log N \cdot e^{-\mathfrak{c}' (\log N)^2} \le 1$ and $|\log |\lambda_k-\lambda_i|| \le (\log N)^2$ for $i \ne k$ (both by our restriction to $\mathsf{BND}_{\bm{L}(0)} (\log N) \cap \mathsf{SEP}_{\bm{L}(0)} (e^{-(\log N)^2})$, from the event $\mathsf{E}_2$ in the proof of \Cref{ztlambda4}), it suffices to show that
		\begin{flalign}
			\label{q5} 
			| Q_k (t) - Q_k (s)| \le (\log N)^{10}; \qquad \displaystyle\sum_{i=N_1}^{N_2} |  \mathbbm{1}_{\varphi_t (i) < \varphi_t (k)} - \mathbbm{1}_{\varphi_s (i) < \varphi_s (k)}| \le (\log N)^8.
		\end{flalign} 
		
		The first bound in \eqref{q5} follows from the fact that 
		\begin{flalign*}
			|Q_k (t) - Q_k (s)| & \le | q_{\varphi_t (k)} (t) - q_{\varphi_s (k)} (t) | + | q_{\varphi_s (k)} (t) - q_{\varphi_s (k)} (s) | \\ 
			& \le  \alpha \cdot | \varphi_t (k) - \varphi_s (k) | + | \varphi_t (k) - \varphi_s (k) |^{1/2} \cdot (\log N)^2 + |s-t| \cdot  \displaystyle\max_{r \in [s, t]} | b_{\varphi_s (k)} (s) | \\
			& \le (\log N)^5. 
		\end{flalign*} 
		
		\noindent where the first statement follows from the definition \eqref{qjs2} of $Q_k$; the second follows from \eqref{alphaq}, \eqref{qtpt}, and \eqref{abr}; and the third follows from the fact that $|\varphi_t (k) - \varphi_s (k)| \le (\log N)^3$ (by \Cref{centerdistance2}, as we restricted to $\mathsf{F}_2$), \Cref{matrixl}, and our restriction to $\mathsf{F}_3 \subseteq \bigcap_{r \ge 0} \mathsf{BND}_{\bm{L}(r)} (\log N)$. The proof of the second bound in \eqref{q5} is entirely analogous to that of \eqref{sum0} (using the fact that $|\varphi_t (j) - \varphi_s (j)| \le (\log N)^3$ for $\varphi_s (j) \in \llbracket N_1 + T(\log N)^3, N_2 - T (\log N)^3 \rrbracket$, by \Cref{centerdistance2} and our restriction to $\mathsf{F}_2$). This establishes the theorem.
	\end{proof}

	\appendix

	\section{Proofs of Results From \Cref{MatrixLattice}}
	
	\label{Proof2}

	\subsection{Proofs of \Cref{aintegral} and \Cref{qij}}
	
	\label{ProofIntegral}

	\begin{proof}[Proof of \Cref{aintegral}]
		
		Observe that the random variable $\mathfrak{r} \in \mathbb{R}$ has density $\mathbb{P} [\mathfrak{r} \in (r,r+dr)] = \beta^{\theta} \cdot \Gamma(\theta)^{-1} \cdot e^{-\theta r - e^{-\beta e^{-r}}} dr$. Thus, denoting
		\begin{flalign}
			\label{ftheta} 
			F(\theta) =  \displaystyle\int_{-\infty}^{\infty} e^{-\theta r -\beta e^{-r}} dr,
		\end{flalign}
		
		\noindent we have 
		\begin{flalign}
			\label{fthetaderivative}
			F'(\theta) = - \displaystyle\int_{-\infty}^{\infty} r e^{-\theta r - \beta e^{-r}} dr = -F(\theta) \cdot \mathbb{E} [\mathfrak{r}].
		\end{flalign}
		
		\noindent Changing variables $u = \beta e^{-r}$ in \eqref{ftheta} yields
		\begin{flalign*}
			F(\theta) =  \beta^{-\theta} \displaystyle\int_0^{\infty} u^{\theta-1} e^{-u} du = \beta^{-\theta} \cdot \Gamma(\theta),
		\end{flalign*} 
		
		\noindent so the lemma follows from \eqref{fthetaderivative}, together with the fact (recall \eqref{alpha}) that $\alpha = -F'(\theta) \cdot F(\theta)^{-1}$.
	\end{proof}

	\begin{proof}[Proof of \Cref{qij}]
		
		We may assume that $i < j$; in what follows, we abbreviate $q_k = q_k (0)$ and $a_k = a_k (0)$ for each $k$.  By \eqref{q00}, we have $q_{m+1}  - q_m  = -2 \log a_m$ for any $m \in \llbracket i+1, j \rrbracket$. Since $(\bm{a}; \bm{b})$ is sampled from the measure $\mu_{\beta,\theta;N-1,N}$ from \Cref{mubeta2}, all such $a_m$ are mutually independent with the same law $\mathfrak{a}$ of density $\mathbb{P} [\mathfrak{a} \in (a,a+da)] = 2\beta^{\theta} \cdot \Gamma(\theta)^{-1} \cdot a^{2\theta-1} e^{-\beta a^2} da$. Hence, denoting $\mathfrak{r} = -2 \log \mathfrak{a}$, each $q_{m+1} - q_m$ has law $\mathfrak{r}$. By \Cref{aintegral}, we have $\alpha = \mathbb{E} [\mathfrak{r}]$, so we find for any $u \in (0, 1 / 2)$ that 
		\begin{flalign}
			\label{qiqjr}
			\mathbb{P} \big[ q_j  - q_i \ge \alpha(j-i) + R \big] \le e^{-uR} \cdot \mathbb{E} [e^{u(q_j-q_i)} \cdot e^{\alpha u (i-j)}] = \big( e^{-uR/(j-i)} \cdot \mathbb{E}[e^{u(\mathfrak{r}-\alpha)}] \big)^{j-i},
		\end{flalign}
		
		\noindent where the first statement follows from a Markov bound, and the second from the mutual independence of the $(q_{m+1}-q_m)$ with the finiteness of the moment generating function $\mathbb{E} [e^{u\mathfrak{r}}] = \mathbb{E} [\mathfrak{a}^{-2u}] < \infty$ for $u \in (0, 1 / 2)$. Setting $u = \min \{ cR / (j-i), c \}$ for some sufficiently small constant $c > 0$, and using the (quickly verified) fact that $e^{-uR/(j-i)} \cdot \mathbb{E}[e^{u(\mathfrak{r}-\alpha)}] \le e^{-c u^2}$ for this $u$, yields from \eqref{qiqjr} that 
		\begin{flalign*} 
			\mathbb{P} \big[ q_j - q_i \ge \alpha(j-i) + R \big] \le e^{-c u^2 |i-j|} \le e^{-c^3 R^2/|j-i|} + e^{-c^3 |i-j|}.
		\end{flalign*} 
		
		\noindent By similar reasoning, we have $\mathbb{P} [ q_j - q_i \le \alpha(j-i) - R ] \le e^{-c^3 R^2/|i-j|}+e^{-c^3|i-j|}$, and so the lemma follows by a union bound.
	\end{proof}

	\subsection{Proofs of Results From \Cref{LinearOpen} and \Cref{ResolventG}} 
	
	\label{ProofR}
	
	\begin{proof}[Proof of \Cref{abltt}]

		Let $\lambda \in \eig \bm{L} (t)$ denote the eigenvalue of $\bm{L}(t)$ with maximal absolute value; if more than one exists, we select one arbitrarily. Then, 
		\begin{flalign}
			\label{lambda01}
			|\lambda| \le \displaystyle\max_{i \in \mathscr{I}} \displaystyle\sum_{j \in \mathscr{I}} | L_{ij} (t') | \le 2 \big( A(t') + B(t') \big). 
		\end{flalign}
		
		\noindent where the first inequality holds since $\lambda \in \eig \bm{L}(t')$ (by \Cref{ltt}), and the second inequality holds by \Cref{matrixl}. Moreover, 	
		\begin{flalign}
			\label{lambda02}
			|\lambda| \ge \displaystyle\max_{i \in \mathscr{I}} | L_{ii} (t) | = B(t); \qquad |\lambda| \ge \displaystyle\frac{1}{2} \cdot \displaystyle\max_{i, j \in \mathscr{I}} \big| L_{ii} (t) + 2L_{ij} (t) + L_{jj} (t) \big| \ge A(t) - B(t),
		\end{flalign} 
		
		\noindent where, in both statements, the first inequality follows from applying the min-max principle (to $\bm{L}(t)$ if $|\lambda| \in \eig \bm{L}(t)$, and to $-\bm{L}(t)$ if $-|\lambda| \in \eig \bm{L}(t)$), and the second inequality follows from \Cref{matrixl}. Combining \eqref{lambda01} and \eqref{lambda02} yields $A(t) + B(t) \le 3 |\lambda| \le 6 ( A(t') + B(t') )$, confirming the lemma.
	\end{proof}

	\begin{proof}[Proof of \Cref{abgh}]
		
		Set $\eig \bm{A} = (\lambda_1, \lambda_2, \ldots , \lambda_n)$ and $\eig \bm{B} = (\mu_1, \mu_2, \ldots , \mu_n)$. Further let $(\bm{u}_j)$ and $(\bm{v}_j)$, for $j \in \llbracket 1, n \rrbracket$, denote orthonormal eigenbases for $\bm{A}$ and $\bm{B}$, respectively. In this way, $\bm{u}_j = ( u_j (i))_{i \in \mathscr{I}}$ is an eigenvector of $\bm{A}$ with eigenvalue $\lambda_j$, and $\bm{v}_j = ( v_j (i) )_{i \in \mathscr{I}}$ is an eigenvector of $\bm{B}$ with eigenvalue $\mu_j$, for each $j \in \llbracket 1, n \rrbracket$. Let $k \in \llbracket 1, n \rrbracket$ be such that $\lambda_k = \lambda$; we may assume that the eigenbasis $(\bm{u}_j)$ includes $\bm{u}$, in particular, that $\bm{u}_k = \bm{u}$. Then,
		\begin{flalign}
			\label{g2k}
			\Imaginary G_{\varphi \varphi} (z) = \displaystyle\sum_{j=1}^n u_j (\varphi)^2 \cdot \Imaginary (\lambda_j - z)^{-1} \ge u_k (\varphi)^2 \cdot \Imaginary (\lambda - z)^{-1} \ge \eta^{-1} \zeta^2. 
		\end{flalign}
		
		\noindent where in the first statement we used \eqref{gsum}; in the second we used the bound $\Imaginary (\lambda_j - z)^{-1} \ge 0$; and in the third we used the fact that $z = \lambda + \mathrm{i} \eta$ and the second estimate in \eqref{deltaetachi}. It follows that
		\begin{flalign}
			\label{sumvj2} 
			\Bigg| \displaystyle\sum_{j = 1}^n \displaystyle\frac{v_j (\varphi)^2}{\mu_j - z} \Bigg| = | H_{\varphi \varphi} (z) | \ge \big| G_{\varphi \varphi} (z) \big| - \delta \ge \eta^{-1} \zeta^2 - \delta \ge (2\eta)^{-1} \zeta^2,
		\end{flalign}
		
		\noindent where the first statement follows from \eqref{gsum}; the second from the third bound in \eqref{deltaetachi}; the third from \eqref{g2k}; and the fourth from the first bound in \eqref{deltaetachi}. 
		
		Therefore, there exists $m \in \llbracket 1, n \rrbracket$ such that $|\mu_m - \lambda| \le 3n \zeta^{-2} \eta$ and $| v_m (\varphi)| \ge (6n)^{-1/2} \zeta$. Indeed, assuming otherwise, we would have
		\begin{flalign*}
			(2 \eta)^{-1} \zeta^2 \le \Bigg| \displaystyle\sum_{j=1}^n \displaystyle\frac{v_j (\varphi)^2}{\mu_j - z} \Bigg| < n \cdot (3n \zeta^{-2} \eta)^{-1} + n \cdot \eta^{-1} \cdot \big( (6n)^{-1/2} \zeta \big)^2 = (2 \eta)^{-1} \zeta^2,
		\end{flalign*}
		
		\noindent where the first bound follows from \eqref{sumvj2} and the second from our assumption (with the fact that $z = \lambda + \mathrm{i} \eta$). This is a contradiction, which establishes the lemma.
	\end{proof}

	\subsection{Proofs of Results From \Cref{EigenvectorM}}
	
	\label{ProofMatrixS}
	
	\begin{proof}[Proof of \Cref{n1n2u}]
		
		For any $E \in \mathbb{R} \setminus \eig \bm{M}$, denote the resolvent $\bm{G} (E) = [ G_{ij} (E) ] = (\bm{M} - E)^{-1}$. By \eqref{gsum} (with the fact by \cite[Proposition 2.40(a)]{PRM} that all eigenvalues of $\bm{M}$ are mutually disjoint), we have 
		\begin{flalign}
			\label{un1un2}
			u_{N_1} \cdot u_{N_2} = \displaystyle\lim_{E \rightarrow \mu} (\mu - E) \cdot G_{N_1 N_2} (E). 
		\end{flalign}
		
		\noindent To evaluate the right side of \eqref{un1un2}, we let $\bm{C}(E) = [ C_{ij} (E) ]$ denote the cofactor matrix of $\bm{M} - E \cdot \Id$ and use the fact that $G_{N_1 N_2} (E) = (-1)^{N+1} \cdot C_{N_1 N_2} (E) \cdot (\det \bm{M} - E \cdot \Id)^{-1}$. Since removing the row of index $N_1$ and column of index $N_2$ from $\bm{M}$ yields an upper triangular $(N-1) \times (N-1)$ matrix with diagonal entries $(M_{i,i+1})_{N_1 \le i < N_2}$, we deduce that $C_{N_1 N_2} (E) = \prod_{i=N_1}^{N_2 - 1} M_{i,i+1}$. Hence,
		\begin{flalign}
			\label{gn1n2e} 
			\begin{aligned} 
				G_{N_1 N_2} (E) & = (-1)^{N+1} \cdot C_{N_1 N_2} (E) \cdot (\det \bm{M} - E \cdot \Id)^{-1} \\
				& = (-1)^{N+1} \cdot \displaystyle\prod_{i=N_1}^{N_2-1} M_{i,i+1} \cdot \displaystyle\prod_{\mu' \in \eig \bm{M}} (\mu' - E)^{-1}.
			\end{aligned} 
		\end{flalign}
		
		\noindent Combining \eqref{gn1n2e} with \eqref{un1un2} yields the lemma.
	\end{proof}
	
	\begin{proof}[Proof of \Cref{smatrixu}] 
		
		Since $\bm{M} \cdot \bm{u} = \mu \cdot \bm{u}$ and $\bm{M}$ is tridiagonal, we have for any index $k \in \llbracket N_1, N_2 - 1\rrbracket$ that $M_{k,k+1} \cdot u_{k+1} = (\mu - M_{k,k}) \cdot u_k - M_{k,k-1} \cdot u_{k-1}$. Since $M_{k,k-1} = M_{k-1,k}$, this is equivalent to $\bm{S}_k (\mu) \cdot \bm{w}_k = \bm{w}_{k+1}$, so the lemma follows by induction on $j-i$.
	\end{proof} 
	
	\begin{proof}[Proof of \Cref{sij}]
		
		It quickly follows from the explicit forms \eqref{ks} of $\bm{S}_k$ and \eqref{k2s} of $\bm{S}_{\mathcal{K}}$ that there exist monic polynomials $P_{\ell-1}$, $P_{\ell}$, $Q_{\ell-2}$, and $Q_{\ell-1}$ satisfying the following two properties. First, we have 
		\begin{flalign*}
			\bm{S}_{\llbracket i, j \rrbracket} (E) = \left[ \begin{array}{cc} -Q_{\ell-2} (E)  \cdot M_{i-1,i} \cdot \displaystyle\prod_{k=i}^{j-1} M_{k,k+1}^{-1} & P_{\ell-1} (E) \cdot \displaystyle\prod_{k=i}^{j-1} M_{k,k+1}^{-1}  \\
				-Q_{\ell-1} (E) \cdot M_{i-1,i} \cdot \displaystyle\prod_{k=i}^{j} M_{k,k+1}^{-1} & P_{\ell} (E) \cdot \displaystyle\prod_{k=i}^{j} M_{k,k+1}^{-1} \end{array} \right].
		\end{flalign*}
		
		\noindent Second, we have that $\deg Q_{\ell-2} = \ell-2$; that $\deg P_{\ell - 1} = \deg Q_{\ell-1} = \ell - 1$; and that $\deg P_{\ell} = \ell$. Thus, it suffices to show that 
		\begin{flalign}
			\label{ppqq} 
			\begin{aligned}
				& P_{\ell} (E) = \displaystyle\prod_{h=1}^{\ell} \big( E - \mu_h^{[i,j]} \big); \qquad \qquad P_{\ell-1} (E) = \displaystyle\prod_{h=1}^{\ell-1} \big( E - \mu_h^{[i,j-1]} \big); \\
				& Q_{\ell-1} (E) = \displaystyle\prod_{h=1}^{\ell-1} \big( E - \mu_h^{[i+1,j]} \big); \qquad Q_{\ell-2} (E) = \displaystyle\prod_{h=1}^{\ell-2} \big( E - \mu_h^{[i+1,j-1]} \big). 
			\end{aligned} 
		\end{flalign}
		
		To that end, we locate the zeroes of $P_{\ell}$. Fix any $\mu \in \eig \bm{M}^{[i,j]}$, and let $\bm{v} = (v_i, v_{i+1}, \ldots , v_j) \in \mathbb{R}^{\ell}$ denote an eigenvector of $\bm{M}^{[i,j]}$ with eigenvalue $\mu$. Then setting $\bm{w}_k = (v_{k-1}, v_k)$ for each $k \in \llbracket i, j \rrbracket$, with $v_{i-1} = 0$, we have from \Cref{smatrixu} that $\bm{S}_{\llbracket i, j-1 \rrbracket} (\mu) \cdot \bm{w}_i = \bm{w}_j$. Furthermore, the second coordinate of $S_j (\mu) \cdot \bm{w}_j \in \mathbb{R}^2$ is equal 
		\begin{flalign*} 
			M_{j,j+1}^{-1} \cdot \big(  (\mu - M_{j,j}) v_j - M_{j-1,j} v_{j-1} \big) = M_{j,j+1}^{-1} \cdot \big( \mu v_j - M_{j,j} v_j - M_{j,j-1} v_{j-1} \big) = 0,
		\end{flalign*}  
		
		\noindent where the first statement follows from the fact that $\bm{M}$ is symmetric and the second from the fact that $\bm{v}$ is an eigenvector of $\bm{M}^{[i,j]}$ with eigenvalue $\mu$. Therefore, the second coordinate of $\bm{S}_{\llbracket i, j \rrbracket} (\mu) \cdot \bm{w}_i$ is equal to $0$. Since the first coordinate of $\bm{w}_i$ is equal to $v_{i-1} = 0$ (and its second coordinate is nonzero), this implies that the $(2, 2)$-entry of $\bm{S}_{\llbracket i, j \rrbracket} (\mu)$ is equal to $0$. Hence, $P_{\ell} (\mu) = 0$ for each $\mu \in \eig \bm{M}^{[i,j]}$. Since $\deg P_{\ell} = \ell$ and $P_{\ell}$ is monic, this confirms the first statement in \eqref{ppqq}. 
		
		We can now quickly deduce the remaining equalities in \eqref{ppqq}. Indeed, the second statement in \eqref{ppqq} follows from the first, together with the fact that the $(1,2)$-entry of $\bm{S}_{\llbracket i, j \rrbracket} (E)$ is equal to the $(2,2)$-entry of $\bm{S}_{\llbracket i,j-1 \rrbracket} (E)$ (by \eqref{ks} and \eqref{k2s}). The third statement in \eqref{ppqq} follows from the first, together with the fact that the $(2,1)$-entry of $\bm{S}_{\llbracket i, j \rrbracket} (E)$ is equal to the $(2, 2)$-entry of $\bm{S}_{\llbracket i+1, j \rrbracket} (E)$ multiplied by $-M_{i,i+1}^{-1} M_{i-1,i}$ (again by \eqref{ks} and \eqref{k2s}). The fourth statement in \eqref{ppqq} follows from the third, again together with the fact that the $(1,1)$-entry of $\bm{S}_{\llbracket i, j \rrbracket} (E)$ is equal to the $(2,1)$-entry of $\bm{S}_{\llbracket i,j-1 \rrbracket} (E)$. This establishes the lemma. 
	\end{proof}

	\subsection{Proofs of Results From \Cref{Localization}}
	
	\label{ProofL}
	
	\begin{proof}[Proof of \Cref{l0eigenvalues}]
		
		By \Cref{abltt} and \Cref{ltt}, it suffices to verify the lemma at $t=0$, that is, to show $\mathbb{P} [ \mathsf{BND}_{\bm{L}(0)} (A) ] \ge 1 - c^{-1} N e^{-cA^2}$. To that end, observe by \Cref{mubeta2} (or \Cref{mubeta}) and a union bound, there is a constant $C > 1$ such that  
		\begin{flalign*}
			\mathbb{P} \Bigg[ \displaystyle\max_{a \in \bm{a} (0)} |a| + \displaystyle\max_{b \in \bm{b} (0)} |b| \ge \displaystyle\frac{A}{4} \Bigg] \le CN A^{2\theta} e^{-\beta A^2 / 2}.
		\end{flalign*}
		
		\noindent By the deterministic bound \eqref{lambda01}, this yields $\mathbb{P} [ \max_{\lambda \in \eig \bm{L}(0)} |\lambda| \ge A ] \le CN A^{2 \theta} e^{-\beta A^2/2}$. Together, these two estimates imply $\mathbb{P} [ \mathsf{BND}_{\bm{L}(0)} (A) ] \ge 1 - 2CN A^{2\theta} e^{-\beta A^2/2}$, which (as mentioned above) implies the lemma.
	\end{proof}

	\begin{proof}[Proof of \Cref{eigenvalues0}]
		
		Denote $\eig \bm{L} = (\lambda_1, \lambda_2, \ldots , \lambda_N)$, and define the events 
		\begin{flalign*} 
			\mathsf{E} = \bigcap_{i=1}^{N-1} \{ \lambda_i \ge \lambda_{i+1} + \delta \}; \qquad \mathsf{F} = \bigg\{ \displaystyle\max_{1 \le i \le N} |\lambda_i| \le 3N \bigg\}.
		\end{flalign*} 
		
		\noindent Let $\mathcal{J} = \llbracket -4 \delta^{-1} N, 4 \delta^{-1} N \rrbracket$. For each $j \in \mathcal{I}$, set $z_j = j \delta + \delta \mathrm{i}$, and let 
		\begin{flalign*}
			S =  \displaystyle\sum_{j \in \mathcal{J}} \displaystyle\sum_{1 \le i \ne k \le N} \Imaginary (\lambda_i - z_j)^{-1} \cdot  \Imaginary (\lambda_k - z_j)^{-1}. 
		\end{flalign*}
		
		\noindent Observe on $\mathsf{E}^{\complement} \cap \mathsf{F}$ that there exists some $j \in \mathcal{J}$ and $i \in \llbracket 1, N-1 \rrbracket$ such that $\Imaginary (\lambda_i - z_j)^{-1} \ge (2\delta)^{-1}$ and $\Imaginary (\lambda_{i+1} - z_j)^{-1} \ge (2 \delta)^{-1}$. Hence, on $\mathsf{E}^{\complement} \cap \mathsf{F}$ we have $S \ge (2\delta)^{-2}$, so taking expectations yields 
		\begin{flalign}
			\label{efs} 
			\mathbb{P} \big[ \mathsf{E}^{\complement} \cap \mathsf{F} \big] \le (2\delta)^2 \cdot \mathbb{E}[S].
		\end{flalign}
		
		\noindent Now, set $\bm{G} (z_j) = [ G_{ik} (z_j) ] = (\bm{L} - z_j)^{-1}$ for each $j \in \mathcal{J}$.  By \cite[Equations (2.64) and (2.65)]{LFSM},\footnote{The results there are stated for a different random matrix, but it is quickly verified that its proofs apply for the one $\bm{L}$ we study.} there exists a constant $C_0 > 1$ such that 
		\begin{flalign*}
			\mathbb{E}[S] = \displaystyle\sum_{j \in \mathcal{J}} \displaystyle\sum_{1 \le i, k \le N} \mathbb{E} \Bigg[ \det \bigg[ \begin{array}{cc} \Imaginary G_{ii} (z_j) & \Imaginary G_{ik} (z_j) \\ \Imaginary G_{ik} (z_j) & \Imaginary G_{kk} (z_j) \end{array} \bigg] \Bigg] \le C_0 N^2 |\mathcal{J}|.
		\end{flalign*}
		
		\noindent Together with \eqref{efs} and the fact that $|\mathcal{J}| \le 9 \delta^{-1} N$, this yields 
		\begin{flalign*}
			\mathbb{P} [ \mathsf{E}^{\complement} \cap \mathsf{F} ] \le 36 C_0 \delta N^3.
		\end{flalign*}
		
		\noindent Since \Cref{l0eigenvalues} yields constants $C_1 > 1 > c_1 > 0$ such that $\mathbb{P} [\mathsf{F}^{\complement} ] \le C_1 e^{-c_1 N^2}$, this with a union bound yields the lemma. 
	\end{proof}

	\section{Heuristics for Eigenvalue Velocities} 
	
	\label{0Speed}

	Throughout this section, we adopt \Cref{lbetaeta}. For each $k \in \llbracket 1, N \rrbracket$ and $t \ge 0$, define $v_k (t)$ of $Q_k$ by setting $Q_k (t) = Q_k (0) + t \cdot v_k (t)$. This quantity $v_k (t)$ can be thought of as the ``velocity'' of the eigenvalue $\lambda_k$ under the Toda lattice. In this section, following the physics literature (for example, \cite{GHCS}), we explain how \Cref{ztlambda} can be used to heuristically derive the large $T$ limit for $v_k (T)$, which will coincide with the predictions of \cite[Equation (90)]{GHCS} and \cite[Equation (6.19)]{HSIMS}. The discussion in this section is not entirely rigorous, and we do not know of a direct way to mathematically justify it completely; the sequel work \cite{P} will be devoted to establishing its output (when $\theta$ is sufficiently small) through another route. In what follows, we let $T \gg 1$ be large; abbreviate $v_k (T) = v_k$; and suppose for notational simplicity that $\alpha > 0$ (the case when $\alpha<0$ is entirely analogous). 
	
	Applying the asymptotic scattering relation \eqref{lambdak22}, we obtain with high probability that 
	\begin{flalign}
		\label{lambdav}
		\lambda_k  \approx v_k  + 2T^{-1} \displaystyle\sum_{i = 1}^{N} (\mathbbm{1}_{Q_i (0) < Q_k (0) + T(v_k - v_i)} - \mathbbm{1}_{Q_i (0) < Q_k (0)}) \cdot \log |\lambda_k - \lambda_i|.
	\end{flalign}
	
	\noindent Any eigenvalue $\lambda_i$ should (by \Cref{lleigenvalues}, for example) approximately only depend on the entries of the Lax matrix $\bm{L}$ with indices close to $\varphi_0 (i)$. Since the entries of $\bm{L}$ are independent under thermal equilibrium, this indicates that $\lambda_i$ should be approximately independent from most of the other $\lambda_j$, and from $Q_i(0)$ (where the latter holds since the $a$-entries in $\bm{L}$ depend on only the differences $q_{i+1} - q_i$). Now, let us assume that $v_k \approx v(\lambda_k)$ asymptotically only depends on $\lambda_k$ (that it, it is approximately independent from the other $(\lambda_j)_{j \ne k}$).\footnote{We are unaware of an (even heuristic) explanation for this, and so the proofs in \cite{P} will proceed differently (based on regularizing the indicator functions), instead.}  By the above approximate independence, this suggests that $\mathbbm{1}_{Q_i (0) < Q_k (0) + T (v_k - v_i)} - \mathbbm{1}_{Q_i (0) < Q_k (0)}$ should behave as $\mathbbm{1}_{Q_j (0) < Q_k (0) + T(v_k - v_i)} - \mathbbm{1}_{Q_j (0) < Q_k (0)}$, for a uniformly random index $j \in \llbracket 1, N \rrbracket$. Therefore, the sum on the right side of \eqref{lambdav} should approximate
	\begin{flalign}
		\label{sumq} 
		\displaystyle\sum_{i = 1}^{N} (\mathbbm{1}_{Q_i (0) < Q_k (0) + T(v_k - v_i)} - \mathbbm{1}_{Q_i (0) < Q_k (0)}) \cdot \log |\lambda_k - \lambda_i| \approx \displaystyle\sum_{i \ne k} \mathfrak{N}_i \cdot \log |\lambda_k - \lambda_i|,
	\end{flalign} 
	
	\noindent where 
	\begin{flalign*}
		\mathfrak{N}_i = N^{-1} \cdot \big( \# \{ j : Q_j (0) < Q_k (0) + T(v_k - v_i) \} - \# \{ j : Q_j (0) < Q_k (0)  \} \big). 
	\end{flalign*} 
	
	To approximate $\mathfrak{N}$, recall under thermal equilibrium that the $q_{i+1} - q_i$ are identically distributed under thermal equilibrium (due to \eqref{abr}, since the $a_i$ are, by \Cref{ajbjequation}) and that $\mathbb{E}[q_{i+1}-  q_i] = \alpha$ (by \Cref{aintegral}). Therefore, $\mathfrak{N}_i \approx N^{-1} \cdot \alpha^{-1} \cdot T(v_k  - v_i )$, which upon insertion into \eqref{sumq} yields 
	\begin{flalign}
		\label{sumq2} 
		\begin{aligned}
			\displaystyle\sum_{i = 1}^N (& \mathbbm{1}_{Q_i (0) < Q_k (0)} - \mathbbm{1}_{Q_i (0) < Q_k (0) + T (v_k  - v_i )}) \cdot \log |\lambda_k - \lambda_i| \\
			& \qquad \qquad \qquad \qquad  \approx T (\alpha N)^{-1} \displaystyle\sum_{i \ne k} (v_k - v_i) \cdot \log |\lambda_k - \lambda_i|.
		\end{aligned}
	\end{flalign}
	
	\noindent Now, the $(\lambda_j)$ are eigenvalues of the random Lax matrix $\bm{L}$. The empirical spectral distribution of the latter is known (from \cite[Lemma 4.3]{DSME}) to converge to an explicit probability density $\varrho(x) dx$. Since $v_i = v(\lambda_i)$ is a function of $\lambda_i$, we expect the above sum to likely concentrate, namely,
	\begin{flalign}
		\label{sumlambdav} 
		(\alpha N)^{-1} \displaystyle\sum_{i \ne k} (v_k - v_i) \cdot \log |\lambda_k - \lambda_i| \approx \alpha^{-1} \displaystyle\int_{-\infty}^{\infty} \big( v(\lambda_k) - v (\lambda) \big) \cdot \log |\lambda_k - \lambda| \varrho (\lambda) d \lambda.
	\end{flalign}
	
	Together, the approximations \eqref{lambdav}, \eqref{sumq2}, and \eqref{sumlambdav} yield
	\begin{flalign}
		\label{lambda2} 
		\lambda_k \approx v(\lambda_k)  + 2 \alpha^{-1} \displaystyle\int_{-\infty}^{\infty} \big( v(\lambda_k) - v(\lambda) \big) \cdot \log |\lambda_k - \lambda| \varrho (\lambda) d \lambda.
	\end{flalign}
	
	\noindent Replacing the approximation in \eqref{lambda2} with an equality yields a linear system of equations for $v$ that coincides with \cite[Equation (90)]{GHCS} and \cite[Equation (6.19)]{HSIMS}. It was explained in \cite[Equation (6.21)]{HSIMS} that this system has a specific solution $v(\lambda) = v_{\eff}(\lambda)$, given by \cite[Equation (6.20)]{HSIMS} (and \cite[Equation (93)]{GHCS}); presumably, this solution $v_{\eff}$ is unique. This would then indicate $v_k \approx v_{\eff} (\lambda_k)$.

\end{document}